\documentclass{new_tlp}
\usepackage{times}
\usepackage{helvet}
\usepackage{courier}
\usepackage{graphicx}
\usepackage{enumitem}
\usepackage[cmtip,arrow]{xy}
\usepackage{krudces}
\usepackage{url}

\makeatletter
\renewcommand\paragraph{%
  \@startsection{paragraph}{4}{\z@}
    {-13\p@ \@plus -1.5\p@ \@minus -1.5\p@}
    {-0.5em}
    {\normalfont\normalsize\itshape\raggedright\bfseries}%
}
\makeatother

\title{Deriving Conclusions From Non-Monotonic Cause-Effect Relations}
\author[J. Fandinno]
         {Jorge Fandinno\\
          Department of Computer Science\\
		  University of Corunna, Corunna, Spain\\
		  \email{jorge.fandino@udc.es}
		 }

\def\afine{\ensuremath{fine}}
\def\adead{\ensuremath{dead}}
\def\alongprison{\ensuremath{long\_prison}}
\def\ashortprison{\ensuremath{short\_prison}}

\begin{document}

\newpage
\setcounter{page}{1}
\label{firstpage}
\maketitle








\begin{abstract}
We present an extension of Logic Programming (under stable models semantics) that, not only allows concluding whether a true atom is a cause of another atom, but also \emph{deriving new conclusions} from these causal-effect relations. This is expressive enough to capture informal rules like
``if some agent's actions~$\ag$ have been \emph{necessary} to cause an event~$E$ then conclude atom~$caused(\ag,E)$,''
something that, to the best of our knowledge, had not been formalised in the literature. To this aim, we start from a first attempt that proposed extending the syntax of logic programs with so-called \emph{causal literals}. These causal literals are expressions that can be used in rule bodies and allow inspecting the derivation of some atom~$\rA$ in the program with respect to some query function $\cquery$. Depending on how these query functions are defined, we can model different types of causal relations such as sufficient, necessary or contributory causes, for instance. The initial approach was specifically focused on monotonic query functions. This was enough to cover sufficient cause-effect relations but, unfortunately, necessary and contributory are essentially \mbox{\emph{non-monotonic}}. In this work, we define a semantics for non-monotonic causal literals showing that, not only extends the stable model semantics for normal logic programs, but also preserves many of its usual desirable properties for the extended syntax. Using this new semantics, we provide precise definitions of \emph{necessary} and \emph{contributory} causal relations and briefly explain their behaviour on a pair of typical examples from the Knowledge Representation literature. (Under consideration for publication in Theory and Practice of Logic Programming)
\end{abstract}


\section{Introduction}
\label{sc:introduction}

An important difference between classical models and most Logic Programming (LP) semantics is that, in the latter, true atoms must be founded or justified by a given derivation.
Consequently, falsity is understood as absence of proof: for instance, a common informal way of reading for default literal $\Not \rA$ is
``there is no way to derive $\rA$.''
Although this idea seems quite intuitive and, in fact, several approaches have studied how to syntactically build these derivations or \emph{justifications}~\cite{specht1993generating,pemmasani2004online,pontelli2009justifications,deneckerBS15,schulzT2016justifying},
it actually resorts to a concept, the \emph{ways to derive}~$\rA$, outside the scope of the standard LP semantics.

Such information on justifications for atoms can be of great interest for Knowledge Representation (KR), and especially, for dealing with problems related to causality.
For instance, in the area of legal reasoning where determining a legal responsibility usually involves finding out which agent or agents have eventually caused a given result, regardless the chain of effects involved in the process.
In this sense, an important challenge in causal reasoning is the capability of not only deriving facts of the form
\mbox{``$A$ has caused $B$,''}
but also being able to represent and reason about them. As an example, take the assertion: 
\REVIEW{\ref{R2.0}}
\begin{gather}
\text{``If somebody causes an accident, (s)he would receive a fine''}
	\label{law}
\end{gather}
\ENDREV
\noindent This law does not specify the possible ways in which a person may cause an accident. Depending on a representation of the domain, the chain of events from the agent's action(s) to the final effect may be simple (a direct effect) or involve a complex set of indirect effects and defaults like inertia.
Focussing on representing \eqref{law} in an elaboration tolerant manner~\cite{McC98},
we should be able to write a single rule whose body only refers to the $agent$ involved and the $accident$.
For instance, consider the following program
\begin{IEEEeqnarray}{l C ' l C l}
&& accident &\lparrow& oil
	\label{eq:r1}
\\
&& oil &\lparrow& suzy
	\label{eq:r2}
\\
&& suzy &&
	\label{eq:suzy}
\end{IEEEeqnarray}
representing that $accident$ is an indirect effect of Suzy's actions.
We may then represent~\eqref{law} by the following rule
\begin{IEEEeqnarray}{l C ' l C l}
&& \afine(suzy) &\lparrow&  suzy \necessary accident
	\label{eq:r4}
\end{IEEEeqnarray}
that states that Suzy would receive a $\afine$ whenever the fact $suzy$ was necessary to cause the atom $accident$.

With this long term goal in mind,~\cite{CabalarFF14} proposed a multi-valued semantics for LP that extends the stable model semantics~\cite{GelfondL88} and where justifications are treated as \emph{algebraic} constructions.
In this semantics, \emph{causal stable models} assign, to each atom, one of these algebraic expressions that captures the set of all non-redundant logical proofs for that atom.
Recently, this semantics was used in \cite{fandinno2015aspocp} to extend the syntax of logic programs with a new kind of literal, called~\emph{causal literal}, that allow representing rules like
\begin{IEEEeqnarray}{l C ' l C l}
&& \afine(suzy) &\lparrow&  suzy \sufficient accident
	\label{eq:r3}
\end{IEEEeqnarray}
and derive, from a program~\newprogram\label{prg:suff} containing rules~(\ref{eq:r1}-\ref{eq:suzy},\ref{eq:r3}), that~$\afine(suzy)$~holds.
However, the major limitation of this semantics is that causal literals must be monotonic and, therefore, rule~\eqref{eq:r4} cannot be represented.
It is easy to see that rule~\eqref{eq:r4} is non-monotonic:
in a program~\newprogram\label{prg:necc} containing rules~(\ref{eq:r1}-\ref{eq:r4}),
the fact~$suzy$ is necessary for $accident$ is satisfied and, thus,
$\afine(suzy)$ must hold,
but in a program~\newprogram\label{prg:necc2} \review{\ref{R2.1}}{obtained} by adding a fact $oil$ to this last program, $suzy$ is not longer necessary and, thus,
$\afine(suzy)$ should not be a conclusion.

In this paper,
we present a semantics for logic programs with causal literals defined in terms of \emph{non-monotonic} query functions.
More specifically, we summarise our contributions as follows.
In Section~\ref{sec:causal.semantics}, we define the syntax of causal literals and a multi-valued semantics for logic programs whose causal values rely on a completely distributive lattice based on causal graphs.
Section~\ref{sec:monotonic.programs} shows that positive monotonic program
has a least model that can be computed by an extension of the direct consequences operator~\cite{kowalski76}.
In Section~\ref{sec:csm}, we define semantics for programs with negation and non-monotonic causal literals and show that it is a conservative extension of the standard stable model semantics.
Besides, with a running example, we show how causal literals can be used to derive new conclusion from necessary causal relations and, in Section~\ref{sec:contributory}, briefly relate this notion with the actual cause literature.
In this section, we also formalise the weaker notion of \emph{contributory cause}, also related to the actual cause literature, and show how causal literals may be used to derive new conclusion from them.
In Section~\ref{sec:properties}, we show that our semantics satisfy the usual properties of the stable modles semantics for the new syntax.
Finally, Section~\ref{sec:conclusions} concluded the paper.
The online appendices include the definition of our semantics with nested expression in the body, the formal relation with~\cite{fandinno2015aspocp}, the proof of formal results from the paper and the formalisation of a Splitting Theorem for causal programs analgous to~\cite{lifschitz1994splitting}.

\section{Causal Programs}
\label{sec:causal.semantics}

\REVIEW{\ref{R1.1}}
We start by reviewing some definitions from~\cite{CabalarFF14}.
\ENDREV

\begin{definition}[Term]
Given a set of labels $Lb$, a \emph{term} $t$ is recursively defined as one of the following expressions
\vspace{-9pt}
\begin{gather*}
t \ \ ::= \ \ l \ \ \Big| \ \ \prod S \ \ \Big| \ \ \sum S \ \ \Big| \ \ t_1 \cdot t_2
\end{gather*}
where $l \in Lb$ is a label, $t_1, t_2$ are in their turn terms and $S$ is a (possibly empty and possible infinite) set of terms.\QED
\end{definition}

\noindent
When $S = \set{t_1, \dotsc, t_n}$ is a finite set, we will write $t_1 * \dotsc * t_n$ and $t_1 + \dotsc + t_n$ instead of~$\prod S$ and $\sum S$, respectively.
When $S=\emptyset$, we denote $\prod S$ and $\sum S$ by $1$ and $0$, respectively.
We assume that application~`$\cdot$' has higher priority than product~`$*$' and, in its turn, product~`$*$' has higher priority than addition~`$+$'.
\emph{Application}~`$\cdot$' represents application of a rule label to a previous justifications.
For instance, the justification in program~\programref{prg:suff} for atom $suzy$ is the fact $suzy$ itself.
If rules~(\ref{eq:r1}-\ref{eq:r2}) in program~\programref{prg:suff} are labelled in the following way
\begin{IEEEeqnarray}{l C ' l C l}
r_1 &:& accident &\lparrow& oil
	\label{eq:r1.labelled}
\\
r_2 &:& oil &\lparrow& suzy
	\label{eq:r2.labelled}
\end{IEEEeqnarray}
we may represent the justification of $oil$ as $suzy \cdotl r_2$, \review{\ref{R2.2}}{in other words}, $oil$ is true because of the the application of rule $r_2$ to the fact $suzy$.
Similarly, we may represent the justification of $accident$ as
$suzy \cdotl r_2 \cdotl r_1$.
Addition~`$+$' is used to capture alternative independent causes: each addend is one of those independent causes.
For instance, the justification of $oil$, in program~\programref{prg:necc2}, may be represented as $suzy \cdotl r_2 + oil$ and the justification of $accident$
as $(suzy \cdotl r_2 + oil) \cdot r_1$.
As we will see below application distributes over addition, so that,
the justification of $accident$ can also be written as
$suzy \cdotl r_2 \cdotl r_1 + oil \cdotl r_1$, which better illustrates the existence of two alternatives.
Product~`$*$' represents conjunction or joint causation.
For instance, in a program~\newprogram\label{prg:necc3} obtained by adding the fact $billy$ to~\programref{prg:necc2} and replacing rule~\eqref{eq:r2.labelled} by
\begin{IEEEeqnarray}{l C ' l C l}
r_2 &:& oil &\lparrow& suzy,\ billy
	\label{eq:r2.labelled.billy}
\end{IEEEeqnarray}
the justifications of $oil$ will be
$(suzy * billy) \cdotl r_2 + oil$.
Similarly, the justification of $accident$ will be
$(suzy * billy) \cdotl r_2 \cdotl r_1 + oil \cdotl r_1$.
Intuitively, terms without addition~`$+$' represent individual causes while terms with~ addition~`$+$' represent sets of causes.
It is worth to mention that these algebraic expressions are in a one-to-one correspondence with non-redundant proofs of an atom~\cite{CabalarFF14} and that they may also be understood as a formalisation of Lewis' concept of causal chain~\cite{lewis1973causation} (see \citeNP{fandinno2015aspocp}).

\begin{figure}[htbp]
\begin{center}
\footnotesize
\newcommand{\titleSep}{0pt}
\newcommand{\contentSep}{-10pt}
\newcommand{\rowSep}{5pt}
$
\begin{array}{c}
\hbox{\em Associativity}\vspace{\titleSep}\\
\hline\vspace{\contentSep}\\
\begin{array}{r@{\ }c@{\ }r@{}c@{}l c r@{}c@{}l@{\ }c@{\ }l@{\ }}
t & \cdot & (u & \cdot & w) & = & (t & \cdot & u) & \cdot & w\\
\\
\end{array}
\end{array}
$
\ \ \ \
$
\begin{array}{c}
\hbox{\em Absorption}\vspace{\titleSep}\\
\hline\vspace{\contentSep}\\
\begin{array}{r@{\ }c@{\ }c@{\ }c@{\ }l c r@{\ }c@{\ }r@{\ }c@{\ }c@{\ }c@{\ }c@{\ }l@{\ }}
&& t &&& = & t & + & u & \cdot & t & \cdot & w \\
u & \cdot & t & \cdot & w & = & t & * & u & \cdot & t & \cdot & w
\end{array}
\end{array}
$
\ \ \ \
$
\begin{array}{c}
\hbox{\em Identity}\vspace{\titleSep}\\
\hline\vspace{\contentSep}\\
\begin{array}{rc r@{\ }c@{\ }l@{\ }}
t & = & 1 & \cdot & t\\
t & = & t & \cdot & 1
\end{array}
\end{array}
$
\ \ \ \
$
\begin{array}{c}
\hbox{\em Annihilator}\vspace{\titleSep}\\
\hline\vspace{\contentSep}\\
\begin{array}{rc r@{\ }c@{\ }l@{\ }}
0 & = & t & \cdot & 0\\
0 & = & 0 & \cdot & t\\
\end{array}
\end{array}
$
\\
\vspace{\rowSep}
$
\begin{array}{c}
\hbox{\em Indempotence}\vspace{\titleSep}\\
\hline\vspace{\contentSep}\\
\begin{array}{r@{\ }c@{\ }l@{\ }c@{\ }l }
l & \cdot & l  & = & l\\
\\
\\
\end{array}
\end{array}
$
\hspace{.05cm}
$
\begin{array}{c}
\hbox{\em Addition\ distributivity}\vspace{\titleSep}\\
\hline\vspace{\contentSep}\\
\begin{array}{r@{\ }c@{\ }r@{}c@{}l c r@{}c@{}l@{\ }c@{\ }r@{}c@{}l@{}}
t & \cdot & (u & + & w) & = & (t & \cdot & u) & + & (t & \cdot & w)\\
( t & + & u ) & \cdot & w & = & (t & \cdot & w) & + & (u & \cdot & w)\\ \\
\end{array}
\end{array}
$
\hspace{.05cm}
$
\begin{array}{c}
\hbox{\em Product\ distributivity}\vspace{\titleSep}\\
\hline\vspace{\contentSep}\\
\begin{array}{rcl}
c \cdot d \cdot e & = & (c \cdot d) * (d \cdot e) \ \hbox{with} \ d \neq 1 \\
c \cdot (d*e)     & = & (c \cdot d) * (c \cdot e) \\
(c*d) \cdot e     & = & (c \cdot e) * (d \cdot e)
\end{array}
\end{array}
$
\end{center}
\vspace{-5pt}
\caption{Properties of the  `$\cdot$'operators:
$t,u,w$ are terms, $l$ is a label and $c,d,e$ are terms without~`$+$'.
Addition and product distributivity are also satisfied over infinite sums and products.}
\label{fig:appl}
\end{figure}
\begin{definition}[Value]
\label{def:values}
\emph{(Causal) values} are the equivalence classes of terms under axioms for a completely distributive (complete) lattice with meet~`$*$' and join~`$+$' plus the axioms of Figure~\ref{fig:appl}.
The set of values is denoted by~$\values$.
Furthermore, by $\causes$ we denote the subset of causal values with some representative term without sums~`$+$'.\QED
\end{definition}

All three operations, `$*$', `$+$' and `$\cdot$' are associative. Product~`$*$' and addition `$+$' are also commutative, and they satisfy the usual absorption and distributive laws with respect to infinite sums and products of a completely distributive lattice. 
The lattice order relation is defined as:
\begin{IEEEeqnarray*}{c"C"c"C"c}
t \leq u & \text{ iff } & t * u = t & \text{ iff } & t + u = u
\end{IEEEeqnarray*}
An immediately consequence of this definition is that product, addition, $1$ and $0$ respectively are the greatest lower bound, the least upper bound and the top and the bottom element of the \mbox{$\leq$-relation}.
Term~$1$ represents a value which holds by default, \review{\ref{R2.3}}{without an explicit} cause, and will be assigned to the empty body.
Term~$0$ represents the absence of cause or the empty set of causes, and will be assigned to false.
Furthermore, applying distributivity (and absorption) of \review{R2.4}{product and application over addition}, every term can be represented in \emph{(minimal) disjunctive normal form} in which \review{R2.4}{addition} is not in the scope of any other operation and every pair of addends are pairwise $\leq$-incomparable.
In the following, we will assume that every term is in disjunctive normal form.

\REVIEW{\ref{R1.1}}
This semantics was used in~\cite{fandinno2015aspocp},
to define the concept of causal query, here \emph{\mbox{m-query}}: a monotonic function \mbox{$\phi: \causes \longrightarrow \set{0,1}$}.
Unfortunately, m-queries are not expressive enough to capture necessary causation for two reasons: $(i)$ they are monotonic and $(ii)$ they cannot capture relations between sets of causes.
We introduced here the following definition which removes these two limitations.
\ENDREV

\begin{definition}[Causal query]
\label{def:causal.query}
A \emph{causal query} \, $\cquery\!: \causes \times \values \, \longrightarrow \, \set{0,1}$ \, is a function mapping pairs cause-value into $1$ (true) and  $0$ (false) which is 
anti-monotonic in the second argument, that is,
$\cquery(G,t) \leq \cquery(G,u)$ for every
\mbox{$G \in \causes$} and $\set{t,u} \subseteq \values$ such that
$t \geq u$.\qed
\end{definition}

\paragraph{Syntax.}
We define the semantics of logic programs using its grounding.
Therefore, for the remainder of this paper, we restrict
our attention to ground logic programs.
A \emph{signature} is a triple
\signature\ where $\at$, $\lb$ and $\cqueries$ respectively represent sets of 
\emph{atoms} (or \emph{propositions}), \emph{labels}
and causal queries.
We assume the signature of every program contains a causal query
$\cqueryone \in \cqueries$ s.t. \mbox{$\cqueryone(G,t) \eqdef 1$} for every~$G \in \causes$ and value $t \in \values$.

\begin{definition}[Causal literal]
\label{def:causal.literal}
A \emph{(causal) literal} is an expression 
$(\cliteral{\cquery}{\rA})$ where $\rA \in \at$ is an atom and $\cquery \in \cqueries$ is a causal query.\QED
\end{definition}

A causal atom $(\cliteral{\cquery^1}{\rA})$ is said to be \emph{regular} and, by abuse of notation, we will use atom~$\rA$ as shorthand for regular causal literals of the form~$(\cliteral{\cqueryone}{\rA})$.
We will see below the justification for this notation.
A \emph{literal} is either a causal literal~$(\cliteral{\cquery}{\rA})$ \ (\emph{positive literal}),
or a negated causal literal~\mbox{$\Not (\cliteral{\cquery}{\rA})$} \ (\emph{negative literal})
or a double negated causal literal~~$\Not\Not (\cliteral{\cquery}{\rA})$ \
(\emph{consistent literal}) with $\rA \in \at$ an atom and $\cquery \in \cqueries$ a causal query.

\begin{definition}[Causal program]\label{def:causal.P}
A \emph{(causal) program} $P$ is a set of rules of the form:
\begin{IEEEeqnarray}{l C ' l C l,C,l}
r_i &:& \rH &\lparrow& \rB_1, \dotsc, \rB_\npbody
    \label{eq:rule} 
\end{IEEEeqnarray}
\noindent where $0 \leq m$ is a non-negative integer, \mbox{$r_i\in Lb$} is a label or $r_i=1$, $\rH$ (the \emph{head} of the rule) is an atom and each $\rB_i$ with $1 \le i \leq \npbody$ (the \emph{body} of the rule) is a literal or a term.\QED
\end{definition}

A rule $\R$ ºis said to be \emph{positive} iff all literals in its body are positive  and it is said to be \emph{regular} if all causal literals in its body are regular.
When $\npbody = 0$, we say that the rule is a \emph{fact} and omit the body and sometimes the symbol~~`$\leftarrow$.'
Furthermore, for clarity sake, 
we also assume that, for every atom $\rA \in \at$, there is an homonymous label $\rA \in \lb$
and that the label of an unlabelled rule is assumed to be its head.
In this sense, a fact $\rA$ in a program actually
stands for the labelled rule $(\rA : \rA \leftarrow)$.
A program $P$ is \emph{positive} or \emph{regular}
when all its rules are positive (i.e. it contains no default negation) or regular,
respectively.
A \emph{standard program} is a regular program in which the label of every rule is `$1:$'.

\paragraph{Semantics.}
A \emph{(causal) interpretation} is a mapping \mbox{$\cI:At\longrightarrow\values$} assigning a value to each atom. For interpretations $\cI$ and $\cJ$, we write $\cI\leq \cJ$ when
\mbox{$\cI(\rH) \leq \cJ(\rH)$} for every atom $\rA \in At$.
Hence, there is a \mbox{$\leq$-bottom} interpretation \botI\ (resp. a $\leq$-top interpretation~\topI) that stands for the interpretation mapping every atom $\rA$ to $0$ (resp. $1$).
For an interpretation $I$ and atom $\rA \in \at$, by $\max I(\rA)$ we denote the set
\begin{gather*}
\max I(\rA) \ \ \eqdef \ \
\setbm{ G \in \causes }{  G \leq I(\rA) \text{ and there is no } G' \in \causes \text{ s.t. } G < G' \leq I(\rA) }  
\end{gather*}
containing the maximal terms without \review{R2.4}{addition} (or individual causes) of $\rA$ w.r.t. $I$.

\begin{definition}[Causal literal valuation]
\label{def:causal.literal.evaluation}
The \emph{valuation of a causal literal} of the form $(\cliteral{\cquery}{\rA})$ with respect to an interpretation $I$, in symbols $I(\cliteral{\cquery}{\rA})$, is given by
\vspace{-0.01cm}
\begin{align*}
I(\cliteral{\cquery}{\rA}) \ \ &\eqdef \ \
  \sum\setbm{G \leqmax I(\rA)}{ \cquery(G,\, I(\rA)\,) = 1  }
\end{align*}
We say that $I$ satisfies a causal literal $(\cliteral{\cquery}{\rA})$, in symbols
$I \models (\cliteral{\cquery}{\rA})$, iff $I(\cliteral{\cquery}{\rA}) \neq 0$.\qed \end{definition}

Notice now that 
\mbox{$I(\cliteral{\cqueryone}{\rA}) = I(\rA)$} for any atom $\rA$ and, thus, writing a standard atom~$\rA$ as a shorthand for causal literal~$(\cliteral{\cqueryone}{\rA})$ does not modify its intended meaning.
Causal literals can be used to represent the body of rule~\eqref{eq:r4}.
For instance, given a set of labels $\ag \subseteq\lb$ representing the actions of some agent $\ag$, we may define the query function 
\begin{gather}
\cquerynec_{\ag}(G,t) \ \ \eqdef \ \
  \begin{cases}
  1   &\text{if } \ \
        t \ \leq \ \sum \ag
      \\
  0   &\text{otherwise}
  \end{cases}
  \label{eq:nec.agent.definition}
\end{gather}
and represent the body of rule~\eqref{eq:r4} by a causal literal of the form~$(\cliteral{\cquerynec_{Suzy}}{accident})$
where $Suzy$ is the set of labels~$\set{suzy}$.
In the sake of clarity, we usually will write $(\ag \necessary \rA)$ in rule bodies instead $(\cliteral{\cquerynec_{\ag}}{\rA})$.

If we consider an interpretation $I$ which assigns to the atom $accident$ its justification in program~\programref{prg:necc}, that is, $I(accident) = suzy \cdotl r_2 \cdotl r_1$,
then any term without addition~\mbox{$G \in \causes$},
satisfies
\begin{align*}
\cquerynec_{Suzy}(G,I)(accident) = 1
&\hspace{0.25cm}\text{ iff }\hspace{0.25cm}
suzy \cdotl r_2 \cdotl r_1 \ \ \leq \ \ \sum \set{suzy}
\\&\hspace{0.25cm}\text{ iff }\hspace{0.25cm}
suzy \cdotl r_2 \cdotl r_1 \ \ \leq \ \ suzy
\\&\hspace{0.25cm}\text{ iff }\hspace{0.25cm}
suzy \cdotl r_2 \cdotl r_1 + suzy \ \ = \ \ suzy
\end{align*}
which holds applying application identity, associativity and absorption w.r.t. \review{R2.4}{addition}
\begin{align*}
suzy \cdotl r_2 \cdotl r_1 + suzy 
	\ \ = \ \ 1 \cdot suzy \cdot (r_2 \cdotl r_1) + suzy
	\ \ = \ \ suzy
\end{align*}
Similarly, in program~\programref{prg:necc2},
$\cquerynec_{Suzy}(G,I'(accident)) =  1$
iff
$suzy \cdotl r_2 \cdotl r_1 + oil \leq suzy$
which does not hold.
In other words, Suzy's actions has been necessary in program~\programref{prg:necc} but not in program~\programref{prg:necc2}.


The valuation of a causal term $t$ is the class of equivalence of $t$.
The valuation of
non-positive literals is defined as follows
\begin{IEEEeqnarray*}{l C l}
I(\Not (\cliteral{\cquery}{\rA})) \ \ &\eqdef& \ \ \begin{cases}
1 &\text{iff } I(\cliteral{\cquery}{\rA}) \ = \ 0\\
0 &\text{otherwise}
\end{cases}
\\
I(\Not\Not (\cliteral{\cquery}{\rA})) \ \ &\eqdef& \ \ \begin{cases}
1 &\text{iff } I(\cliteral{\cquery}{\rA}) \ \neq \ 0\\
0 &\text{otherwise}
\end{cases}
\end{IEEEeqnarray*}
Furthermore, for any literal or term~$\rL$, we write $I \models \rL$ iff $I(\rL) \neq 0$.

\begin{definition}[Causal model]\label{def:causal.model}
Given a rule $\R$ of the form~\eqref{eq:rule},
we say that an interpretation $I$ \emph{satisfies} $\R$, in symbols $I \models \R$,
if and only if the following condition holds:
\begin{gather}
\big( \, I(\rB_1) * \dotsc * I(\rB_\npbody) \, \big)  \cdot r_i \ \ \leq \ \ I(\rH)
  \label{eq:causa.stable.model}
\end{gather}
An interpretation $I$ is a \emph{causal model} of $P$, in symbols $I \models P$, iff $I$~satisfies all rules in~$P$.\QED
\end{definition}

Let \newprogram\label{prg:accident} be the program containing rules~\eqref{eq:r1.labelled} and~\eqref{eq:r2.labelled} plus the labelled fact $(suzy : suzy \leftarrow)$
and~\newprogram\label{prg:accident.billy} be the program containing rules~\eqref{eq:r1.labelled} and~\eqref{eq:r2.labelled.billy} plus the labelled facts $(suzy : suzy \leftarrow)$ and $(billy : billy \leftarrow)$.
Then, it can be checked that these programs respectively have unique $\leq$-minimal models $I_{\ref{prg:accident}}$ and $I_{\ref{prg:accident.billy}}$ which satisfy
\begin{IEEEeqnarray*}{l ; C ; c ? C ? c ; C ; c}
I_{\ref{prg:accident}}(accident)
	&=& suzy \cdotl r_2 \cdotl r_1
&\hspace{1cm}&
I_{\ref{prg:accident.billy}}(accident)
	&=& (suzy * billy) \cdotl r_2 \cdotl r_1 + oil
\end{IEEEeqnarray*}
Let now \newprogram\label{prg:necc.labelled} and~\newprogram\label{prg:necc3.labelled} be the labelled programs respectively obtained by adding the following rule
\begin{IEEEeqnarray}{l C ' l C l}
r_3 &:& \afine(suzy) &\lparrow&  suzy \necessary accident
	\label{eq:r4.labelled}
\end{IEEEeqnarray}
(resulting of labelling rule~\eqref{eq:r4} with $r_3$) to programs~\programref{prg:accident} and~\programref{prg:accident.billy}.
Then it can be checked that these programs also have unique $\leq$-minimal models $I_{\ref{prg:necc.labelled}}$ and $I_{\ref{prg:necc3.labelled}}$ which respectively agree with $I_{\ref{prg:accident}}$ and $I_{\ref{prg:accident.billy}}$ in all atoms but in $\afine(suzy)$ and, as we have seen above,
\begin{IEEEeqnarray*}{l , C ,  l  , C , c C l , C , l}
I_{\ref{prg:necc.labelled}}(\cliteral{\cquerynec_{Suzy}}{accident})
	&=& I_{\ref{prg:necc.labelled}}(accident)
	&=& suzy \cdotl r_2 \cdotl r_1
&\hspace{1cm}&
I_{\ref{prg:necc3.labelled}}(\cliteral{\cquerynec_{Suzy}}{accident})
	&=& 0
\end{IEEEeqnarray*}
Furthermore, by definition, it holds that
$I_j(\afine(suzy))
	\ = \ I_j(\cliteral{\cquerynec_{Suzy}}{accident}) \cdotl r_3$
for $j\in\set{\ref{prg:necc.labelled}, \ref{prg:necc3.labelled}}$
which implies that
\begin{IEEEeqnarray*}{l ? C ? l}
I_{\ref{prg:necc.labelled}}(\hspace{-1pt}\afine(suzy))\hspace{-1pt})\hspace{-1pt}
	&=& \hspace{-1pt} suzy \cdotl r_2 \cdotl r_3
\\
I_{\ref{prg:necc3.labelled}}(\afine(suzy))\hspace{-1pt})
	&=& 0 \cdotl r_3
	= 0
\end{IEEEeqnarray*}
That is, Suzy would receive a fine for causing the accident,
$I_{\ref{prg:necc.labelled}} \models \afine(suzy)$, w.r.t~$\programref{prg:necc.labelled}$, but not w.r.t. program~$\programref{prg:necc3.labelled}$ because $I_{\ref{prg:necc3.labelled}} \not\models \afine(suzy)$.

It is worth to note that positive programs may contain non-monotonic causal literals that, somehow, play the role of negation and, hence, they may have several $\leq$-minimal causal models.
Consider, for instance, the following positive program~\newprogram\label{prg:positive.two-minimal-models}
\begin{gather*}
\begin{IEEEeqnarraybox}[][t]{lC ' l C l , lrl}
r_1 &:& p
\end{IEEEeqnarraybox}
\hspace{2cm}
\begin{IEEEeqnarraybox}[][t]{lC ' l C l , lrl}
r_2 &:& q  &\lparrow& \ag_1 \necessary p
\end{IEEEeqnarraybox}
\end{gather*}
where $\ag_1 \eqdef \set{r_1}$.
Program~$\programref{prg:positive.two-minimal-models}$
has two $\leq$-minimal causal models.
The first one which satisfies
\mbox{$I_{\ref{prg:positive.two-minimal-models}}(p) = r_1$}
and
$I_{\ref{prg:positive.two-minimal-models}}(q) = r_1\cdotl r_2$;
and a second unintended one which satisfies
$I'_{\ref{prg:positive.two-minimal-models}}(p) = r_1 + r_2$
and
$I'_{\ref{prg:positive.two-minimal-models}}(q) = 0$.
In the following section, we introduce the notion of \emph{monotonic programs}
which have a least model and a well-behaved direct consequences operator (when they are positive).
In Section~\ref{sec:csm}, we will see that, in fact, only $I_{\ref{prg:positive.two-minimal-models}}$ is a causal stable model of program~$\programref{prg:positive.two-minimal-models}$.

\section{Positive monotonic Programs}
\label{sec:monotonic.programs}

A causal query $\cquery$ is said to be \emph{monotonic} iff
$\cquery(G,u) \leq \cquery(G',w)$ for any values~$\set{G,G'} \subseteq \causes$ and $\set{u,w}\subseteq\values$ such that $G \leq G'$.
A causal literal $(\cliteral{\cquery}{\rA})$ is \emph{monotonic} if $\cquery$ is monotonic.
A program $P$ is \emph{monotonic} iff $P$ all causal literals occurring in $P$ are monotonic.
\REVIEW{R1.1}
We show next that every monotonic program can be reduced to the syntax and semantics of~\cite{fandinno2015aspocp}. 
For space reasons, we omit here the details of~\cite{fandinno2015aspocp}, which can be found in Appendix~\ref{sec:m-programs}.

\ENDREV

\begin{definition}
Given a query $\cquery$ (resp. m-query $\phi$), its \emph{corresponding m-query (resp. query)} is given by
$\phi_\cquery(G) \eqdef \cquery(G,1)$
(resp. $\cquery_\phi(G,t) \eqdef \phi(G)$).
Similarly, for any program $P$ (resp. m-program $Q$) its \emph{corresponding m-program $Q$ (resp. program $P$)} is obtained by replacing every query $\cquery$ in $P$ (resp. m-query $\phi$ in $Q$) by its corresponding m-query $\phi_\cquery$ (resp.query $\cquery_\phi$).\qed
\end{definition}

\begin{Theorem}{\label{thm:fandino2015.monotonic}}
If $P$ is the corresponding program of some positive m-program~$Q$ with the syntax of Definition~\ref{def:causal.P} or $Q$ is the corresponding m-program of some positive monotonic  program $P$, then an interpretation~$I$ is a model of~$P$ iff $I$ is a model of $Q$.\qed
\end{Theorem}

An immediate consequence of Theorem~\ref{thm:fandino2015.monotonic}, plus
Theorem~3.8 in~\cite{fandinno2015aspocp}, is that positive monotonic programs have a least model that can be computed by iteration of the following extension of the direct consequences operator of~\citeN{kowalski76}.

\begin{definition}[Direct consequences]\label{def:tp}
Given a causal program $P$, the operator of \emph{direct consequences} is a function $\tp$ from interpretations to interpretations such that
\begin{align*}
\tp(I)(\rA)
    \ \ &\eqdef \ \
    \sum \setbm{ \ \big( \, I(\rB_1) * \dotsc * I(\rB_\npbody) \, \big)  \cdot r_1 }
    		{ (r_i : \ \rA \lparrow \rB_1, \dotsc, \rB_\npbody) \in P \ }
\end{align*}
for any interpretation $I$ and any atom $\rA \in \at$.
The iterative procedure is defined as usual
\begin{align*}
\tpr{\alpha} &\ \ \eqdef \ \ \tp(\tpr{\alpha-1})
                &&\text{if } \alpha \text{ is a successor ordinal}
\\
\tpr{\alpha} &\ \ \eqdef \ \ \sum_{\beta < \alpha}\tpr{\beta}
                &&\text{if } \alpha \text{ is a limit ordinal}
\end{align*}
As usual $0$ and $\omega$ respectively denote the first limit ordinal and the first limit ordinal that is greater than all integers. Thus, $\tpr{0} = \botI$.\qed
\end{definition}

\begin{Corollary}{\label{thm:tp.properties}}
Any (possibly infinite) positive monotonic program $P$ has a least causal model $I$ which
 (i) coincides with the least fixpoint $\lfp(\tp)$ of the direct consequences operator $\tp$ and (ii) can be iteratively computed from the bottom interpretation $I = \lfp(\tp) = \tpr{\omega}$.\qed
\end{Corollary}

Corollary~\ref{thm:tp.properties} guarantees that the least fixpoint of $T_P$ is well-behaved and corresponds to the least model of the program~$P$.
In fact, we can check now that the least model $I_{\ref{prg:accident.billy}}$ of program~\programref{prg:accident.billy} satisfies
$I_{\ref{prg:accident.billy}}(accident) = (suzy * billy) \cdotl r_2 \cdotl r_1 + oil \cdotl r_1$.
First note, that program~\programref{prg:accident.billy} contains facts $suzy$, $billy$ and $oil$ whose label is the same as the name atom and, thus,
$\tprP{\programref{prg:accident.billy}}{1}(\rA) = \rA$ for each atom $\rA \in \set{ suzy,\ billy,\ oil }$.
Then, since
$\tprP{\programref{prg:accident.billy}}{1}(suzy) = suzy$,
$\tprP{\programref{prg:accident.billy}}{1}(billy) = billy$
and 
rule~\eqref{eq:r2.labelled}
and fact $oil$ belong to program~\programref{prg:accident.billy},
it follows that
$\tprP{\programref{prg:accident.billy}}{2}(oil) = (suzy * billy) \cdot r_2 + oil$.
Similarly, we can check that
\begin{gather*}
\tprP{\programref{prg:accident.billy}}{3}(accident) 
	\ \ = \ \ (\, (suzy * billy) \cdot r_2 + oil ) \cdot r_1
	\ \ = \ \  (suzy * billy) \cdotl r_2 \cdotl r_1 + oil \cdotl r_1 
\end{gather*}
and, thus, $I_{\ref{prg:accident.billy}} = \tprP{\programref{prg:accident.billy}}{3}$ is the least fixpoint of $\tpP{\programref{prg:accident.billy}}$.
Checking that
$\tprP{\programref{prg:accident}}{3} = I_{\ref{prg:accident}}$,
that
$\tprP{\programref{prg:necc.labelled}}{4} = I_{\ref{prg:necc.labelled}}$ 
and that
$\tprP{\programref{prg:necc3.labelled}}{4} = I_{\ref{prg:necc3.labelled}}$ are the least fixpoint and the least models respectively of programs~\programref{prg:accident},~\programref{prg:necc.labelled} and~\programref{prg:necc3.labelled} is analogous.

It is easy to see that every true atom, according to the standard least model semantics, has a non-zero causal value associated in the causal least model of the program, that is, some associated cause.
An interpretation $I$ is \emph{two-valued} when it maps each atom into the set~$\set{0,1}$.
By~$I^{cl}$, we denote the two-valued (or ``classic'') interpretation corresponding to some interpretation~$I$ s.t.
\begin{gather*}
I^{cl}(\rH) \ \ \eqdef \ \ \begin{cases}
        1 &\text{iff } I(\rH) > 0\\
        0 &\text{iff } I(\rH) = 0
\end{cases}
\end{gather*}

\begin{Corollary}{\label{thm:positive.two-valued.model.correspondence}}
Let $P$ be a regular, positive monotonic program and $Q$ its standard unlabelled version obtained by removing all labels from the rules in $P$.
Let $\cI$ and $\cJ$ be the least models of $P$ and $Q$, respectively.
Then, $\cI^{cl} = \cJ$.\QED
\end{Corollary}

\section{Non-monotonic causal queries and negation}
\label{sec:csm}

\REVIEW{R1.1}
We introduce now the semantics for programs with non-monotonic causal queries and negation by extending the concept of reduct~\cite{GelfondL88} to causal queries.
\ENDREV

\begin{definition}[Reduct]\label{def:reduct}
For any term $t$, by
$\cquery^{t}$ we denote a query such that
\begin{gather*}
\cquery^t(G,u) \ \ \eqdef \ \
  \begin{cases}
  1 &\text{iff exists some }  G' \leq G \text{ s.t. } G' \leqmax t
      \text{ and }
      \cquery(G',\,t) = 1
  \\
  0 &\text{otherwise}
  \end{cases}
\end{gather*}



\noindent
The \emph{reduct} of a causal literal~$(\cliteral{\cquery}{\rA})$ w.r.t some interpretation $I$ is itself if $\cquery$ is monotonic and $(\cliteral{\cquery^{I(\rA)}}{\rA})$ if $\cquery$ is non-monotonic.
The reduct of a program $P$ w.r.t. an interpretation $I$, in symbols $P^I$, is the result of (i) removing all rules whose body contains a non satisfied negative or consistent literal,
(ii) removing all the negative and consistent literals for the remaining rules and
(iii) replacing the remaining causal literals~$(\cliteral{\cquery}{\rA})$ by their reducts $(\cliteral{\cquery}{\rA})^I$.\qed
\end{definition}

It is easy to see that the reduct~$P^I$ of any program~$P$ is a positive monotonic program and, therefore, it has a least causal model.

\begin{definition}[Causal stable model]\label{def:causal.smodel}
We say that an interpretation $I$ is a \emph{causal stable model} of a program~$P$ iff $I$~is the least model of the positive program~$P^I$.\QED
\end{definition}

We can check now that
interpretation
$I_{\ref{prg:positive.two-minimal-models}}$ is, in fact, the unique causal stable model of program~$\programref{prg:positive.two-minimal-models}$.
Let
$Q = \programref{prg:positive.two-minimal-models}^{I_{\ref{prg:positive.two-minimal-models}}}$
be the reduct of program~$\programref{prg:positive.two-minimal-models}$ w.r.t.
$I_{\ref{prg:positive.two-minimal-models}}$ consisting in the following rules
\begin{gather*}
\begin{IEEEeqnarraybox}[][t]{lC ' l C l , lrl}
r_1 &:& p
\end{IEEEeqnarraybox}
\hspace{2cm}
\begin{IEEEeqnarraybox}[][t]{lC ' l C l , lrl}
r_2 &:& q  &\lparrow& (\cliteral{\cquery}{p})
\end{IEEEeqnarraybox}
\end{gather*}
where
$\cquery(G,t) = 1$ iff there exists some $G' \leq G$ s.t.
$G' \leqmax I_{\ref{prg:positive.two-minimal-models}}(p) = r_1$
and $\cquerynec_{\ag_1}(G',I_{\ref{prg:positive.two-minimal-models}}(p))$
iff $r_1 \leq G$ and $r_1 \leq \sum\ag_1 = r_1$
iff $r_1 \leq G$.
First note that $\tprP{Q}{\alpha}(p)  =  r_1 = I_{\ref{prg:positive.two-minimal-models}}(p)$ for any ordinal~$\alpha\geq 1$
because $r_1$ is the only rule with the atom~$p$ in the head.
Then, note that $\tprP{Q}{\alpha}( \cliteral{\cquery}{p} ) = \tprP{Q}{\alpha}(p)$ because
$r_1 \leq G$ for every $G \leqmax \tprP{Q}{\alpha}(p) = r_1$ (there is only one such $G = r_1$)
and, thus,
\begin{gather*}
\tprP{Q}{\beta}(q) 	\ \ = \ \ \tprP{Q}{\alpha}(\cliteral{\cquery}{p}) \cdotl r_2
					\ \ = \ \ \tprP{Q}{\alpha}(p) \cdotl r_2
					\ \ = \ \ r_1 \cdotl r_2
					\ \ = \ \ I_{\ref{prg:positive.two-minimal-models}}(q)
\end{gather*}
for any ordinal~$\beta \geq 2$.
Hence, $I_{\ref{prg:positive.two-minimal-models}}$ is a causal stable model of~$\programref{prg:positive.two-minimal-models}$.
On the other hand, we can check that
$I_{\ref{prg:positive.two-minimal-models}}'$ is not a causal stable model of~$\programref{prg:positive.two-minimal-models}$.
Let $Q' = \programref{prg:positive.two-minimal-models}^{I_{\ref{prg:positive.two-minimal-models}}'}$ be the reduct of program~$\programref{prg:positive.two-minimal-models}$ w.r.t. $I_{\ref{prg:positive.two-minimal-models}}'$ consisting in the same rules than program $Q$, but replacing $\cquery$ by $\cquery'$ where $\cquery'(G,t) = 1$ iff there exists some $G' \leq G$ s.t.
$G' \leqmax I_{\ref{prg:positive.two-minimal-models}}'(p) = r_1 + r_2$
and $\cquerynec_{\ag_1}(G',I_{\ref{prg:positive.two-minimal-models}}'(p))$.
As above, $\tprP{Q'}{\alpha}(p)  =  r_1 \neq I_{\ref{prg:positive.two-minimal-models}}'(p) = r_1 + r_2$ for any ordinal~$\alpha\geq 1$ and, therefore, $I_{\ref{prg:positive.two-minimal-models}}$ is not a causal stable model of program~$\programref{prg:positive.two-minimal-models}$.


It is worth to mention that, as happened with positive programs, we can stablish a correspondence between the causal stable models of regular programs and the standard stable models of their standard version.

\begin{definition}[Two-valued equivalence]
Two programs $P$ and~$Q$ are said to be
\emph{two-valued equivalent}
iff for every causal stable model $I$ of $P$ there is an unique
causal stable model $J$ of $Q$ such that $I^{cl}=J^{cl}$, and vice-versa.\QED
\end{definition}

\begin{Theorem}{\label{thm:regluar.standard.correspondnece}}
Let $P$ be a regular program and $Q$ be its corresponding standard program obtained by removing all labels in $P$.
Then $P$ and $Q$ are two-valued equivalent.\QED
\end{Theorem}

Theorem~\ref{thm:regluar.standard.correspondnece} asserts that, labelling a standard program does not change which atoms are true or false in its stable models, in other words, the causal stable semantics presented here is a conservative extension of the standard stable model semantic.


\section{Contributory cause and its relation with actual causation}
\label{sec:contributory}

Until now we have considered that an agent is a cause of an event when its actions have been necessary to cause that event.
This understanding is similar to the definition of the \emph{modified Halpern-Pearl definition of causality} given by~\citeN{halpern2015modification}.
However, in some scenarios it makes sense to consider a weaker definition in which those agents whose actions have \emph{contributed} to that event
are also considered causes, even if their actions have not been necessary~\cite{Pearl00}.
Consider, for instance, the following example from~\cite{hopkins2003clarifying}.

\begin{example}\label{ex:load.firing.squad}
For a firing squad consisting of shooters
Billy and Suzy, it is John's job to load Suzy's gun.
Billy loads and fires his own gun.
On a given day, John loads Suzy's gun.
When the time comes, Suzy and Billy shoot the prisoner.
The agents who caused the prisoner death would be punished with imprisonment.\qed
\end{example}

In this example,
although the actions of any of the agents are not necessary for the prisoner's death,
commonsense tells that all three  should be considered responsible of it.
If we represent Example~\ref{ex:load.firing.squad} by the following program~\newprogram\label{prg:load.firing.squad}
\begin{gather*}
\begin{IEEEeqnarraybox}[][t]{l C ; l C l,C,l}
r_1   &:& dead     &\lparrow& shoot(suzy), loaded\\
r_2   &:& dead     &\lparrow& shoot(billy)\\
r_3   &:& loaded     			&\lparrow& load(john)\\
r_{\ag}   &:& \alongprison(\ag)  &\lparrow& \ag \necessary dead
\end{IEEEeqnarraybox}
\hspace{1.25cm}
\begin{IEEEeqnarraybox}[][t]{l C ; l C l,C,l}
&& shoot(suzy)\\
&& shoot(billy)\\
&& load(john)
\end{IEEEeqnarraybox}
\end{gather*}
for $\ag \in \set{suzy,\ billy,\ john}$,
it can be shown that its unique causal stable model $I_{\ref{prg:load.firing.squad}}$ satisfies
\begin{gather*}
I_{\ref{prg:load.firing.squad}}(dead)
		\ \ = \ \ \big( load(john) \cdotl r_3 * shoot(suzy) \big) \cdot r_1
		\ \ + \ \ shoot(billy) \cdotl r_2
\end{gather*}
Recall that, we assume that every fact has a label with the same name.
According to~$I_{\ref{prg:load.firing.squad}}$,
the actions of the three agents appear in the causes of the atom~$dead$,
but there is no agent whose actions occur in all causes.
Then, the causal literal $(\ag \necessary dead)$ is not satisfied for any agent $\ag$ and, therefore, it holds that
$I_{\ref{prg:load.firing.squad}}(\alongprison(\ag)) = 0$
for every agent~$\ag \in \set{suzy,\ billy,\ john}$.
That is, no agent is punished with imprisonment for the prisoner's death.
On the other hand, if~\newprogram\label{prg:load.firing.squad.contributed} is a program obtained by replacing rules $r_\ag$ by rules
\begin{gather*}
\begin{IEEEeqnarraybox}[][t]{l C ; l C l,C,l}
c_{\ag}   &:& \ashortprison(\ag,\adead)  &\lparrow& \ag \contributed dead
\end{IEEEeqnarraybox}
\end{gather*}
in program~\programref{prg:load.firing.squad},
we may expect that $\ashortprison(\ag)$ holds, in its unique causal stable model~$I_{\ref{prg:load.firing.squad.contributed}}$,
for any $\ag \in \set{suzy,\ billy,\ john}$.
We formalise this by defining the following query
\begin{gather}
\cquerycont_{\ag}(G,t) \ \ \eqdef \ \
  \begin{cases}
  1   &\text{if } G \ \leq \ \sum \ag
      \\
  0   &\text{otherwise}
  \end{cases}
  \label{eq:cont.agent.definition}
\end{gather}
In the sake of clarity,
we will write
$(\ag \contributed dead)$
instead of
$(\cliteral{\cquerycont_{\ag}}{dead})$.
It can be checked that
$\big( load(john) \cdotl r_3 * shoot(suzy) \big) \cdot r_1 \ \leq \ load(john)$ and, therefore,
\begin{gather*}
I_{\ref{prg:load.firing.squad.contributed}}(john \contributed dead)
	\ \ = \ \ \big( load(john) \cdotl r_3 * shoot(suzy) \big) \cdot r_1
\end{gather*}
Consequently,
$I_{\ref{prg:load.firing.squad.contributed}}(\ashortprison(john))
	\ = \ \big( load(john) \cdotl r_3 * shoot(suzy) \big) \cdot r_1 \cdot c_{john}$.
Similarly, it can be shown that
\begin{IEEEeqnarray*}{l c l}
I_{\ref{prg:load.firing.squad.contributed}}(\ashortprison(suzy))
	&\ \ = \ \ &\big( load(john) \cdotl r_3 * shoot(suzy) \big) \cdot r_1 \cdotl c_{suzy}
\\
I_{\ref{prg:load.firing.squad.contributed}}(\ashortprison(billy))
	&\ \ = \ \ & shoot(billy) \cdotl r_2 \cdot c_{billy}
\end{IEEEeqnarray*}

It is worth to note that contributory causes are non-monotonic when defaults are taken into account.
Consider now the following variation of Example~\ref{ex:load.firing.squad}.

\begin{example}\label{ex:load.firing.squad.default}
Now Suzy also loads her gun as Billy does.
However, Suzy's gun was broken and John repaired it.\qed
\end{example}

As in Example~\ref{ex:load.firing.squad}, John's repairing action is necessary in order for Suzy to be able to fire her gun.
However, in this case, it seems too severe to consider that John has contributed \review{R2.8}{to} the prisoner's death.
This consideration has been widely attributed to the fact that we consider that, by default, things are not broken and that causes must be events that deviate from the norm~\cite{maudlin2004,hall2007structural,Halpern08,hitchcock2009cause}.
If we represent this variation by a program~\newprogram\label{prg:broken.firing.squad} containing the following rules%
\footnote{We have chosen this representation in order to illustrate the non-monotonicity of contributory cause.
However, solving the Frame and Qualification Problems~\cite{McCarthy1969,mccarthy1987epistemological} would require the introduction of time and the inertia laws, plus the replacement of rule $r_1$ by the pair of rules $(r_1 : \ dead \lparrow shoot(suzy),\ \Not ab)$ and $(ab \lparrow broken)$.
For a detailed discussion of how causality and the inertia laws can combined we refer to~\cite{fandinno2015thesis}.}
\REVIEW{\ref{R2.7}}
\begin{gather*}
\begin{IEEEeqnarraybox}[][t]{l C ; l C l,C,l}
r_1   	&:& dead     &\lparrow& shoot(suzy), un\_broken\\
r_2   	&:& dead     &\lparrow& shoot(billy)\\
r_3   	&:& un\_broken   &\lparrow& repair(john)\\
c_{\ag} &:&\ashortprison(\ag)  &\lparrow& \ag \contributed dead
\end{IEEEeqnarraybox}
\hspace{1.25cm}
\begin{IEEEeqnarraybox}[][t]{l C ; l C l,C,l}
&& shoot(suzy)\\
&& shoot(billy)\\
&& repair(john)
\end{IEEEeqnarraybox}
\end{gather*}
\ENDREV
for $\ag \in \set{suzy,\ billy,\ john}$,
then it is easy to see that 
\begin{gather*}
I_{\ref{prg:broken.firing.squad}}(dead)
		\ \ = \ \ \big( repair(john) \cdotl r_3 * shoot(suzy) \big) \cdot r_1
		\ \ + \ \ shoot(billy) \cdotl r_2
\end{gather*}
where $I_{\ref{prg:broken.firing.squad}}$ is the least model of program~\programref{prg:broken.firing.squad} and, thus, $responsible(john,dead)$ will be a conclusion of it.
Just note that program~\programref{prg:broken.firing.squad} is the result of replacing atoms $loaded$ and $load(john)$ in program~\programref{prg:load.firing.squad.contributed} by $un\_broken$ and $repair(john)$, respectively.
Note also that nothing in program~\programref{prg:broken.firing.squad} reflects the fact that by default guns are $un\_broken$.
We state that guns are $un\_broken$ by default adding the following rule
\begin{IEEEeqnarray}{l C ? l C l}
1 &:& un\_broken &\lparrow& \Not broken
	\label{eq:unbroken.default}
\end{IEEEeqnarray}
If \newprogram\label{prg:broken.firing.squad.default} is the result of adding rule~\eqref{eq:unbroken.default} to program~\programref{prg:broken.firing.squad} and
$I_{\ref{prg:broken.firing.squad.default}}$ is the least model of~\programref{prg:broken.firing.squad.default},
then
\begin{gather*}
I_{\ref{prg:broken.firing.squad.default}}(un\_broken)
		\ \ = \ \ I_{\ref{prg:broken.firing.squad}}(un\_broken)
		\ \ + \ \ 1
		\ \ = \ \ 1
\end{gather*}
and, consequently,
\begin{align*}
I_{\ref{prg:broken.firing.squad}}(dead)
		&\ \ = \ \ \big( 1 \cdotl r_3 * shoot(suzy) \big) \cdot r_1
		\ \ + \ \ shoot(billy) \cdotl r_2
		\\
		&\ \ = \ \ \big( r_3 * shoot(suzy) \big) \cdot r_1
		\ \ + \ \ shoot(billy) \cdotl r_2
\end{align*}
which shows that John is not considered to have contributed to the prisoner's death.
Hence,
$\ashortprison(john)$ is not a conclusion of program~\programref{prg:broken.firing.squad.default}.
It is worth to mention that
besides the two syntactic differences between causal queries and m-queries already mentioned,
there is a, perhaps, less noticeable difference in the evaluation of causal literals.
Note that,
\begin{gather*}
\big( repair(john) \cdotl r_3 * shoot(suzy) \big) \cdot r_1
	\ \ \leq \ \ \big( r_3 * shoot(suzy) \big) \cdot r_1
\end{gather*}
and, thus,  if we replaced~$G \leqmax I(\rA)$ by~$G \leq I(\rA)$ in Definition~\ref{def:causal.literal.evaluation} (as done in~\citeNP{fandinno2015aspocp}),
it would follows that
atom~$\ashortprison(john)$ would be an unintended conclusion of program~\programref{prg:broken.firing.squad.default}.
It is also worth to mention that, besides \cite{Pearl00}
approach, the notion of contributory cause is also behind the definitions of actual cause given in~\cite{HP05,hall2007structural}.

\vspace{-0.15cm}

\section{Properties of causal logic programs}
\label{sec:properties}

Theorem~\ref{thm:regluar.standard.correspondnece}
established a correspondence for regular programs, but they say nothing about programs with causal queries.
For instance, positive program with non-monotonic causal literals may have more than one causal stable model.
Consider the following positive program~\newprogram\label{prg:positive.two-cmodels}
\begin{gather*}
\begin{IEEEeqnarraybox}[][t]{lC ' l C l , lrl}
r_1 &:& p
\\
r_3 &:& q
\end{IEEEeqnarraybox}
\hspace{2cm}
\begin{IEEEeqnarraybox}[][t]{lC ' l C l , lrl}
r_2 &:& q  &\lparrow& \ag_1 \necessary p
\\
r_4 &:& p  &\lparrow& \ag_2 \necessary q
\end{IEEEeqnarraybox}
\end{gather*}
obtained by adding rules $r_3$ and $r_4$ to program~\programref{prg:positive.two-minimal-models} and where $\ag_2 \eqdef \set{r_3}$.
Program~$\programref{prg:positive.two-cmodels}$
has two causal stable causal models.
The first that satisfies
$I_{\ref{prg:positive.two-cmodels}}(p) = r_1 + r_3 \cdotl r_4$
and
$I_{\ref{prg:positive.two-cmodels}}(q) = r_3$.
The second
$I'_{\ref{prg:positive.two-cmodels}}(p) = r_1$
and
$I'_{\ref{prg:positive.two-cmodels}}(q) = r_3 + r_1 \cdotl r_2$.
Let now $Q=\programref{prg:positive.two-cmodels}^{I_{\ref{prg:positive.two-cmodels}}}$ be the reduct of program~$\programref{prg:positive.two-cmodels}$ w.r.t. $I_{\ref{prg:positive.two-cmodels}}$, which consists in the following rules
\begin{gather*}
\begin{IEEEeqnarraybox}[][t]{lC ' l C l , lrl}
r_1 &:& p
\\
r_3 &:& q
\end{IEEEeqnarraybox}
\hspace{2cm}
\begin{IEEEeqnarraybox}[][t]{lC ' l C l , lrl}
r_2 &:& q  &\lparrow& (\cliteral{\cquery_1}{p})
\\
r_4 &:& p  &\lparrow& (\cliteral{\cquery_2}{q})
\end{IEEEeqnarraybox}
\end{gather*}
where
$\cquery_1(G,t) = 1$ iff there exists some $G' \leq G$ such that
$G' \leqmax I_{\ref{prg:positive.two-cmodels}}(p) = r_1 + r_3 \cdotl r_4$
and $\cquerynec_{\ag_1}(G',I_{\ref{prg:positive.two-cmodels}}(p))$
and 
$\cquery_2(G,t) = 1$ iff there exists $G' \leq G$ such that
$G' \leqmax I_{\ref{prg:positive.two-cmodels}}(q) = r_3$
and $\cquerynec_{\ag_2}(G',I_{\ref{prg:positive.two-cmodels}}(p))$.
First, note that
$\cquerynec_{\ag_1}(G', I_{\ref{prg:positive.two-cmodels}}(p))$
iff $I_{\ref{prg:positive.two-cmodels}}(p) = r_1 + r_3 \cdotl r_4 \leq \sum\ag_1 = r_1$ which does not hold.
Thus,
$\cquery_1(G,t) = 0$ for every $G \in \causes$ and $t \in \values$.
Then, it is clear that the body of rule $r_2$ is never satisfied and, therefore,
$\tprP{Q}{\alpha}(q)  =  r_3$ for any ordinal~$\alpha\geq 1$.
It can also be checked that
$\cquery_2(r_3, \tprP{Q}{\alpha}(q)) = 1$
because there exists $G' = r_3$ such that $G' \leqmax I_{\ref{prg:positive.two-cmodels}}'(q)= r_3$
and
$\cquerynec_{\ag_2}(G', I_{\ref{prg:positive.two-cmodels}}(q))
	= \cquerynec_{\ag_2}(r_3,r_3)=1$ since $r_3 \leq \sum\ag_2= r_3$.
Hence, since
$r_3 \leqmax \tprP{Q}{\alpha}(q)$ and
$\cquery_2(r_3, \tprP{Q}{\alpha}(q)) = 1$,
it follows that
$\tprP{Q}{\alpha}(\cliteral{\cquery_2}{q})  =  r_3$
and
$\tprP{Q}{\beta}(p) = r_1 + \tprP{Q}{\alpha}(q) \cdotl r_4 =  r_1 + r_3 \cdotl r_4 =  I_{\ref{prg:positive.two-minimal-models}}(p)$ for any ordinal~$\beta \geq 2$.
Hence,
$I_{\ref{prg:positive.two-cmodels}}$ is the least model of~$\programref{prg:positive.two-cmodels}^{I_{\ref{prg:positive.two-cmodels}}}$ and a causal stable model of  program~$\programref{prg:positive.two-cmodels}$.
Showing that $I_{\ref{prg:positive.two-cmodels}}'$ is also a causal stable model of~$\programref{prg:positive.two-cmodels}$ is symmetric. 

In the following we revise some desired general properties for a LP semantics.
First, causal stable models should also be supported models.
Note that the concept of supported model bellow
is analogous to the usual concept used in standard~LP,
but it is stronger in the sense that, not only requires that true atoms are supported, but also all their causes must be supported by a rule and a cause of its body.

\begin{definition}
A interpretation $\cI$ is a \emph{(causally) supported model} of a program $P$ iff $\cI$ is a model of $P$ and for every true atom $\rA$ and cause $G \in \causes$ such that $G \leq I(\rA)$ there is a rule~$\R$ in $P$ of the form of~\eqref{eq:rule} such that
$G \leq (\, I(\rB_1) * \dotsc * I(\rB_\npbody)) \cdot r_i$.\qed
\end{definition}

\begin{Proposition}{\label{prop:csm.are.supported.models}}
Any causal stable model $I$ of a program $P$ is a also supported model of $P$.\qed
\end{Proposition}

Furthermore, as happen with programs with nested negation under the standard stable models  semantics (where stable models may not be minimal models of the program),
causal stable models may not be minimal models either.
In fact, this may happen even when the nested negation is replaced by a non-monotonic causal literal.
Consider, for instance, the following program~\newprogram\label{prg:non-minimal.models}
\begin{gather*}
\begin{IEEEeqnarraybox}{lC? l ;C; rl , lrl}
r_1   &:&   p
\end{IEEEeqnarraybox}
\hspace{2cm}
\begin{IEEEeqnarraybox}{lC? l ;C; rl , lrl}
r_2   &:&   p   &\lparrow& \Not\ (\ag_1 \necessary p)
\end{IEEEeqnarraybox}
\end{gather*}
where $\ag_1 \eqdef \set{ r_1 }$.
Program~$\programref{prg:non-minimal.models}$ has two causal models.
One which satisfies
$I_{\ref{prg:non-minimal.models}}(p) = r_1$.
The other which satisfies
$I_{\ref{prg:non-minimal.models}}'(p) = r_1 + r_2$.
We define now the notion of \emph{normal program} whose causal stable models are also $\leq$-minimal models.
A program~$P$ is \emph{normal} iff no body rule in~$P$ contains a consistent literal (double negated literal) nor a negated non-monotonic causal literal.
In other words, a program is normal iff it does not contains nested negation nor non-monotonic causal literals in the scope of negation.

\begin{Proposition}{\label{prop:csm.are.minimal.models}}
Any causal stable model $I$ of normal program~$P$ is also a \mbox{$\leq$-minimal} model.\qed
\end{Proposition}

\paragraph{Splitting programs.}
\label{sec:splitting}

The intuitive meaning of the causal rule~\eqref{eq:r4.labelled} in programs~\programref{prg:necc.labelled} and~\programref{prg:necc3.labelled} is to cause the atom~$\afine(suzy)$ whenever the causal query expressed by its body is true with respect to a programs~\programref{prg:accident} and~\programref{prg:accident.billy}, respectively.
This intuitive understanding can be formalised as a splitting theorem in~\cite{lifschitz1994splitting}.

\begin{Theorem}[Splitting]{\label{thm:splitting}}
Let $\tuple{P_b,P_t}$ a partition of a program $P$ such that no atom occurring in the head of a rule in $P_t$ occurs in $P_b$.
An interpretation~$I$ is a causal stable model of $P$ iff there is some causal stable model $J$ of $P_b$ such that
$I$ is a causal stable model of
$(J \cup P_t)$.
\qed
\end{Theorem}

In our running example,
the bottom part are
$P_{\ref{prg:necc.labelled},b} = \programref{prg:accident}$ and
$P_{\ref{prg:necc3.labelled},b} = \programref{prg:accident.billy}$ 
while their top part
$P_{\ref{prg:necc.labelled},t}=P_{\ref{prg:necc3.labelled},t}$ is the program containing the rule~\eqref{eq:r4.labelled}.
This result can be generalised to infinite splitting sequences as follows.

\begin{definition}
A \emph{splitting sequence} of a program $P$ is a family $(P_\alpha)_{\alpha < \mu}$ of pairwise disjoint sets such that
$P = \bigcup_{\alpha < \mu} P_\alpha$
and no atom occurring in the head of a rule in some $P_\alpha$ occurs in the body of a rule in $\bigcup_{\beta < \alpha} P_\beta$.
A \emph{solution} of a splitting $(P_\alpha)_{\alpha < \mu}$
is a family $(I_\alpha)_{\alpha < \mu}$
such that
\begin{enumerate_thm}
\item $I_0$ is a stable model of $P_0$,
 \item $I_{\alpha}$
 is a stable model of $(J_{\alpha} \cup P_{\alpha})$ for any ordinal $0 < \alpha < \mu$  where
 $J_\alpha = \sum_{\beta < \alpha} I_\beta$.
\end{enumerate_thm}
A splitting sequence is said to be \emph{strict in $\alpha$} if, in addition, no atom occurring in the head of a rule in $P_\alpha$ occurs (in the head of a rule) in $\bigcup_{\beta < \alpha} P_\beta$   
and it is said to be \emph{strict} if it is strict in $\alpha$ for every $\alpha < \mu$.\qed
\end{definition}

\begin{Theorem}[Splitting sequences]{\label{thm:infinity.splitting}}
Let $(P_\alpha)_{\alpha < \mu}$ a splitting sequence of some program $P$.
An interpretation~$I$ is a causal stable model of $P$ iff
there is some solution $(I_\alpha)_{\alpha < \mu}$ of $(P_\alpha)_{\alpha < \mu}$
such that $I = \sum_{\alpha < \mu} I_\alpha$.
Furthermore, if such solution is strict in $\alpha$, then
$I_\alpha = \restr{I}{S_\alpha}$ where $S_\alpha$ is the set of all atoms not occurring in the head of any rule in~$\bigcup_{\alpha < \beta < \mu } P_\beta$
and $\restr{I}{S_\alpha}$ denotes the restriction if $I$ to $S_\alpha$.\qed
\end{Theorem}


A program $P$ is said to be \emph{stratified} if there is a some ordinal $\mu$ and mapping~$\lambda$ from the set of atoms~$\at$ into the set of ordinals $\set{\alpha < \mu}$ such that, for every rule of the form~\eqref{eq:rule} and atom $\rB$ occurring in its body, it satisfies $\lambda(\rH) \geq \lambda(\rB)$ if $\rB$ does not occur in the scope of negation nor in a non-monotonic causal literal, and $\lambda(\rH) > \lambda(\rB)$ if $\rB$ does occur under the scope of negation or in a non-monotonic causal literal.

\begin{Proposition}{\label{prop:stratified}}
Every stratified causal program $P$ has a unique causal stable model.\qed
\end{Proposition}


\section{Conclusions, related work and open issues}
\label{sec:conclusions}

The main contribution of this work is the introduction of a semantics for non-monotonic \emph{causal literals} that allow deriving new conclusions by inspecting the causal justifications of atoms in an \emph{elaboration tolerant} manner.
In particular, we have used causal literals to define \emph{necessary} and \emph{contributory causal relations} which are intuitively related to some of the most established definitions of actual causation in the literature~\cite{Pearl00,HP05,hall2007structural,halpern2015modification}.
Besides, by some running examples we have shown that causal literals allow, not only to derive whether some event is the cause or not of another event, but also to derive new conclusions from this fact.
From a technical point of view, we have shown that our semantics is a conservative extension of the stable model semantics and that satisfy the usual desired properties for an LP semantics (casual stable models are supported models, minimal models in case of normal programs and can be iteratively computed by split table programs).
It worth to mention that, besides the syntactic approaches to justifications in LP, the more related approach to our semantics is~\cite{damasio2013justifications}, for which a formal comparative can be found in~\cite{CabalarFF14inhibitors}
 and that \cite{pontelli2009justifications} allows a Prolog system to reason about justifications of an ASP program, but justifications cannot be inspected inside the ASP program.

Regarding complexity, it has been shown in~\cite{CabalarFF14Jelia} that  there may be an exponential number of causes for a given atom w.r.t. each causal stable model.
Despite that, the existence of stable model for programs containing only monotonic queries evaluable in polynomial time is $\NP$-complete~\cite{fandinno2015aspocp}.
For programs containing only necessary causal literals we can prove $\NP$-complete ($\NP$-hard holds even for programs containing a single negated regular literal or positive programs containing a single constraint, see Proposition~\ref{prop:complx.complete.necessary.cliterals} in the Appendix). The complexity for programs including other non-monotonic causal literals (like contributory) is still an open question.
A preliminary prototype extending the syntax of logic programs with causal literals capturing sufficient, necessary and contributory causal relation can be tested on-line at~\url{http://kr.irlab.org/cgraphs-solver/nmsolver}.

In a companion paper~\cite{CabalarF16},
the causal semantics used here has been extended to disjunctive logic programs, which will be useful for representing non-deterministic causal laws.
Interesting topics include a complexity assessment or studying an extension to arbitrary theories as with Equilibrium Logic \cite{Pea06} for the non-causal case;
 and formalise the relation between our notions of necessary and contributory cause with the above definitions of the actual causation and, in particular, with~\cite{Vennekens11} who has studied it in the context of CP-logic.
A promising approach seems to translate structural equations into logic programs in a similar way as it has been done to translate them into the causal theories~\cite{Giunchiglia2004,bochmanL15}.


\bibliographystyle{acmtrans}
\bibliography{refs}








\newpage

\appendix

\section{Nested expressions in rule bodies}
\label{sec:nested}

In this section we extend the syntax presented in Section~\ref{sec:causal.semantics} in order to allow nested expressions in rule bodies~\cite{lifschitzTT99nested}.

\begin{definition}
A \emph{formula} $\fF$ is recursively defined as one of the following expressions
\begin{gather*}
\fF \ \ ::= \ \ t \ \mid \ \rC \ \mid \ \fG , \fH \ \mid \ \fG ; \fH \ \mid \ \Not \fG
\end{gather*}
where $t$ is a term, $\rC$ is a causal literal (Definition~\ref{def:causal.literal}) and both, $\fG$ and $\fH$ are formulas in their turn. \QED
\end{definition}

A formula $\fF$ is said to be \emph{elementary} iff it is a term $t$ or a causal literal $\rC$.
It is said to be \emph{regular} iff every causal literal occurring in it is regular and is said to be \emph{positive} iff the operator $\Not$ does not occur in it.
$\fF$ is said to \emph{monotonic} iff every causal literal occurring in $\fF$ is monotonic.
In formulas, we will write $\top$ and $\bot$ instead of $1$ and $0$, respectively.

\begin{definition}[Causal logic program]\label{def:causal.P.nested}
Given a signature $\signature$, a \emph{(causal logic) program} $P$ is a set of rules of the form:
\begin{eqnarray}
r_i: \ \ \rH \ 
    \leftarrow \ \fF
    \label{eq:rule.nested}
\end{eqnarray}
\noindent where \mbox{$r_i\in Lb$} is a label or $r_i=1$, $\rH \in \at$ (the \emph{head} of the rule) is an atom or $\rH = \bot$ and $\fF$ (the \emph{body} of the rule) is a formula.\QED
\end{definition}

A rule $\R$ is said to be \emph{regular} iff its body is regular and its said to be \emph{positive} iff its body is positive and $\rH \neq \bot$.
It is said to be \emph{monotonic} iff $\fF$ is monotonic.
If \mbox{$\fF =\top$}, we say the rule is a \emph{fact} and omit the body and sometimes the symbol~~`$\leftarrow$.'
A program $P$ is \emph{regular}, \emph{positive} or \emph{monotonic} when all its rules are regular, positive or monotonic, respectively.
A \emph{standard program} is a regular in which the label of every rule is `$1:$'.
Definition~\ref{def:causal.P.nested} extends Definition~\ref{def:causal.P} by allowing nested expressions in the rule bodies.
A causal program in the sense of Definition~\ref{def:causal.P} is a program in which the body~$\fF$ of all rules are conjunctions of regular causal literals or their negation.
Note that every  rule of the form of~\eqref{eq:rule} with $\npbody = 0$ corresponds to a rule of the form of $(r_i : \ \rH \lparrow \top)$.

\paragraph{Semantics.}
The semantics of causal logic programs with nested expressions is given as follows.

\begin{samepage}
\begin{definition}[Valuation]
\label{def:formula.evaluation}
The valuation of causal literals and causal terms is as given by Definition~\ref{def:causal.literal.evaluation}.
Otherwise, the valuation of a formula $\fF$ is recursively defined as follows
\\[-4pt]
\begin{minipage}{\textwidth}
\begin{minipage}[t]{.45\textwidth}
\begin{IEEEeqnarray*}{l ; C ; l C l}
I(\fG,\fH)      &=&     I(\fG) &*& I(\fH)\\
I(\fG;\fH)      &=&     I(\fG) &+& I(\fH)
\end{IEEEeqnarray*}
\end{minipage}
\begin{minipage}[t]{.538\textwidth}
\begin{IEEEeqnarray*}{l ; C ; l}
I(\Not \fG)   &=&     \begin{cases}
                    1   &\text{iff } I(\fG) = 0\\
                    0   &\text{otherwise}
                    \end{cases}
\end{IEEEeqnarray*}
\end{minipage}
\end{minipage}
\\[6pt]
We say that $I$ satisfies a formula $\fF$, in symbols $I \models \fF$, iff $I(\fF) \neq 0$.\QED
\end{definition}
\end{samepage}


\begin{definition}[Causal model]\label{def:causal.model.nested}
Given a rule $\R$ of the form~\eqref{eq:rule.nested},
we say that an interpretation $I$ \emph{satisfies} $\R$, in symbols $I \models \R$,
if and only if the following condition holds:
\begin{gather}
I(\fF)  \cdot r_i \ \ \leq \ \ I(\rH)
  \label{eq:causal.stable.model.nested}
\end{gather}
An interpretation $I$ is a \emph{causal model} of $P$, in symbols $I\models P$ iff $I$~satisfies all rules in~$P$.\QED
\end{definition}


The following result shows that Definition~\ref{def:causal.model.nested} agrees with Definition~\ref{def:causal.model} for programs within the syntax of Definition~\ref{def:causal.P} and, thus, the former is a conservative extension of the last to programs with nested expressions in the body.

\begin{Proposition}{\label{prop:nested.correspondence}}
For any program~$P$ with the syntax of Definition~\ref{def:causal.P}, an interpretation $I$ is a model of~$P$ w.r.t. Definition~\ref{def:causal.model} iff $I$ is a model of $P$ w.r.t. Definition~\ref{def:causal.model.nested}.\qed
\end{Proposition}

We also can extend the definition of the direct consequences operator to programs with nested expressions as follows.

\begin{definition}[Direct consequences]\label{def:tp.nested}
Given a causal program with nested expressions~$P$, the operator of \emph{direct consequences} is a function $\tp$ from interpretations to interpretations such that
\begin{align*}
\tp(I)(\rA)
    \ \ &\eqdef \ \
    \sum \setbm{ \ I(\fF)  \cdot r_1 }
    		{ (r_i : \ \rA \lparrow \fF) \in P \ }
\end{align*}
for any interpretation $I$ and any atom $\rA \in \at$.
The iterative procedure is defined as usual
\begin{align*}
\tpr{\alpha} &\ \ \eqdef \ \ \tp(\tpr{\alpha-1})
                &&\text{if } \alpha \text{ is a successor ordinal}
\\
\tpr{\alpha} &\ \ \eqdef \ \ \sum_{\beta < \alpha}\tpr{\beta}
                &&\text{if } \alpha \text{ is a limit ordinal}
\end{align*}
As usual $0$ and $\omega$ respectively denote the first limit ordinal and the first limit ordinal that is greater than all integers. Thus, $\tpr{0} = \botI$.\qed
\end{definition}

We will show in the Appendix~\ref{sec:m-programs} that, if $P$ is monotonic and positive, then the $\tp$ operator has a least fixpoint that can be computed by iteration from the bottom interpretation~$\botI$.

\subsection{Causal stable models of programs with nested expressions.}

\begin{definition}[Reduct]\label{def:reduct.nested}
The reduct of a causal literal and terms is as in Definition~\ref{def:reduct}.
The reduct of formulas is inductively defined as follows\\
\begin{minipage}{\textwidth}
\begin{minipage}[t]{.45\textwidth}
\begin{IEEEeqnarray*}{l ? C ? l}
(\fE,\fH)^I      &=&    (\fE^I , \fH^I)\\
(\fE;\fH)^I      &=&    (\fE^I ; \fH^I)
\end{IEEEeqnarray*}
\end{minipage}
\begin{minipage}[t]{.538\textwidth}
\begin{IEEEeqnarray*}{l ? C ? l}
(\Not \fE)^I   &=&     \begin{cases}
                   \bot    &\text{if } I \models \fE^I\\
                   \top    &\text{otherwise}
                   \end{cases}
\end{IEEEeqnarray*}
\end{minipage}
\end{minipage}\\[6pt]
The reduct of program $P$ is the program
\begin{gather*}
P^I \ \ \eqdef \ \ \setm{ \R^I }{ \R \in P }
\end{gather*}
where the reduct $\R^I$ of a rule $\R$ like~\eqref{eq:rule} is given by
$(r_i: \ \aH \lparrow \fF^I)$.
\qed
\end{definition}

\begin{definition}[Formula equivalence]\label{def:equivalent.formula}
A formula $\fF$ is said to be \emph{equivalent} to a formula~$\fG$, in symbols $\fF \Leftrightarrow \fG$, iff any pair of causal interpretations $I$ and $J$ satisfy that
$I(\fF^J) = I(\fG^J)$.
\end{definition}

\begin{Proposition}{\label{prop:simplification}}
For any formula $F$, the following simplifications are valid
\begin{enumerate_thm}
\item 	$(\fF,\top) \Leftrightarrow  \fF$
	and $(\top,\fF) \Leftrightarrow  \fF$.
\item $(\fF;\top) \Leftrightarrow \top$
	and $(\top;F) \Leftrightarrow \top$.
\item $(\fF,\bot) \Leftrightarrow  \bot$
	and $(\bot,\fF) \Leftrightarrow  \fF$
\item $(\fF;\bot) \Leftrightarrow  \fF$
	and $(\bot;\fF) \Leftrightarrow  \fF$.
\qed
\end{enumerate_thm}
\end{Proposition}

\begin{definition}[Causal stable model]\label{def:causal.smodel.nested}
We say that an interpretation $I$ is a \emph{causal stable model} of a program with nested expressions~$P$ iff $I$~is the least model of the positive monotonic program~$P^I$ (Definition~\ref{def:reduct.nested}).\QED
\end{definition}

\begin{Proposition}{\label{prop:nested.reduct.correspondence}}
For any program~$P$ with the syntax of Definition~\ref{def:causal.P}, the reduct of $P$ w.r.t. to an interpretation $I$ and Definition~\ref{def:reduct} is the same as the reduct of $P $ w.r.t. $I$ and Definition~\ref{def:reduct.nested} after applying the simplifications from Proposition~\ref{prop:simplification} and removing all rules whose body is $\bot$.
Consequently, the causal stable models of $P$ w.r.t. Definitions~\ref{def:causal.smodel}  and~\ref{def:causal.smodel.nested} are the same.\qed
\end{Proposition}

\begin{Proposition}{\label{prop:csm.are.models.nested}}
Let $P$ be a causal program with nested expressions.
Any causal stable model $I$ of $P$ is a model of~$P$.\qed
\end{Proposition}

\begin{Proposition}{\label{prop:csm.are.supported.models.nested}}
Let $P$ be a causal program with nested expressions.
Any causal stable model $I$ of a $P$ is a also supported model of $P$.\qed
\end{Proposition}

\begin{Proposition}{\label{prop:csm.are.minimal.models.nested}}
Let $P$ be a causal program with nested expressions.
Then, any causal stable model $I$ of~$P$ is also a \mbox{$\leq$-minimal} model of~$P$.\qed
\end{Proposition}

Note that Propositions~\ref{prop:csm.are.supported.models} and~\ref{prop:csm.are.minimal.models}
in the main part of the paper, are direct consequences of Proposition~\ref{prop:nested.reduct.correspondence} together with Propositions~\ref{prop:csm.are.supported.models.nested} and~\ref{prop:csm.are.minimal.models.nested}, respectively.

\paragraph{Splitting programs.}
\label{sec:splitting2}

The intuitive meaning of the causal rule~\eqref{eq:r4.labelled} in programs~\programref{prg:necc.labelled} and~\programref{prg:necc3.labelled} is to cause the atom~$responsible(suzy,accident)$ whenever the causal query expressed by its body is true with respect to a programs~\programref{prg:accident} and~\programref{prg:accident.billy}, respectively.
This intuitive understanding can be formalised as a splitting theorem in~\cite{lifschitz1994splitting}.

\begin{Theorem}[Splitting]{\label{thm:splitting.nested}}
Let $\tuple{P_b,P_t}$ a splitting of some program with nested expressions $P$.
An interpretation~$I$ is a causal stable model of $P$ iff there is some causal stable model $J$ of $P_b$ such that
$I$ is a causal stable model of
$(J \cup P_t)$.
Furthermore, if $\tuple{P_b,P_t}$ is a strict splitting,
then $J = \restr{I}{S}$ where $S$ is the set of atoms of all atoms not occurring in the head of any rule in~$P_t$.\qed
\end{Theorem}

In our running example,
the bottom part are
$P_{\ref{prg:necc.labelled},b} = \programref{prg:accident}$ and
$P_{\ref{prg:necc3.labelled},b} = \programref{prg:accident.billy}$ 
while their top part
$P_{\ref{prg:necc.labelled},t}=P_{\ref{prg:necc3.labelled},t}$ is the program containing the rule~\eqref{eq:r4.labelled}.
We also can generalise this to infinite splitting sequences.

\begin{definition}
A \emph{splitting sequence} of a program $P$ is a family $(P_\alpha)_{\alpha < \mu}$ of pairwise disjoint sets such that
$P = \bigcup_{\alpha < \mu} P_\alpha$
and no atom occurring in the head of a rule in some $P_\alpha$ occurs in the body of a rule in $\bigcup_{\beta < \alpha} P_\beta$.
A \emph{solution} of a splitting $(P_\alpha)_{\alpha < \mu}$
is a family $(I_\alpha)_{\alpha < \mu}$
such that
\begin{enumerate_thm}
\item $I_0$ is a stable model of $P_0$,
 \item $I_{\alpha}$
 is a stable model of $(J_{\alpha} \cup P_{\alpha})$ for any ordinal $0 < \alpha < \mu$  where
 $J_\alpha = \sum_{\beta < \alpha} I_\beta$.
\end{enumerate_thm}
A splitting sequence is said to be \emph{strict in $\alpha$} if, in addition, no atom occurring in the head of a rule in $P_\alpha$ occurs (the head of a rule) in $\bigcup_{\beta < \alpha} P_\beta$   
and it is said to be \emph{strict} if it is strict in $\alpha$ for every $\alpha < \mu$.\qed
\end{definition}

\begin{Theorem}[Splitting sequences]{\label{thm:infinity.splitting.nested}}
Let $(P_\alpha)_{\alpha < \mu}$ a splitting sequence of some program with nested expressions $P$.
An interpretation~$I$ is a causal stable model of $P$ iff
there is some solution $(I_\alpha)_{\alpha < \mu}$ of $(P_\alpha)_{\alpha < \mu}$
such that $I = \sum_{\alpha < \mu} I_\alpha$.
Furthermore, if such solution is strict in $\alpha$, then
$I_\alpha = \restr{I}{S_\alpha}$ where $S_\alpha$ is the set of all atoms not occurring in the head of any rule in~$\bigcup_{\alpha < \beta < \mu } P_\beta$.\qed
\end{Theorem}


A program $P$ is said to be \emph{stratified} iff there is a some ordinal $\mu$ and mapping mapping $\lambda$ from the set of atoms~$\at$ into the set of ordinals $\set{\alpha < \mu}$ such that, for every rule of the form~\eqref{eq:rule.nested} and atom $\rB$ occurring in the body $\fF$, it satisfies $\lambda(\rH) \geq \lambda(\rB)$ if $\rB$ does not occur in the scope of negation or a non-monotonic causal literal, and $\lambda(\rH) > \lambda(\rB)$ if $\rB$ does occur under the scope of negation or a non-monotonic causal literal.

\begin{Proposition}{\label{prop:stratified.nested}}
Every stratified causal program with nested expressions $P$ has a unique causal stable model if it does not contain any rule whose head is $\bot$.\qed
\end{Proposition}

Propositions~\ref{prop:stratified}, in the main part of the paper, is a direct consequence of Propositions~\ref{prop:nested.reduct.correspondence} and~\ref{prop:stratified.nested}.


\paragraph{Normal form.}
Proposition~\ref{prop:nested.reduct.correspondence} show that Definition~\ref{def:causal.smodel.nested} is a conservative extension of Definitions~\ref{def:causal.smodel}.
In the following we show that, in fact, the syntax of~Definition~\ref{def:causal.P} is a normal form, that is, for every program $P$ in the syntax of Definition~\ref{def:causal.P.nested}, there is some program $Q$ with the syntax of Definition~\ref{def:causal.P} which has exactly the same causal stable models than $P$.

\begin{definition}\label{def:equivalent.program}
For program $P$ and $Q$ we write $P \Leftrightarrow Q$ when $I$ satisfies all rules in~$P^J$ iff $I$ satisfies all rules in~$Q^J$ for any pair of causal interpretations $I$ and $J$.\qed
\end{definition}

\begin{definition}[Strong equivalence]\label{def:strong.equivalence}
Two programs $P$ and $Q$ are said to be \emph{strongly equivalent} iff for every program $P'$, $(P \cup P')$ and $(Q \cup P')$ have the same causal stable models.
\end{definition}

\begin{Proposition}{\label{prop:strongly.equivalent}}
Any two causal programs $P$ and $Q$ s.t. $P \Leftrightarrow Q$ are strongly equivalent.\qed
\end{Proposition}

\begin{Proposition}{\label{prop:program.replacmentet}}
Let $P$ be a causal program, and let $\fF$ and $\fG$ be a pair of equivalent formulas, that is $\fF \Leftrightarrow \fG$. Any program obtained from $P$ by replacing some occurrences of $\fF$ by $\fG$ is strongly equivalent to~$P$.\qed
\end{Proposition}

The following result collects some of equivalence among formulas that correspond to those in~\cite{lifschitzTT99nested}.

\begin{Proposition}{\label{prop:transformations}}
For any formulas $\fF$, $\fG$ and $\fH$,
\begin{enumerate_thm}
\item $\fF,\fG \Leftrightarrow  \fG,\fF$ and $\fF;G \Leftrightarrow  G;F$.
\item $\fF,(\fG,\fH) \Leftrightarrow  (\fF,\fG),\fH$ and $\fF;(\fG;\fH) \Leftrightarrow  (\fF;\fG);\fH$.
\item $\fF,(\fG;\fH) \Leftrightarrow  (\fF,\fG);(\fF,\fH)$ and $\fF;(\fG,\fH) \Leftrightarrow  (\fF;\fG),(\fF;\fH)$.

\item $\Not (\fF,\fG) \Leftrightarrow  I(\Not \fF; \Not \fG)$ and
$\Not (\fF;\fG)) \Leftrightarrow  \Not \fF, \Not \fG$.

\item  $\Not\Not\Not F \Leftrightarrow \Not \fF$.

\item $\fF,\top \Leftrightarrow  F$ and $\fF;\top \Leftrightarrow \top$.
\item $\fF,\bot \Leftrightarrow  \bot$ and $\fF;\bot \Leftrightarrow  \fF$.
\item $\Not \top \Leftrightarrow \bot $ and $\Not \bot \Leftrightarrow  \top$.
\qed
\end{enumerate_thm}
\end{Proposition}

A formula $F$ is said to be a \emph{simple conjunction} (resp. \emph{simple disjunction}) iff is a conjunction (resp. disjunction) of elementary formulas.

\begin{Proposition}{\label{prop:formula.normal.form}}
Any formula $F$ is equivalent to a formula of the form
\begin{enumerate_thm}
\item  $F_1; \dotsc ; F_n$ where $n \geq 1$ and each $F_i$ is a simple conjunction, and
\item  $F_1, \dotsc , F_n$ where $n \geq 1$ and each $F_i$ is a simple disjunction.\qed
\end{enumerate_thm}
\end{Proposition}

\begin{Proposition}{\label{prop:body.disjuntion.elimination}}
A causal rule ($r_i : \ \rA \leftarrow \fF; \fG$) is equivalent to
\begin{align*}
r_i : \ \rA &\ \leftarrow \fF\\
r_i : \ \rA &\ \leftarrow \fG
\end{align*}
for any label $r_i$, atom $\rA$ and formulas $\fF$ and $\fG$.\qed
\end{Proposition}

\begin{Proposition}{\label{prop:program.normal.form.strong}}
Any program is strongly equivalent of a set of rules of the form~\eqref{eq:rule} if $\bot$ is allowed in the head.\qed
\end{Proposition}

\begin{Proposition}{\label{prop:program.normal.form}}
For every program $P$, there is some program $Q$ with the syntax of Definition~\ref{def:causal.P} which has exactly the same causal stable models than $P$.\qed
\end{Proposition}

\section{Uniform reduct for monotonic and non-monotonic queries}

An issue with Definitions~\ref{def:reduct} and~\ref{def:reduct.nested} is that it is necessary to know whether a causal query is monotonic or not to apply the reduct.
This can be provided by the user, but otherwise automatically checked whether a causal query is monotonic or not can be computationally costly.
In the following, we show that, in fact, this distinction is not necessary and that the reduct can be applied uniformly to monotonic and non-monotonic causal literals.

\begin{definition}[Reduct]\label{def:reduct.uniform}
The \emph{reduct} of causal queries is defined as in Definition~\ref{def:reduct}.
The reduct of a causal literal is given by~$(\cliteral{\cquery^{I(\rA)}}{\rA})$ for any causal literal of the form of~$(\cliteral{\cquery}{\rA})$.
The reduct of formulas, rules and programs is then defined as in Definition~\ref{def:reduct.nested}.\qed
\end{definition}

Definition~\ref{def:reduct.uniform} applies the reduct uniformly to monotonic and non-monotonic causal literals.
A consequence of this fact is that the reduct of monotonic programs is not itself and, in fact, the least model of the reduct of a monotonic program $P^I $w.r.t. an interpretation $I$ can be different according to Definitions~\ref{def:reduct.nested} and~\ref{def:reduct.uniform}.
Despite that, the following result shows that the causal stable models of a program~$P$ are the same in spite of whether Definition~\ref{def:reduct.nested} or Definition~\ref{def:reduct.uniform} is used.

\begin{Proposition}{\label{prop:program.reduct.is.monotonic.nested}}
Let $P$ be a causal program with nested expressions.
An interpretation $\cI$ is the least model of $P^\cI$ (according to Definition~\ref{def:reduct.nested}) iff $\cI$ is the least model of $P^\cI$ (according to Definition~\ref{def:reduct.uniform}). 
\qed
\end{Proposition}

\section{Comparative with~\protect\cite{fandinno2015aspocp}}
\label{sec:m-programs}

In this section we revise the syntax and semantics of causal programs given in~\cite{fandinno2015aspocp} and show how programs in this framework can be translated in ours.

\paragraph{Syntax.}
A \emph{m-query} is a monotonic function
$\mquery: \graphs \longrightarrow \set{0,1}$ assigning true or false to every causal graphs $G \in \causes$.
A \emph{signature} is a triple
\msignature\ where $\at$, $\lb$ and $\mqueries$ respectively represent sets of 
\emph{atoms} (or \emph{propositions}), \emph{labels}
and query functions.

\begin{definition}[m-literal]\label{def:m-literal}
A \emph{m-literal} is an expression 
$(\cliteral{\mquery}{\rA})$ where $\rA \in \at$ is an atom and $\mquery \in \mqueries$ is a m-query.\QED
\end{definition}

Formulas, rules and programs are defined as in our framework (Section~\ref{sec:nested}), but replacing causal literals (Definition~\ref{def:causal.literal}) by m-literals (Definition~\ref{def:m-literal}).

\paragraph{Semantics.}
The semantics of m-programs is as follows.

\begin{definition}[Valuation]
\label{def:m-formula.evaluation}
The \emph{valuation of a causal literal} of the form $(\cliteral{\mquery}{\rA})$ with respect to an interpretation $I$ is given by
\vspace{-0.01cm}
\begin{align*}
I(\cliteral{\mquery}{\rA}) \ \ &\eqdef \ \
  \sum\setbm{ G \in \graphs}{G \leq I(\rH) \text{ and } \ \mquery(G) = 1  }
\end{align*}
The valuation of causal terms and formulas is inductively defined as in Definition~\ref{def:formula.evaluation}.
\end{definition}

The definition of causal models and the $\tp$ operator is as in Definitions~\ref{def:causal.model.nested} and~\ref{def:tp.nested}, respectively, but evaluating formulas according to Definition~\ref{def:m-formula.evaluation} instead of Definition~\ref{def:formula.evaluation}.

\begin{theorem}[From~\protect\citeNP{fandinno2015aspocp}]\label{thm:m-tp.properties}
Let $P$ be a (possibly infinite) positive logic program (with nested expressions).
Then,
($i$) $\lfp(T_P)$ is the least model of $P$ and
($ii$)  $\lfp(T_P)=\tpr{\omega}$.\QED
\end{theorem}

\begin{theorem}[From~\protect\citeNP{fandinno2015aspocp}]\label{thm:m-positive.two-valued.model.correspondence}
Let $P$ be a regular positive program (with nested expressions) and $Q$ its standard unlabelled version.
Then, the least model $J = I^{cl}$ of $Q$ is the two-valued interpretation corresponding to the least model $I$ of $P$.\QED
\end{theorem}

The definition of reduct and causal stable models is as Definitions~\ref{def:reduct.nested} and Definition~\ref{def:causal.smodel.nested}.

\begin{theorem}[From~\protect\citeNP{fandinno2015aspocp}]
\label{thm:m-regluar.standard.correspondnece}
Let $P$ be a regular program (with nested expressions) and $Q$ be its corresponding standard program obtained by removing all labels in $P$.
Then, $P$ and $Q$ are two-valued equivalent.\QED
\end{theorem}



\paragraph{Encoding \protect\cite{fandinno2015aspocp} m-programs in our framework.}
In the following we show that every program according to~\cite{fandinno2015aspocp} can be fitted in our framework.

\begin{definition}
Given a m-program~$Q$, its corresponding program~$P$ consists of rule of the form
\begin{gather*}
r_i: \ \ \rH \ 
    \leftarrow \ \fF'
\end{gather*}
for every rule of the form~$(r_i: \ \ \rH \ \leftarrow \ \fF)$ in $Q$ where $\fF'$ is the result of replacing every m-query $\mquery$ by its corresponding query $\cquery$ given by $\cquery(G,t) = \mquery(G)$.\qed
\end{definition}

\begin{Proposition}{\label{prop:corresponding.of.m-program}}
If $P$ is the corresponding program of some positive m-program (with nested expressions)~$Q$,
then  an interpretation~$I$ is a model of~$P$ iff $I$ is a model of $Q$.\qed
\end{Proposition}

\paragraph{Encoding of monotonic programs into~\protect\cite{fandinno2015aspocp}.}
It is clear that not every program in our framework can be fitted into a m-program because the last only allows monotonic queries.
However, if all causal queries in a program are monotonic, then there is an equivalent m-program given as follows.

\begin{definition}
Given a program with nested expressions~$P$ in which all causal queries are monotonic, its corresponding m-program~$Q$ consists of rule of the form
\begin{gather*}
r_i: \ \ \rH \ 
    \leftarrow \ \fF'
\end{gather*}
for every rule of the form~$(r_i: \ \ \rH \ \leftarrow \ \fF)$ in $Q$ where $\fF'$ is the result of replacing every query $\cquery$ by its corresponding query $\mquery$ given by $\mquery(G) = \cquery(G,1)$.\qed
\end{definition}

\begin{Proposition}{\label{prop:corresponding.m-program}}
If $Q$ is the corresponding m-program of some positive monotonic program with nested expressions~$P$,
then an interpretation~$I$ is a model of~$P$ iff $I$ is a model of $Q$.\qed
\end{Proposition}

Note that Theorem~\ref{thm:fandino2015.monotonic} is a direct consequence of Proposition~\ref{prop:nested.correspondence} together with the result of Propositions~\ref{prop:corresponding.of.m-program} and~\ref{prop:corresponding.m-program}.
Furthermore, the following Corollaries~\ref{thm:tp.properties.nested},  \ref{thm:positive.two-valued.model.correspondence.nested} and~\ref{thm:regluar.standard.correspondnece.nested} are direct consequences of Proposition~\ref{prop:corresponding.m-program} together with the results of Theorems~\ref{thm:m-tp.properties}, \ref{thm:m-positive.two-valued.model.correspondence} and~\ref{thm:m-regluar.standard.correspondnece}, respectively.
Corollary~\ref{cor:regluar.correspondnece.nested} is a direct consequence of Corollary~\ref{thm:regluar.standard.correspondnece.nested}.

\begin{Corollary}{\label{thm:tp.properties.nested}}
Any (possibly infinite) positive monotonic causal program with nested expressions~$P$ has a least causal model $I$ which
 (i) coincides with the least fixpoint $\lfp(\tp)$ of the direct consequences operator $\tp$ and (ii) can be iteratively computed from the bottom interpretation $I = \lfp(\tp) = \tpr{\omega}$.\qed
\end{Corollary}

\begin{Corollary}{\label{thm:positive.two-valued.model.correspondence.nested}}
Let $P$ be a regular positive monotonic program with nested expressions and $Q$ its standard unlabelled version obtained by removing all labels from the rules in $P$.
Let $\cI$ and $\cJ$ be the least models of $P$ and $Q$, respectively.
Then, $\cI^{cl} = \cJ$.\QED
\end{Corollary}

\begin{Corollary}{\label{thm:regluar.standard.correspondnece.nested}}
Let $P$ be a regular program with nested expressions and $Q$ be its corresponding standard program obtained by removing all labels in $P$.
Then $P$ and $Q$ are two-valued equivalent.\QED
\end{Corollary}

\begin{Corollary}{\label{cor:regluar.correspondnece.nested}}
Any two regular programs with nested expressions that only differ in their labels are two-valued equivalent.\QED
\end{Corollary}

Corollaries~\ref{thm:tp.properties} and \ref{thm:positive.two-valued.model.correspondence} and Theorem~\ref{thm:regluar.standard.correspondnece}
in the main part of the paper are direct consequences of Proposition~\ref{prop:nested.correspondence} plus 
Corollaries~\ref{thm:tp.properties.nested}, \ref{thm:positive.two-valued.model.correspondence.nested} and~\ref{thm:regluar.standard.correspondnece.nested}, respectively.

\section{Proof of Results}
\label{sec:proofs}

\subsection{Preliminary facts}

\begin{proposition}[From~\protect\citeNP{CabalarFF14}]\label{prop:operations.monotonicity}
\review{R2.4}{Addition}, product and application are monotonic operations, that is,
$t + u \leq t' + u'$, \
$t * u \leq t' * u'$ and
$t \cdot u \leq t' \cdot u'$
for any causal values $\set{t,u,t',u'} \subseteq \values$ such that
$t \leq t'$ ant $u \leq u'$.\qed
\end{proposition}


\begin{proposition}[From~\protect\citeNP{CabalarFF14}]\label{prop:join.prime}
Every causal value $G \in \causes$ without \review{R2.4}{addition} is completely \review{R2.4}{addition-prime}, that is,
$G \leq \sum_{t \in T} t$ implies that $G \leq t$ for some $t \in T$ where $T \subseteq \values$ is a set of causal values.\qed
\end{proposition}

\subsection{Properties of the causal queries and causal literals}


\begin{Proposition}{\label{prop:cliteral.monotonic}}
The evaluation of a causal literal~$(\cliteral{\cquery}{\rH})$ is $\leq$-monotonic for every monotonic causal query~$\cquery$, that is,
$J(\cliteral{\cquery}{\rH}) \leq I(\cliteral{\cquery}{\rH})$
for every pair of interpretations $I$ and $J$ such that $J \leq I$.\qed
\end{Proposition}

\begin{proof}
By definition, it follows that
\begin{align*}
X(\cliteral{\cquery}{\rH}) \ \ &\eqdef \ \
  \sum\setbm{ G \in \causes}{G \leqmax X(\rH) \text{ and } \ \cquery(G,\, X(\rH)\,) = 1  }
\end{align*}
with $X \in \set{I,J}$.
For the sake of contradiction, suppose that $J(\cliteral{\cquery}{\rH}) \not\leq I(\cliteral{\cquery}{\rH})$.
Then, there is
$G \leqmax J(\cliteral{\cquery}{\rH})$
such that
$G \not\leq I(\cliteral{\cquery}{\rH})$.
Note that
$G \leqmax J(\cliteral{\cquery}{\rH})$
implies
$G \leqmax J(\rH)$
and, since $J \leq I$,
this implies that there exists $G' \leqmax I(\rH)$ such that $G \leq G'$.
Hence,
since $J \leq I$ and $\cquery$ is monotonic,
$\cquery(G,\, J(\rH)\,) = 1$
implies
$\cquery(G',\, I(\rH)\,) = 1$
and, therefore, $G \leq G' \leq I(\cliteral{\cquery}{\rH})$ which contradicts the assumption.
\end{proof}

\begin{Proposition}{\label{prop:cquery.reduct.monotonic}}
The reduct of a causal query $\cquery$ w.r.t. term $t$, in symbols $\cquery^t$ is monotonic.\qed
\end{Proposition}


\begin{proof}
Suppose that $\cquery^{t}$ is not monotonic.
Then there are $G,G'' \in \causes$ and $u,w\in\values$ such that $G \leq G''$
and
$\cquery^{t}(G,u) = 1$, but
$\cquery^{t}(G'',w) = 0$.
By definition,
\begin{gather}
\cquery^t(G,u) = 1
\quad\text{iff exists some }  G' \leq G \text{ s.t. } G' \leqmax t
\text{ and }
\cquery(G',\,t) = 1
    \label{eq:1prop:cquery.reduct.monotonic}
\end{gather}
Similar for $G''$ and $w$.
Pick some $G$ satisfying~\eqref{eq:1prop:cquery.reduct.monotonic}.
Since $G' \leq G$ and $G \leq G''$, it follows that $G' \leq G''$
and, since $G' \leq G''$ and $G' \leqmax t$ and
\mbox{$\cquery(G',t) = 1$},
it also follows that
\mbox{$\cquery^{t}(G'',w) = \cquery^{t}(G'',u) = 1$}
which is a contradiction with the assumption.
Hence, $\cquery^{t}$ is monotonic.
\end{proof}

\begin{Proposition}{\label{prop:monotonic.cquery.reduct.leq}}
Any monotonic causal query $\cquery$ satisfies
that $\cquery^t(G,u) \leq \cquery(G,u)$ for any causal values $G \in \causes$ and $\set{t,u} \subseteq \values$.\qed
\end{Proposition}


\begin{proof}
Suppose that
$\cquery^t(G,u) \not\leq \cquery(G,u)$.
Then,
$\cquery^t(G,u) = 1$
and
$\cquery(G,u) = 0$.
By definition,
\begin{gather*}
\cquery^t(G,u) \ \ \eqdef \ \
  \begin{cases}
  1 &\text{iff exists some }  G' \leq G \text{ s.t. } G' \leqmax t
      \text{ and }
      \cquery(G',\,t) = 1
  \\
  0 &\text{otherwise}
  \end{cases}
\end{gather*}
and, thus, there exists some $G' \leq G$ such that $G' \leqmax t$ and that
$\cquery(G',\,t) = 1$.
Since $\cquery$ is monotonic,
$G' \leq G$ and $\cquery(G',t) = 1$ 
implies that $\cquery(G,u) = 1$ for any $u\in\values$ which is a contradiction with the fact that $\cquery(G,u) = 0$.
\end{proof}


\begin{Proposition}{\label{prop:monotonic.cliteral.reduct.leq}}
Let $I$ and $J$ be two interpretations.
Then,
$J(\cliteral{\cquery}{\rH})^I \leq J(\cliteral{\cquery}{\rH})$ for any atom~$\rH \in \at$ and any a monotonic causal query~$\cquery$.
\end{Proposition}

\begin{proof}
Pick any $G \leqmax J(\cliteral{\cquery}{\rH})^I$.
By definition, $G \leqmax J(\rH)$ and $\cquery^{I(\rA)}(G,J(\rA)) = 1$.
Furthermore,
from~Proposition~\ref{prop:monotonic.cquery.reduct.leq},
$\cquery^{I(\rA)}(G,J(\rA)) = 1$
implies
$\cquery(G,J(\rA)) = 1$
and, thus,
$G \leq J(\cliteral{\cquery}{\rH})$.
Therefore,
$J(\cliteral{\cquery}{\rH})^I \leq J(\cliteral{\cquery}{\rH})$.
\end{proof}

Note that in general
$J(\cliteral{\cquery}{\rH})^I \neq J(\cliteral{\cquery}{\rH})$
may hold even if $J \leq I$.
Consider, for instance, a pair of interpretation
$J(\rH) = a*b$ and $I(\rH) = a$
and a monotonic causal query $\cquery(a*b) = \cquery(a) = 1$.
Then, $J(\cliteral{\cquery}{\rH}) = a*b$, but
$J(\cliteral{\cquery}{\rH})^I = 0$
because $a*b \not\!\!\leqmax I(\rA)$.

\subsection{Properties of formulas}


\begin{Proposition}{\label{prop:formula.monotonic}}
Any monotonic formula $\fF$ is $\leq$-monotonic,
that is, $J(\fF) \leq I(\fF)$ for any causal interpretations $I$ and $J$ such that $J \leq I$.\qed
\end{Proposition}

\begin{proof}
In case that $F$ is a causal literal of the form
$(\cliteral{\psi}{\rH})$, from Proposition~\ref{prop:cliteral.monotonic},
it follows that
$J(\cliteral{\psi}{\rH}) \leq I(\cliteral{\psi}{\rH})$.
Otherwise, we proceed by structural induction assuming the lemma holds for every subformula of $\fF$.
In case that $\fF = (\rE,\rH)$, by induction hypothesis $\rE$ and $\rH$ are $\leq$-monotonic and, thus, since product is also monotonic, it follows that $\fF$ is $\leq$-monotonic.
The case  $\fF = (\rE;\rH)$ is analogous.
Finally, for the case $\fF =\Not \rE$, just note that $\fF$ is not positive
and, thus, $\fF$ is not monotonic by definition.
\end{proof}


\begin{Proposition}{\label{prop:formula.valuation.reduct.stable}}
Any causal interpretation $I$ and formula $\fF$ satisfy $I(\fF^I) = I(\fF)$.\qed
\end{Proposition}

\begin{proof}
In case that $\fF$ is a causal literal of the form
$(\cliteral{\psi}{\rA})$,
its reduct
$(\cliteral{\psi}{\rA})^I$ is $(\cliteral{\psi^{I(\rA)}}{\rA})$.
Furthermore, by definition,
\begin{gather*}
\cquery^{I(\rA)}(G,t) = 1
\quad \text{iff exist } G' \leq G \text{ s.t. }
G' \leqmax I(\rA) \text{ and } \cquery(G',\,I(\rA))
\end{gather*}
Then, $G \leqmax I(\cliteral{\cquery^{I(\rA)}}{\rA})$ implies that
$G \leqmax I(\rA)$
and
there exists some $G' \leq G$ such that
$G' \leqmax I(\rA)$ and $\cquery(G',I(\rA)) = 1$.
Note that, since $G' \leq G\leqmax I(\rA)$ and $G' \leqmax I(\rA)$, it follows that $G= G'$.
Then, $\cquery(G',I(\rA)) = 1$ implies that
$\cquery(G,I(\rA)) = 1$
and, consequently, $G \leq I(\cliteral{\cquery}{\rA})$.
That is, $I(F^I) \leq I(F)$.
The other way around.
$G \leqmax I(\cliteral{\cquery}{\rA})$
implies that
$G \leqmax I(\rA)$
and
$\cquery(G,I(\rA)) = 1$
which in turn imply that
$\cquery^{I(\rA)}(G,I(\rA)) = 1$
and
$G \leqmax I(\cliteral{\cquery^{I(\rA)}}{\rA})$.
Consequently, 
$I(\cliteral{\cquery^{I(\rA)}}{\rA}) = I(\cliteral{\cquery}{\rA})$.
\\[-5pt]

\noindent
In any other case, we proceed by structural induction assuming the lemma holds for every subformula of $\fF$.
In case that $\fF = (\fE,\fH)$, by definition,
\begin{align*}
I(\fF^I) 
      \ \ = \ \ I(\fE,\fH)^I
      \ \ = \ \ I(\fE^I , \fH^I)
      \ \ = \ \ I(\fE^I) * I(\fH^I)
\end{align*}
Furthermore, by induction hypothesis
$I(\fE^I) = I(\fE)$
and
$I(\fH^I) = I(\fH)$ and, thus,
\begin{align*}
I(\fF^I) 
      \ \ = \ \ I(\fE) * I(\fH)
      \ \ = \ \ I(\fE,\fH)
      \ \ = \ \ I(\fF)
\end{align*}
The case  $\fF = (\fE;\fH)$ is analogous.
Finally, for the case $\fH =\Not \fE$, just note that
$I(\Not \fE)^I = I(\bot) = 0$ iff $I\models \fE^I$ iff $I \models \fE$ iff
$I(\Not \fE) = 0$.
Otherwise
$I(\Not \fE)^I = I(\top) = 1$ and $I(\Not \fE) = 1$.
\end{proof}

\begin{lemma}\label{lem:monotonic.formula.prime.reduct.leq}
Let $\cI$ and $\cJ$ be two interpretations.
Then,
$\cJ(\fF^\cI) \leq \cJ(\fF') \leq \cJ(\fF)$ for any monotonic formula $\fF$ and $\fF'$ where $\fF'$ is either $\fF^\cI$ or the result of replacing in $\fF^\cI$ some reduced causal query~$\cquery^t$ by its non-reduced form~$\cquery$.\qed
\end{lemma}

\begin{proof}
In case that $\fF$ is a causal literal of the form
$(\cliteral{\cquery'}{\rH})$, from Proposition~\ref{prop:monotonic.cliteral.reduct.leq},
it follows that
$J(\fF^\cI) = J(\cliteral{\cquery^{I(\rH)}}{\rH}) \leq J(\cliteral{\cquery}{\rH}) = J(\fF)$.
Furthermore, if in addition $\cquery' = \cquery$,
it follows that $\fF' = \fF$ and the lemma statement follow from the above inequality.
Otherwise $\fF' = \fF^\cI$ and the result follow in a similar way.
\\[-5pt]

\noindent
We proceed by structural induction assuming the statement holds for every subformula of $\fF$.
In case that $\fF = (\fE,\fH)$, by definition,
\begin{IEEEeqnarray*}{l C l C l C l}
\cJ(\fF^\cI) 
      & \ \ = \ \ & \cJ(\,(\fE,\fH)^\cI\,)
      & \ \ = \ \ & \cJ(\fE^\cI , \fH^\cI)
      & \ \ = \ \ & \cJ(\fE^\cI) * \cJ(\fH^\cI)
\\
\cJ(\fF') 
      & \ \ = \ \ & \cJ(\,(\fE,\fH)'\,)
      & \ \ = \ \ & \cJ(\fE' , \fH')
      & \ \ = \ \ & \cJ(\fE') * \cJ(\fH')
\end{IEEEeqnarray*}
Furthermore, by induction hypothesis,
$\cJ(\fE^\cI) \leq \cJ(\fE') \leq \cJ(\fE)$
and
$\cJ(\fH^\cI) \leq \cJ(\fH') \leq \cJ(\fH)$ and, since product~$*$ is monotonic, it follows that,
\begin{IEEEeqnarray*}{l C l C l C l C l}
\cJ(\fF^\cI) 
      & \ \ = \ \    & \cJ(\fE^\cI) * \cJ(\fH^\cI)
      & \ \ \leq \ \ & \cJ(\fE') * \cJ(\fH')
      & \ \ = \ \    & \cJ(\fE',\fH')
      & \ \ = \ \    & \cJ(\fF')
\\
\cJ(\fF') 
      & \ \ = \ \    & \cJ(\fE') * \cJ(\fH')
      & \ \ \leq \ \ & \cJ(\fE) * \cJ(\fH)
      & \ \ = \ \    & \cJ(\fE,\fH)
      & \ \ = \ \    & \cJ(\fF)
\end{IEEEeqnarray*}
The case  $\fF = (\fE;\fH)$ is analogous.
Finally, note that $\fF =\Not \fE$ is not a positive formula and, by definition, it is not a monotonic formula either.
\end{proof}

\begin{Proposition}{\label{prop:monotonic.formula.reduct.leq}}
Let $\fF$ be a monotonic formula and $\cI$ be an interpretation.
Then,
$J(\fF^I) \leq J(\fF)$ for any interpretation $\cJ$ such that $\cJ \leq \cI$.\qed
\end{Proposition}

\begin{proof}
It follows directly from Lemma~\ref{lem:monotonic.formula.prime.reduct.leq}.
\end{proof}


\begin{Proposition}{\label{prop:normal.formula.reduct.leq.formula}}
Let $\fF$ be a normal formula and $\cI$ be an interpretation.
Then, $J(\fF^I) \leq J(\fF)$ for any interpretation $J$ such that $\cJ \leq \cI$.\qed
\end{Proposition}

\begin{proof}
In case that $\fF$ is a causal literal of the form
$(\cliteral{\psi}{\rA})$,
its reduct
$(\cliteral{\psi}{\rA})^I$ is $(\cliteral{\psi^{I(\rA)}}{\rA})$.
Note that $G \leqmax J(\cliteral{\psi^{I(\rA)}}{\rA})$
implies $G \leqmax J(\rA)$ and
$\cquery^{I(\rA)}(G,J(\rA)) = 1$.
Furthermore, by definition,
\begin{gather*}
\cquery^{I(\rA)}(G,J(\rA)) = 1
\quad \text{iff exist } G' \leq G \text{ s.t. }
G' \leqmax I(\rA) \text{ and } \cquery(G',\,I(\rA))
\end{gather*}
Since $G' \leq G \leq J(\rA) \leq I(\rA)$
and $G \leqmax I(\rA)$, it follows that $G=G'$.
Then, since $J \leq I$, queries are anti-monotonic in the second argument and $\cquery(G,I(\rA)) = 1$, it follows that
 $\cquery(G,J(\rA)) = 1$
 and, since $G \leqmax J(\rA)$, it also follows that
$G \leq J(\cliteral{\cquery}{\rA})$.
That is, $J(\fF^I) \leq J(\fF)$
\\[-5pt]

\noindent
Otherwise, we proceed by structural induction assuming the lemma holds for every subformula of $\fF$.
In case that $\fF = (\rE,\rH)$, by definition,
\begin{align*}
J(\fF^I) 
      \ \ = \ \ J(\,(\fE,\fH)^I\,)
      \ \ = \ \ J(\fE^I , \fH^I)
      \ \ = \ \ J(\fE^I) * J(\fH^I)
\end{align*}
Furthermore, by induction hypothesis,
$J(\fE^I) \leq J(\fE)$
and
$J(\fH^I) \leq J(\fH)$ and, since product~$*$ is monotonic, it follows that,
\begin{align*}
J(\fF^I) 
      \ \ = \ \ J(\fE^I) * J(\fH^I)
      \ \ \leq \ \ J(\fE) * J(\fH)
      \ \ = \ \ J(\fE,\fH)
      \ \ = \ \ J(\fF)
\end{align*}
The case  $\fF = (\fE;\fH)$ is analogous.
Finally, in case $\fF =\Not \fE$, since $\fF$ is a normal formula,
$\rE$ is positive and every query $\cquery$ occurring in $\fE$ is monotonic.
Hence, from the fact $J\leq I$ and the fact that monotonic formulas are also $\leq$-monotonic, it follows $J(\fE) \leq I(\fE)$ (Proposition~\ref{prop:formula.monotonic}).
Furthermore,
\begin{align*}
\text{if }J(\,(\Not \fE)^I\,)=1 \
    &\text{then } I \not\models\fE^I
\\
    &\text{then } I(\fE^I) = 0
\\
    &\text{then } I(\fE) = 0
      &\text{(Proposition~\ref{prop:formula.valuation.reduct.stable})}
\\
    &\text{then } J(\fE) = 0
\\
    &\text{then } J(\Not \fE) = 1
\end{align*}
That is, $J(\,(\Not \fE)^I\,)=1$ implies that $J(\Not \fE) = 1$.
Otherwise, $J(\Not \fE)^I=0$ and the term~$0$ is smaller than any causal value and, thus, $J(\,(\Not \fE)^I\,) \leq J(\Not \rE)$ holds
and, consequently, it follows that $J(\fF^I) \leq J(\fF)$.
\end{proof}



\subsection{Proof of Proposition~\ref{prop:nested.correspondence}}

\begin{proofof}{prop:nested.correspondence}
Assume that $I$ is a model of $P$ w.r.t. Definition~\ref{def:causal.model} and suppose that $I$ is not a model of $P$ w.r.t. Definition~\ref{def:causal.model.nested}.
Then, there is a rule $\R$ of the form of $(r_1 : \rH \lparrow \rB_1, \dotsc, \rB_\npbody)$ such that
$I(\rB_1, \dotsc, \rB_\npbody) \cdot r_i \not\leq \rH$,
but that
$I(\rB_1) * \dotsc * I(\rB_\npbody) \cdot r_i \leq \rH$.
If $\npbody > 0$, then from Definition~\ref{def:formula.evaluation} it follows that
$I(\rB_1, \dotsc, \rB_\npbody) = I(\rB_1) * \dotsc * I(\rB_\npbody)$
which is a contradiction.
If $\npbody = 0$, then $\prod\emptyset = 1 = I(\top)$ which is also a contradiction.
The other way around is symmetrical.
\end{proofof}

\subsection{Proof of Proposition~\ref{prop:simplification}}

\begin{proofof}{prop:simplification}
For $(i)$, note that $\top=1$, then
\begin{align*}
I((\fF,\top)^J) &\ = \ I(\fF^J,\top^J)\\
             &\ = \ I(\fF^J) * I(\top)\\
             &\ = \ I(\fF^J) * I(1)\\
             &\ = \ I(\fF^J) * 1\\
             &\ = \ I(\fF^J)
\end{align*}
The remaining cases are analogous.
Just note that $\bot = 0$.
\end{proofof}

\subsection{Proof of Proposition~\ref{prop:nested.reduct.correspondence}}

\begin{proofof}{prop:nested.reduct.correspondence}
Let $\R$ be a rule of the form of
\begin{gather*}
r_i : \ \ \rH \lparrow \rB_1, \dotsc, \rB_\npbody,
				\rB_{\npbody+1}, \dotsc, \rB_{\nnbody}
\end{gather*}
where $\rB_j$ is a positive literal with $1 \leq j \leq \npbody$
and $\rB_j$ is either a negative or a consistent literal with $\npbody + 1 \leq j \leq \nnbody$.
According to Definition~\ref{def:reduct.nested}, 
if $I \models \rB_j$ with $\npbody + 1 \leq j \leq \nnbody$,
then the reduct of rule $\R$ is a rule $\R^I$ of the form
\begin{gather*}
r_i : \ \ \rH \lparrow \rC_1, \dotsc, \rC_\npbody, \top, \dotsc, \top
\end{gather*}
where $\rC_j$ is the reduct of causal literal~$\rB_j$ for $1 \leq j \leq \npbody$.
After applying the simplifications in Proposition~\ref{prop:simplification}, it follows that $\R^I$ becomes
\begin{gather*}
r_i : \ \ \rH \lparrow \rC_1, \dotsc, \rC_\npbody
\end{gather*}
which agrees with Definition~\ref{def:reduct}.
On the other hand, if $I\not\models \rB_j$ for some $\npbody +1 \leq j \leq \nnbody$,
it follows that $B_J^I = \bot$ and, therefore,
\begin{gather*}
(\rB_1, \dotsc, \rB_\npbody,\rB_{\npbody+1}, \dotsc, \rB_{j-1}, \rB_j, \rB_{j+1} \dots, \rB_{\nnbody})^I
	\ \ = \ \ \bot
\end{gather*}
Hence, $\R^I$ is of the from
\begin{gather*}
r_i : \ \ \rH \lparrow \bot
\end{gather*}
and $\R^I$ does not belong to $P^I$ after removing all rules whose body is $\bot$.
Therefore, the reduct according to Definition~\ref{def:reduct.nested} is the same as the Definition~\ref{def:reduct} for programs  with the syntax of Definition~\ref{def:causal.P} and the causal stable models w.r.t. Definitions~\ref{def:causal.smodel} and~\ref{def:causal.smodel.nested} are the same, too.
\end{proofof}

\subsection{Proof of Proposition~\ref{prop:corresponding.of.m-program}}

\begin{lemma}\label{lem:aux:prop:corresponding.of.m-program}
Let $\fF$ be some m-formula and $\fF'$ be is corresponding formula obtained by replacing every m-query $\mquery$ by its corresponding query $\cquery$ given by $\cquery(G,t) = \mquery(G)$.
Then, it holds that $I(\fF) = I(\fF')$ for every interpretation $I$.\qed
\end{lemma}

\begin{proof}
In case that $F = (\cliteral{\mquery}{\rA})$ is a m-literal, by definition
\begin{align*}
I(\cliteral{\mquery}{\rA}) \ \ &= \ \
  \sum\setbm{ G \in \graphs}{G \leq I(\rH) \text{ and } \ \mquery(G) = 1  }
\end{align*}
Furthermore, since $\mquery$ is monotonic, for every $G \leq I(\rH)$ such that $\mquery(G) = 1$, there is some $G'$ such that $G \leq G' \leqmax I(\rH)$ and $\mquery(G) = 1$ and, thus,
\begin{align*}
I(\cliteral{\mquery}{\rA}) \ \ &= \ \
  \sum\setbm{ G \in \graphs}{G \leqmax I(\rH) \text{ and } \ \mquery(G) = 1  }
\end{align*}
Then, since $\cquery(G,t) = \mquery(G)$ for any $t \in \values$, it is clear that
\begin{align*}
I(\cliteral{\mquery}{\rA}) \ \ &= \ \
  \sum\setbm{ G \in \graphs}{G \leqmax I(\rH) \text{ and } \ \cquery(G,I(\rA)) = 1  }
  \\
  \ \ &= \ \ I(\cliteral{\cquery}{\rA})
\end{align*}
In case that $\fF$ is not a m-literal, the proof follows by structural induction assuming as induction hypothesis that $I(\fG)=I(\fG')$ for every subformula $\fG$ of $\fF$.
\end{proof}

\begin{proofof}{prop:corresponding.of.m-program}
Assume that $I$ is a model and $Q$ and suppose that $I$ is not a model of $P$.
Then, there is some rule $\R$ of the form $(r_i : \ \rA \lparrow \fF)$ in $P$ such that $I(\fF) \cdot r_1 \not\leq I(\rA)$.
However, since $\R$ is in $P$ there is a rule $\R'$ of the form $(r_i : \ \rA \lparrow \fF)$ in $Q$ where $\fF'$ is the result of replacing every m-query $\mquery$ by its corresponding query $\cquery$.
Then, from Lemma~\ref{lem:aux:prop:corresponding.of.m-program},
it follows that
$I(\fF') = I(\fF)$ and, thus,
$I(\fF') \cdot r_1 \not\leq I(\rA)$
which is a contradiction with the assumption that $I$ is a model of $Q$.
\\

\noindent
The other way around is symmetrical.
Assume that $I$ is a model and $P$ and suppose that $I$ is not a model of $Q$.
Then, there is some rule $\R'$ of the form $(r_i : \ \rA \lparrow \fF')$ in $Q$ such that $I(\fF') \cdot r_1 \not\leq I(\rA)$.
However, since $\R'$ is in $Q$ there is a rule $\R$ of the form $(r_i : \ \rA \lparrow \fF)$ in $P$ where $\fF'$ is the result of replacing every m-query $\mquery$ by its corresponding query $\cquery$.
Then, from Lemma~\ref{lem:aux:prop:corresponding.of.m-program},
it follows that
$I(\fF') = I(\fF)$ and, thus,
$I(\fF) \cdot r_1 \not\leq I(\rA)$
which is a contradiction with the assumption that $I$ is a model of $P$.
\end{proofof}

\subsection{Proof of Proposition~\ref{prop:corresponding.m-program}}

\begin{lemma}\label{lem:aux:prop:corresponding.m-program}
Let $\fF$ be some formula and $\fF'$ be is corresponding m-formula obtained by replacing every query $\cquery$ by its corresponding m-query $\mquery$ given by $\mquery(G) = \cquery(G,1)$.
Then, it holds that $I(\fF) = I(\fF')$ for every interpretation $I$.\qed
\end{lemma}

\begin{proof}
In case that $F = (\cliteral{\cquery}{\rA})$ is a causal literal, by definition
\begin{align*}
I(\cliteral{\cquery}{\rA}) \ \ &= \ \
  \sum\setbm{ G \in \graphs}{G \leqmax I(\rH) \text{ and } \ \cquery(G,I(\rA)) = 1  }
\end{align*}
Furthermore, since $\cquery$ is monotonic, $\cquery(G,I(\rA)) = \cquery(G,1))$ and, thus
\begin{align*}
I(\cliteral{\cquery}{\rA}) \ \ &= \ \
  \sum\setbm{ G \in \graphs}{G \leqmax I(\rH) \text{ and } \ \cquery(G,1) = 1  }
\end{align*}
Then, since $\mquery(G) = \cquery(G,1)$, it is clear that
\begin{align*}
I(\cliteral{\cquery}{\rA}) \ \ &= \ \
  \sum\setbm{ G \in \graphs}{G \leqmax I(\rH) \text{ and } \ \mquery(G) = 1  }
  \\
  \ \ &= \ \ I(\cliteral{\mquery}{\rA})
\end{align*}
In case that $\fF$ is not a m-literal, the proof follows by structural induction assuming as induction hypothesis that $I(\fG)=I(\fG')$ for every subformula $\fG$ of $\fF$.
\end{proof}

\begin{proofof}{prop:corresponding.m-program}
Assume that $I$ is a model and $Q$ and suppose that $I$ is not a model of $P$.
Then, there is some rule $\R$ of the form $(r_i : \ \rA \lparrow \fF)$ in $P$ such that $I(\fF) \cdot r_1 \not\leq I(\rA)$.
However, since $\R$ is in $P$ there is a rule $\R'$ of the form $(r_i : \ \rA \lparrow \fF)$ in $Q$ where $\fF'$ is the result of replacing every m-query $\mquery$ by its corresponding query $\cquery$.
Then, from Lemma~\ref{lem:aux:prop:corresponding.m-program},
it follows that
$I(\fF') = I(\fF)$ and, thus,
$I(\fF') \cdot r_1 \not\leq I(\rA)$
which is a contradiction with the assumption that $I$ is a model of $Q$.
\\

\noindent
The other way around is symmetrical.
Assume that $I$ is a model and $P$ and suppose that $I$ is not a model of $Q$.
Then, there is some rule $\R'$ of the form $(r_i : \ \rA \lparrow \fF')$ in $Q$ such that $I(\fF') \cdot r_1 \not\leq I(\rA)$.
However, since $\R'$ is in $Q$ there is a rule $\R$ of the form $(r_i : \ \rA \lparrow \fF)$ in $P$ where $\fF'$ is the result of replacing every m-query $\mquery$ by its corresponding query $\cquery$.
Then, from Lemma~\ref{lem:aux:prop:corresponding.m-program},
it follows that
$I(\fF') = I(\fF)$ and, thus,
$I(\fF) \cdot r_1 \not\leq I(\rA)$
which is a contradiction with the assumption that $I$ is a model of $P$.
\end{proofof}

\subsection{Proof of Theorem~\ref{thm:fandino2015.monotonic}}

\begin{proofof}{thm:fandino2015.monotonic}
If $P$ is the corresponding program of some positive m-program~$Q$, the result directly follows from Proposition~\ref{prop:corresponding.of.m-program} plus Proposition~\ref{prop:nested.correspondence} and if $Q$ is the corresponding m-program of some monotonic program $P$, the result directly follows from Proposition~\ref{prop:corresponding.m-program}  plus Proposition~\ref{prop:nested.correspondence}.
\end{proofof}


\subsection{Proof of Corollary~\ref{thm:tp.properties} and~\ref{thm:tp.properties.nested}}

\begin{proofof}{thm:tp.properties.nested}
This is an immediately consequence of Theorem~\ref{thm:m-tp.properties} and Proposition~\ref{prop:corresponding.m-program}.
Just note that the that from Proposition~~\ref{prop:corresponding.m-program} we can translate a program into its corresponding m-program and, then, use the $\tp$ operator for m-programs to compute its least model.
Note also that, from Lemma~\ref{lem:aux:prop:corresponding.m-program}, the $\tp$ operators for m-programs and programs give the same results.
\end{proofof}

\begin{proofof}{thm:tp.properties}
Note that, from Proposition~\ref{prop:nested.correspondence}, the causal stable models of programs w.r.t. Definition~\ref{def:causal.model} and~\ref{def:causal.model.nested} agree and, therefore, the statement directly follows from Corollary~\ref{thm:tp.properties.nested}.
\end{proofof}

\subsection{Proof of Corollary~\ref{thm:positive.two-valued.model.correspondence} and~\ref{thm:positive.two-valued.model.correspondence.nested}}

\begin{proofof}{thm:positive.two-valued.model.correspondence.nested}
This is an immediately consequence of Theorem~\ref{thm:m-positive.two-valued.model.correspondence} and Proposition~\ref{prop:corresponding.m-program}.
Just note that regular programs only contain the query $\cqueryone$ which is monotonic.
\end{proofof}

\begin{proofof}{thm:positive.two-valued.model.correspondence}
Note that, from Proposition~\ref{prop:nested.correspondence}, the causal stable models of programs w.r.t. Definition~\ref{def:causal.model} and~\ref{def:causal.model.nested} agree and, therefore, the statement directly follows from Corollary~\ref{thm:positive.two-valued.model.correspondence.nested}.
\end{proofof}


\subsection{Proof of Theorem~\ref{thm:regluar.standard.correspondnece} and Corollary~\ref{thm:regluar.standard.correspondnece.nested}}

\begin{proofof}{thm:regluar.standard.correspondnece.nested}
Suppose there is some causal stable model $I$ of $P$ which is not a causal stable model of $Q$ and let $P'$ and $Q'$ be the corresponding m-programs of $P^I$ and $Q^I$, respectively.
Then, $I$ is the least model of $P^I$ and, from Proposition~\ref{prop:corresponding.m-program}, $I$ is also the least model of $P^I$.
Just note that regular programs only contain the query $\cqueryone$ which is monotonic.
Since $Q$ is the result of removing all labels in $P$, then $Q^I$ and $Q'$ are the result of removing all labels in $P^I$ and $P'$, respectively.
From, Theorem~\ref{thm:m-regluar.standard.correspondnece}, this implies that $I$ is the least model of $Q'$.
Then, from Proposition~\ref{prop:corresponding.m-program} again, this implies that $I$ is the least model of $Q^I$ which is a contradiction with the assumption that $I$ is not a causal stable model of $Q$.
The other way around is analogous.
\end{proofof}

\begin{proofof}{thm:regluar.standard.correspondnece}
Note that, from Proposition~\ref{prop:nested.correspondence}, the causal stable models of programs w.r.t. Definition~\ref{def:causal.model} and~\ref{def:causal.model.nested} agree and, therefore, the statement directly follows from Corollary~\ref{thm:regluar.standard.correspondnece.nested}.
\end{proofof}

\subsection{Proof of Corollary~\ref{cor:regluar.correspondnece.nested}}

\begin{proofof}{cor:regluar.correspondnece.nested}
Just note that any two programs that only differ in their labels share the same unlabelled version $Q$ and, thus, the proof immediately follows from~Corollary~\ref{thm:regluar.standard.correspondnece}.
\end{proofof}



\subsection{Proof of Proposition~\ref{prop:program.reduct.is.monotonic.nested} }

For any program $P$ and interpretations $I$ and $J$, by $\tpP{P,I}(J)$ we denote an interpretation satisfying
\begin{align*}
\tpP{P,I}(J)(\rH)
    \ &\eqdef \
    \sum \setbm{ \ G \in \causes }{ G \leq \tp(J)(\rH) \text{ and } G \leqmax I(\rH) }
\end{align*}
for every atom~$\rH \in \at$.

\begin{lemma}\label{lem:aux1:prop:program.reduct.is.monotonic}
Let $I$ be the least model of some monotonic program $P$.
Then $I = \tprP{P,I}{\omega}$.\qed
\end{lemma}

\begin{proof}
It is clear that $\tprP{P,I}{\alpha} \leq \tpr{\alpha}$ for every ordinal $\alpha$.
Furthermore, 
from Theorem~\ref{thm:tp.properties},
it follows that
$I = \tpr{\omega}$ and, thus, $\tprP{P,I}{\omega} \leq I$.
Suppose for the sake of contradiction that this inequality is strict, that is, $\tprP{P,I}{\omega} < I$ holds.
Then, there is some atom $\rH$ and causal value $G \leqmax I(\rH)$ such that
$G \not\leq \tprP{P,I}{\alpha}$ for every $\alpha < \omega$.
Since
$I = \tpr{\omega}$
and
$G \leq I(\rH)$, it follows that
there is some $\alpha < \omega$ such that
$G \leq \tpr{\alpha}(\rH)$.
But $G \leq \tpr{\alpha}(\rH)$ and $G \leqmax I(\rH)$ implies that
$G \leq \tprP{P,I}{\alpha}(\rH)$ which is a contradiction.
\end{proof}

\begin{lemma}\label{lem:aux2:prop:program.reduct.is.monotonic}
Let $I$ be the least model of some monotonic program $P$ and $\alpha$ be an ordinal.
Let $\cquery$ be a causal query and let $Q$ be either $P^\cI$ or the result of replacing in $P^I$ the reduced causal query~$\cquery^t$ by its non-reduced form~$\cquery$.
If $\tprP{P,I}{\alpha} \leq \tprP{Q}{\alpha} \leq I$,
then
$\tprP{P,I}{\alpha}(\fF) \leq \tprP{Q}{\alpha}(\fF')$ for every monotonic formula~$\fF$ and $\fF'$ where $\fF'$ is either $\fF^I$ or the result of replacing in $\fF^I$ the reduced causal query~$\cquery^t$ by its non-reduced form~$\cquery$.\qed
\end{lemma}

\begin{proof}
If $\fF = (\cliteral{\cquery'}{\rH})$ is a causal literal and $\cquery' = \cquery$, then $\fF' = \fF$ and the result trivially holds.
Then, assume that $\cquery' \neq \cquery$.
Thus, $G \leqmax \tprP{P,I}{\alpha}(\cliteral{\cquery}{\rH})$
holds only if
\begin{itemize}
\item $G \leqmax \tprP{P,I}{\alpha}(\rH)$, and
\item $\cquery(G,\tprP{P,I}{\alpha}(\rH))=1$.
\end{itemize}
By definition, it follows that
$G \leqmax \tprP{P,I}{\alpha}(\rH)$ 
holds only if
$G \leqmax I(\rH)$.
Furthermore,
by hypothesis, it follow that
$G \leqmax \tprP{P,I}{\alpha}(\rH) \leq \tprP{Q}{\alpha}(\rH)  \leq I(\rH)$.
Then, 
$G \leqmax I(\rH)$
and
$G \leq \tprP{Q}{\alpha}(\rH)  \leq I(\rH)$
imply
\begin{gather}
G \ \ \leqmax \ \ \tprP{Q}{\alpha}(\rH)
  \label{eq:1:lem:aux2:prop:program.reduct.is.monotonic}
\end{gather}
On the other hand,
$G \leqmax \tprP{P,I}{\alpha}(\cliteral{\cquery}{\rH})$
imply that
$\cquery(G,\tprP{P,I}{\alpha}(\rH))=1$ which, since $\cquery$ is monotonic,
implies that
$\cquery(G,I(\rH))=1$.
Then, since $G \leqmax I(\rH)$ and $\cquery(G,I(\rH))=1$,
it follows that
$\cquery^{I(\rH)}(G,u)=1$ for every $u \in \values$.
This plus \eqref{eq:1:lem:aux2:prop:program.reduct.is.monotonic}
imply
$G \leq \tprP{Q}{\alpha}(\cliteral{\cquery^{\cI(\rH)}}{\rH})$.
\\[-5pt]

\noindent
Let us define the rank of a formula such that the rank of a causal literal is $0$ and the rank of any other formula is the greater than the rank of all their subformulas and assume as induction hypothesis that
$\tprP{P,\cI}{\alpha}(\fE) \leq \tprP{Q}{\alpha}(\fE')$ for every monotonic formula~$\fE$ of less rank than $\fF$.
\\[-5pt]

\noindent
In case that $\fF = (\fE,\fH)$,
it follows that
$G \leqmax \tprP{P,\cI}{\alpha}(\fF)$ holds only if there are causal values
$G_1$ and $G_2$ such that
$G_1 \leq \tprP{P,\cI}{\alpha}(\fE)$
and
$G_2 \leq \tprP{P,\cI}{\alpha}(\fH)$
such that
$G \leq G_1 * G_2$.
Since $\fE$ and $\fH$ have less rank than $\fF$, by induction hypothesis,
it follows that
\begin{align}
G_1 \ \ \leq \ \ \tprP{P,I}{\alpha}(\fE) \ \ \leq \ \ \tprP{Q}{\alpha}(\fE')
\\
G_2 \ \ \leq \ \ \tprP{P,I}{\alpha}(\fH) \ \ \leq \ \ \tprP{Q}{\alpha}(\fH')
\end{align}
and, thus, $G \leq G_1 * G_2 \leq \tprP{Q}{\alpha}(\fF')$.
\\[-5pt]

\noindent
Finally, note that the case in which $\fF = (\fE;\fH)$ is analogous and that since $\fF$ is monotonic the case $\fF = \Not \fE$ is not valid.
\end{proof}

\begin{lemma}\label{lem:monotonic.program.least.reduct}
Let $\cI$ be the least model of some monotonic program $P$.
Then, $\cI$ is the least model of program $Q$ where $Q$ is the result of replacing  some causal literal~$(\cliteral{\cquery}{\rA})$ in $Q$ by its reduced form
$(\cliteral{\cquery^I(\rA)}{\rA})$.\qed
\end{lemma}

\begin{proof}
Suppose for the sake of contradiction that $\cI$ is not a model of program~$Q$.
Then, there is a rule $\R' = (r_i : \rH \lparrow \fF')$ is $Q$ such that
$I(\fF') \cdotl r_i \not\leq I(\rH)$
where $\fF'$ is the result of replacing in $\fF$ some causal literal~$(\cliteral{\cquery}{\rA})$ in $Q$ by its reduced form
$(\cliteral{\cquery^I(\rA)}{\rA})$.
Since, from Lemma~\ref{lem:monotonic.formula.prime.reduct.leq},
it follows that
$I(\fF') \leq I(\fF)$
and `$\cdot$' is monotonic,
$I(\fF') \cdotl r_i \not\leq I(\rH)$
implies that
$I(\fF) \cdotl r_i \not\leq I(\rH)$
which is a contradiction with the fact that $\cI$ is a model of $P$ because there is a rule $\R = (r_i : \rH \lparrow \fF)$ in $P$.
\\

\noindent
To show that $\cI$ is the least model of $Q$
assume as induction hypothesis that
$\tprP{P,I}{\beta} \leq \tprP{Q}{\beta}$
for every ordinal $\beta < \alpha$.
Note that, if $\alpha = 0$, then $\tprP{P,I}{0} = \botI$ and, thus, the hypothesis trivially holds.
\\[-5pt]

\noindent
In case that $\alpha$ is a successor ordinal, 
$G \leqmax\tprP{P,I}{\alpha}(\rH)$ holds only if
$G \leqmax I(\rH)$ and there is some rule $\R = (r_i : \rH \lparrow \fF)$ in $P$
and causal value $G' \in \causes$
such that
$G'  \leq \tprP{P,I}{\alpha-1}(\fF)$
and
$G \leq G' \cdotl r_i$.
Furthermore,
by induction hypothesis,
it follows that
$\tprP{P,I}{\alpha-1} \leq \tprP{Q}{\alpha-1}$
and, thus, Lemma~\ref{lem:aux2:prop:program.reduct.is.monotonic} implies that
$\tprP{P,I}{\alpha-1}(\fF) \leq \tprP{Q}{\alpha-1}(\fF')$ for every monotonic formula~$\fF$ and, thus,
$G  \leq \tprP{Q}{\alpha}(\rH)$.
\\

\noindent
In case that $\alpha$ is a limit ordinal,
$G \leq \tprP{P,I}{\alpha}$
implies
$G \leq \tprP{P,I}{\beta}$ for some $\beta < \alpha$
which, by induction hypothesis, implies
$G \leq \tprP{Q}{\beta} \leq \tprP{Q}{\alpha}$.
\\

\noindent
Consequently, $\tprP{P,I}{\alpha} \leq \tprP{P^I}{\alpha}$ for every ordinal $\alpha$.
Furthermore, from Theorem~\ref{thm:tp.properties}, it follows that $\tprP{Q}{\omega}$ is the least model of $Q$
and, from Lemma~\ref{lem:aux1:prop:program.reduct.is.monotonic} and the fact that $I$ is the least model of $P$
it follows that
$I = \tprP{P,I}{\omega}$.
Since $I$ is a a model of $Q$ and $I \leq \tprP{Q}{\omega}$, it follows that $I$ must be the least model of $Q$.
\end{proof}

\begin{proofof}{prop:program.reduct.is.monotonic.nested}
\label{proofof:prop:program.reduct.is.monotonic}
Let $Q$ be the reduct of program $P$ w.r.t. $I$ and Definition~\ref{def:reduct}
and $Q'$ be the reduct of program $P$ w.r.t. $I$ and Definition~\ref{def:reduct.uniform}.
Then, $Q$ is monotonic and, from Lemma~\ref{lem:monotonic.program.least.reduct}, it follows that $I$ is the least model of $Q$ iff
$I$ is the least model of $Q'$.
\end{proofof}


\subsection{Proof of Proposition~\ref{prop:csm.are.models.nested}}

\begin{proofof}{prop:csm.are.models.nested}
Suppose that $I$ is not a model of $P$.
Then there is a rule $\R$ in $P$ of the form of \eqref{eq:rule} such that
$I(\fF) \cdotl r_i \not\leq I(\rH)$.
Since rule $\R$ is in $P$,
rule $\R^I$ of the form
\begin{eqnarray}
r_i: \ \ \rH \ 
    \leftarrow \ \fF^I
\end{eqnarray}
is in $P^I$.
Furthermore, $I(\fF) = I(\fF^I)$ from Proposition~\ref{prop:formula.valuation.reduct.stable}
and, thus,
$I(\fF^I) \cdotl r_i \not\leq I(\rH)$.
That is, $I$ is not a model of $\R^I$ and, consequently, is not a model of $P^I$ which contradicts the assumption that $I$ is a causal stable model of $P$.
\end{proofof}

\begin{lemma}\label{lem:aux:prop:program.monotonic.queries.no-reduct.necesary}
Let $P$ be a program, $\cI$ be an interpretation and $\alpha$ be an ordinal.
Let $Q$ be the result of replacing in $P^\cI$ the reduced causal query~$\cquery^t$ of every monotonic query by its non-reduced form~$\cquery$.
If $\tprP{Q,\cI}{\alpha} \leq \tprP{P^\cI}{\alpha} \leq \cI$,
then
$\tprP{Q,\cI}{\alpha}(\fF') \leq \tprP{P^\cI}{\alpha}(\fF^\cI)$ for every monotonic formulas~$\fF'$ and $\fF^\cI$ where $\fF'$ is the result of replacing in $\fF^I$ the reduced causal query~$\cquery^t$ of every monotonic query by its non-reduced form~$\cquery$.\qed
\end{lemma}

\begin{proof}
If $\fF = (\cliteral{\cquery}{\rH})$ is a causal literal and $\cquery$ is not a monotonic causal query, then $\fF' = \fF^\cI$ and the result trivially holds.
Then, assume that $\cquery$ is a monotonic causal query.
This implies that $G \leqmax \tprP{Q,I}{\alpha}(\cliteral{\cquery}{\rH})$
holds only if
\begin{itemize}
\item $G \leqmax \tprP{Q,\cI}{\alpha}(\rH)$, and
\item $\cquery(G,\tprP{Q,\cI}{\alpha}(\rH))=1$.
\end{itemize}
By definition, it follows that
$G \leqmax \tprP{Q,\cI}{\alpha}(\rH)$ 
holds only if
$G \leqmax I(\rH)$.
Furthermore,
by hypothesis, it follow that
$G \leqmax \tprP{Q,\cI}{\alpha}(\rH) \leq \tprP{P^\cI}{\alpha}(\rH)  \leq I(\rH)$.
Then, 
$G \leqmax \cI(\rH)$
and
$G \leq \tprP{Q,\cI}{\alpha}(\rH)  \leq \cI(\rH)$
imply
\begin{gather}
G \ \ \leqmax \ \ \tprP{P^\cI}{\alpha}(\rH)
  \label{eq:1:lem:aux:prop:program.monotonic.queries.no-reduct.necesary}
\end{gather}
On the other hand,
$G \leqmax \tprP{Q}{\alpha}(\cliteral{\cquery}{\rH})$
imply that
$\cquery(G,\tprP{Q}{\alpha}(\rH))=1$ which, since $\cquery$ is monotonic,
implies that
$\cquery(G,I(\rH))=1$.
Then, since $G \leqmax I(\rH)$ and $\cquery(G,I(\rH))=1$,
it follows that
$\cquery^{I(\rH)}(G,u)=1$ for every causal value $u \in \values$.
This plus \eqref{eq:1:lem:aux:prop:program.monotonic.queries.no-reduct.necesary}
imply that
$G \leq \tprP{P^\cI}{\alpha}(\cliteral{\cquery^{\cI(\rH)}}{\rH}) = \tprP{P^\cI}{\alpha}(\fF^\cI)$.
\\[-5pt]

\noindent
Let us define the rank of a formula such that the rank of a causal literal is $0$ and the rank of any other formula is the greater than the rank of all their subformulas and assume as induction hypothesis that
$\tprP{Q,\cI}{\alpha}(\fE') \leq \tprP{P^\cI}{\alpha}(\fE^\cI)$ for every monotonic formula~$\fE$ of less rank than $\fF$.
\\[-5pt]

\noindent
In case that $\fF = (\fE,\fH)$,
it follows that
$G \leqmax \tprP{Q}{\alpha}(\fF)$ holds only if there are causal values
$G_1$ and $G_2$ such that
$G_1 \leq \tprP{Q}{\alpha}(\fE')$
and
$G_2 \leq \tprP{Q}{\alpha}(\fH')$
such that
$G \leq G_1 * G_2$.
Since $\fE$ and $\fH$ have less rank than $\fF$, by induction hypothesis,
it follows that
\begin{align}
G_1 \ \ \leq \ \ \tprP{Q}{\alpha}(\fE') \ \ \leq \ \ \tprP{Q}{\alpha}(\fE^\cI)
\\
G_2 \ \ \leq \ \ \tprP{Q}{\alpha}(\fH') \ \ \leq \ \ \tprP{Q}{\alpha}(\fH^\cI)
\end{align}
and, thus, $G \leq G_1 * G_2 \leq \tprP{P^\cI}{\alpha}(\fF^\cI)$.
\\[-5pt]

\noindent
The case in which $\fF = (\fE;\fH)$ is analogous.
In case that $\fF = \Not \fE$, by definition
it follows that $\fF' = \bot$ iff $\fF^\cI = \bot$
and $\fF' = \top$ iff $\fF^\cI = \top$
\end{proof}

\subsection{Proof of Proposition~\ref{prop:csm.are.supported.models} and~\ref{prop:csm.are.supported.models.nested}}
\label{proof:prop:csm.are.supported.models.nested}

\begin{lemma}\label{lem:formula.reduct.monotonic.nested}
The reduct $\fF^\cI$ of a formula $\fF$ w.r.t. any interpretation $\cI$ is $\leq$-monotonic, that is, $\cJ(\fF^\cI) \leq \cK(\fF^\cI)$ for all causal interpretations $\cJ$ and $\cK$ such that $\cJ \leq \cK$.\qed
\end{lemma}

\begin{proof}
From Proposition~\ref{prop:cquery.reduct.monotonic}, the reduct of any query $\cquery^{I(\rA)}$ is monotonic.
Furthermore, the reduct of any formula $\fF^I$ is positive.
Hence, $\fF^\cI$ is monotonic and, from Proposition~\ref{prop:formula.monotonic}, it follows that formula $\fF^I$ is $\leq$-monotonic.
\end{proof}

\begin{proofof}{prop:csm.are.supported.models.nested}
From Proposition~\ref{prop:csm.are.models.nested}, any causal stable model $I$ of a program~$P$ is a model of $P$.
Suppose that $I$ is not supported, that is, there is some true atom $\rA$ and cause $G \leq I(\rA)$ sucht that no rule~$\R$ in $P$ of the form of~\eqref{eq:rule} satisfies $G \leq I(\rF) \cdotl r_i$.
Furthermore, from Proposition~\ref{prop:formula.valuation.reduct.stable}, it follows that $I(\rF^I)=I(\rF)$.
That is, no rule $\R^I$ in~$P^I$ satisfies $G \leq I(\rF^I) \cdotl r_i$.

Let $J$ be a causal interpretation such that $J(\rB) = I(\rB)$ for every atom $\rB \neq \rA$ and
$J(\rA) = \sum \setm{ G' \in \causes}{ G' \leq I(\rA) \text{ and } G \not\leq G' }$.
Clearly $J < I$ and, since $I$ is a \mbox{$\leq$-minimal} model of $P^I$, $J$ cannot be a model of $P^I$.
That is, there is a rule $\R^I$ in $P^I$ of the form of~\eqref{eq:rule} such that
$J(\rF^I) \cdotl r_i \not\leq J(\rA)$.
Then there is a cause $G' \leq J(\rF^I) \cdotl r_i$ such that
$G' \not\leq J(\rA)$.
Since $I \leq J$,
it follows that $J(\rF^I) \leq I(\rF^I)$ (Lemma~\ref{lem:formula.reduct.monotonic.nested}) and thus, since application is monotonic, it follows that
$G' \leq I(\rF^I) \cdotl r_i$.
Note that $G' \leq I(\rF^I) \cdotl r_i$,
but no rule in $P^I$ with $\rA$ in the head satisfies
\mbox{$G \leq I(\rF^I) \cdotl r_i$}.
Then
$G \not\leq G'$.
Moreover, since $I\models \R^I$, it follows that $G' \leq I(\rA)$ and then, since $G \not\leq G'$, it follows that $G' \leq J(\rA)$, which is a contradiction with the fact that $G' \not\leq J(\rA)$.
Consequently, $I$ is a supported model of $P$.
\end{proofof}

\begin{proofof}{prop:csm.are.supported.models}
Note that, from Proposition~\ref{prop:nested.correspondence}, the causal stable models of programs w.r.t. Definition~\ref{def:causal.model} and~\ref{def:causal.model.nested} agree and, therefore, the statement directly follows from Proposition~\ref{prop:csm.are.supported.models.nested}.
\end{proofof}

\subsection{Proof of Proposition~\ref{prop:csm.are.minimal.models} and~\ref{prop:csm.are.minimal.models.nested} }
\label{proof:prop:csm.are.minimal.models.nested}

\begin{lemma}\label{lem:models.are.reduct-models.nested}
Let $P$ be a normal program and $I$ and $J$ be two causal interpretation such that $J \leq I$.
If $J$ is a model of $P$, then $J$ is a model of $P^I$.\qed
\end{lemma}

\begin{proof}
Suppose that $J$ is a model of $P$ and not a model of $P^I$.
Then, there is a rule~$\R$ in $P$ of the form of~\eqref{eq:rule} such that
$J(\fF) \cdotl r_i \leq J(\rH)$
and
$J(\fF^I) \cdotl r_i \not\leq J(\rH)$.
Note that, since $P$ is a normal program,
the formula $\fF$ must also be normal.
Then, since $J \leq I$,
Proposition~\ref{prop:normal.formula.reduct.leq.formula} implies that
$J(\fF^I) \leq J(\fF)$.
Furthermore, since application~`$\cdot$' is monotonic,
it follows that
\begin{gather*}
J(\fF^I) \cdotl r_i \ \ \leq \ \ J(\fF) \cdotl r_i \ \ \leq \ \ J(\rH)
\end{gather*}
which is a contradiction with the fact that 
$J(\fF^I) \cdotl r_i \not\leq J(\rH)$.
\end{proof}

\proofsep

\begin{proofof}{prop:csm.are.minimal.models.nested}
If $I$ is a causal stable model of $P$,
then, Proposition~\ref{prop:csm.are.models.nested} implies that $I$ is a model of $P$.
Suppose that $I$ is not $\leq$-minimal.
Then there exists an interpretation $J \leq I$ such that $J$ is a model of $P$.
But, since $P$ is a normal program, from Lemma~\ref{lem:models.are.reduct-models.nested}, $J$ must be a model of $P^I$ and, thus, $I$ is not a $\leq$-minimal model of $P^I$ which contradicts the assumption that $I$ is a causal stable model of $P$.
\end{proofof}

\begin{proofof}{prop:csm.are.minimal.models}
Note that, from Proposition~\ref{prop:nested.correspondence}, the causal stable models of programs w.r.t. Definition~\ref{def:causal.model} and~\ref{def:causal.model.nested} agree and, therefore, the statement directly follows from Proposition~\ref{prop:csm.are.minimal.models.nested}.
\end{proofof}

\subsection{Proof of Theorem~\ref{thm:splitting} and~\ref{thm:splitting.nested}}

\begin{definition}
A \emph{splitting} of a program $P$ is a pair $\tuple{P_b,P_t}$ of pairwise disjoint sets such that
$P = (P_b \cup P_t)$
and no atom occurring in the head of a rule in $P_t$ occurs in the body of a rule in $P_b$.
A splitting is said to be \emph{strict} if, in addition, no atom occurring in the head of a rule in $P_t$ occurs (the head of a rule) in $P_b$.\qed
\end{definition}


\begin{lemma}\label{lem:splitting.aux}
Let $P_b$ and $P_t$ be two monotonic programs such that no atom occurring in a body in $P_b$ is a head atom of $P_t$.
Let $I$ and $J$ be the least models of $(P_b \cup P_t)$ and $P_b$, respectively.
Then, $I$ is also the least model of program $(J \cup P_t)$.
Furthermore, 
$\restr{J}{S} = \restr{I}{S}$ where $S$ is the set of atoms of all atoms not in the head of any rule in~$P_t$.\qed
\end{lemma}

\begin{proof}
Since interpretation $J$ is the least model of the program $J$ and $J \leq I$, it follows that $I$ satisfies all rules in program~$J$.
In addition, since $I$ is the least model of program~$(P_b \cup P_t)$,
 it is clear that $I$ also satisfies all rules in $P_t$
and, thus, $I$ satisfies all rules in program~$(J \cup P_t)$.
Suppose that $I$ is not the least model of $(I \cup P_t)$.
Then, there is a model $I'$ of
$(J \cup P_t)$ such that $I' < I$.
Since $I$ is the least model of program~$(P_b \cup P_t)$ and $I' < I$, it follow that $I'$ does not satisfy some rule \mbox{$\R = (r_i : \ \rH \leftarrow \fF)$ in $(P_b \cup P_t)$}.
That is,
$I'(\fF)  \cdot r_i  \ \not\leq \ I'(\rH)$.
Since $I'$ is a model of $(J\cup P_t)$, it is clear that $J \leq I'$ and, since in addition $I' < I$, it follows that
$I(\fF)  \cdot r_i  \ \not\leq \ J(\rH)$ also holds.
Furthermore, $I'$ satisfy all rules in $P_t$ because $I'$ is a model of $(J \cup P_t)$ and, thus, rule $\R$ must be in~$P_b$ and no atom occurring in~$\fF$ occurs in the head of a rule in $P_t$.
Hence,
$I(\fF) = J(\fF)$
and, thus,
$I(\fF)  \cdot r_i  \, \not\leq \, J(\rH)$
implies that
$J(\fF)  \cdot r_i  \, \not\leq \, J(\rH)$
which is a contradiction with the hypothesis that $J$ is a model of $P_b$ and the fact that $\R$ in $P_b$.
Consequently, $I$ is also the least model of program $(J\cup P_t)$.
Furthermore, since $I$ is the least model of program $(J \cup P_t)$ and no atom in $S$ occurs in the head of any rule in $P_t$, it follows that $\restr{I}{S} = \restr{J}{S}$.
\end{proof}

\begin{proofof}{thm:splitting.nested}
For the only if direction.
Assume that $I$ is a causal stable model of program $(P_b \cup P_t)$.
Then, $I$ is the least model of the monotonic program
$(P_b \cup P_t)^I = (P_b^I \cup P_t^I)$.
Let $J$ be the least model of $P_b^I$.
Since $I$ and~$J$ respectively are the least models of $(P_b^I \cup P_t^I)$ and $P_b^I$ and no atom occurring in a body in $P_b^I$ is in the head of any rule in $P_t^I$,
from Lemma~\ref{lem:splitting.aux}, it follows that $I$ is the least model of program~$(J \cup P_t^I) = (J \cup P_t)^I$ and, consequently, $I$ is a causal stable model of~$(J \cup P_t)$ and
$\restr{I}{S}=\restr{J}{S}$ where $S$ is the set of atoms of all atoms not occurring in the head of any rule in~$P_t$.
In addition, since $\restr{I}{S}=\restr{J}{S}$ and all atoms in the body of some rule in $P_b$ are in~$S$, it follows that $P_b^I = P_b^J$ and, therefore, $J$ is the least model of $P_b^I = P_b^J$ and a causal stable model of $P_b$.
Furthermore, if no atom occurring in $P_b$ occurs in the head of a rule in $P_t$,
then $\restr{J}{S} = J$
(note that $S$ contains all atoms in $P_b$ since no atom occurring in $P_b$ occurs in the head of a rule in $P_t$)
and, thus,
$\restr{I}{S}=J$.
\\[-5pt]

\noindent
The other way around.
If $I$ is a causal stable model of~$(J \cup P_t)$,
then $I$ is the least model of~$(J \cup P_t)^I = (J \cup P_t^I)$.
Let $S$ be the of all atoms not occurring in the head of a rule in $P_t$.
Then, $S$ contains all atoms occurring in the body of the rules in $P_b$ and,
since $I$ is the least model of~$(J \cup P_t^I)$,
it follows that
$\restr{I}{S} = \restr{J}{S}$ and, thus, $P_b^I = P_b^{J}$.
Then, since $J$ is a causal stable model of $P_b$,
it follows that
$J$ is the least model of $P_b^I$.
From Lemma~\ref{lem:splitting.aux}, this implies that
$I$ is the least model of program~$(P_b^I \cup P_t^I)=(P_b \cup P_t)^I = P^I$ and, thus, $I$ is a causal stable model of~$P$.
\end{proofof}

\begin{proofof}{thm:splitting}
Note that, from Proposition~\ref{prop:nested.correspondence}, the causal stable models of programs w.r.t. Definition~\ref{def:causal.model} and~\ref{def:causal.model.nested} agree and, therefore, the statement directly follows from Theorem~\ref{thm:splitting.nested}.
\end{proofof}

\subsection{Proof of Theorem~\ref{thm:infinity.splitting.nested}}

\begin{lemma}\label{lem:infinity.splitting.aux}
Let $(P_\alpha)_{\alpha < \mu}$ a splitting sequence of some monotonic program $P$.
Then, there is a unique solution
$(I_\alpha)_{\alpha < \mu}$ of $(P_\alpha)_{\alpha < \mu}$
and it satisfies \ (i) \ $I = \sum_{\alpha < \mu} I_\alpha$ and \ (ii) \
$\restr{I_\alpha\,}{S_\alpha} = \restr{I}{S_\alpha}$
where $I$ is the least model of $P$ and $S_\alpha$ is the set of all atoms not occurring in the head of any rule in~$\bigcup_{\alpha < \beta < \mu } P_\beta$.\qed
\end{lemma}

\begin{proof}
First note that, since $P$ is a monotonic program, every $P_\alpha$ with $\alpha < \mu$ is also monotonic and, thus, there is a unique causal stable model $I_0$ of $P_0$.
Suppose that there is a solution
$(I_\alpha')_{\alpha < \mu}$ of $(P_\alpha)_{\alpha < \mu}$
such that $I_\alpha' \neq I_\alpha$ for some $\alpha < \mu$.
Let $\alpha$ be the first ordinal such that $I_\alpha' \neq I_\alpha$.
Then, $0 < \alpha < \mu$ and there are two different causal stable models $I_\alpha$ and $I_\alpha'$ of $(J_\alpha \cup P_\alpha)$ which is a contradiction with the fact that $(J_\alpha \cup P_\alpha)$ is monotonic.
\\[-5pt]

\noindent
Let $I = \sum_{\alpha < \mu} I _\alpha$ and we will show that $I$ is the least model of $P$ and that $I_\alpha = \restr{I}{S_\alpha}$.
Assume as induction hypothesis that the lemma statement holds for every ordinal $\mu' < \mu$ and note that, in case that $\mu = 0$, it follows that $P = \bigcup_{\alpha < 0 } P_\alpha = \emptyset$ and that $I = \sum_{\alpha < 0} I_\alpha = \botI$ and that $\botI$ is the least model of the empty program.
\\[-5pt]

\noindent
In case that $\mu$ is a successor ordinal,
let $\mu'  =  \mu - 1$ be its predecessor,
let $Q = \bigcup_{\alpha < \mu'} P_\alpha$
and $J$ be the least model of $Q$.
Then,
$(I_\alpha)_{\alpha < \mu'}$ is solution of $(P_\alpha)_{\alpha < \mu'}$,
$\tuple{Q,P_{\mu'}}$ is a splitting of $P$,
and, by induction hypothesis
$J = \sum_{\alpha < \mu'} I_\alpha$
and
$\restr{I_\alpha\,}{S_\alpha} = \restr{J}{S_\alpha}$ for every $\alpha < \mu'$.

Let $I_{\mu'}$ be the least model of $(J \cup P_{\mu'})$.
Since $I_{\mu'}$ is the least model of $(J \cup P_{\mu'})$,
it follows that $I_{\mu'} \geq J$ and, thus,
$I  = \sum_{\alpha < \mu} I _\alpha = I_{\mu'} + \sum_{\alpha < \mu'} I _\alpha = I_{\mu'} + J = I_{\mu'}$.
That is, $I = I_{\mu'}$ is the least model of $(J \cup P_{\mu'})$
and, since $J$ is the least model of $Q$,
from Lemma~\ref{lem:splitting.aux},
it follows that $I$ is the least model of $P = (Q \cup P_{\mu'})$
and, that, $\restr{ I_{\mu'}\, }{S_{\mu'}} = \restr{I}{S_{\mu'}}$.
Furthermore, since no atom in $S_\alpha$ with $\alpha < \mu'$ occurs in the head of any rule in $P_{\mu'}$ it follows that
$\restr{I\,}{S_\alpha} = \restr{J}{S_\alpha}$
for every $\alpha < \mu'$.
Consequently
$\restr{I_\alpha\,}{S_\alpha} = \restr{I}{S_\alpha}$ for every $\alpha < \mu$.
\\[-5pt]

\noindent
In case that $\mu$ is a limit ordinal, by induction hypothesis
$\restr{I_\alpha\,}{S_\alpha} = \restr{I}{S_\alpha}$ for every $\alpha < \mu'$
and, thus, since
all atoms occurring in the body of any rule in $P_\alpha$ belong to $S_\alpha$,
it follows that $P_\alpha^{I_\alpha} = P_\alpha^I$.
Furthermore, since $(I_\alpha)_{\alpha < \mu}$ is solution of $(P_\alpha)_{\alpha < \mu}$,
it follows that $I_\alpha$ is the least model of $(J_\alpha \cup P_\alpha)$ and, thus, $I_\alpha$ is a model of $P_\alpha^{I_\alpha} = P_\alpha^I$.
Since $I = \sum_{\alpha < \mu} I_\alpha \geq I_\alpha$,
then $I$ is a model of $P_\alpha^I$ for every $\alpha < \mu'$
and, consequently, $I$ is a model of $P^{I}$.

Suppose that $I$ is not the least model of $P$.
Then, there is a model $I'$ of
$P$ such that $I' < I$.
Since $I = \sum_{\alpha < \mu} I_\alpha$ and $I' < I$, it follows that $I_\alpha \not\leq I'$ for some first ordinal $\alpha < \mu$.
Since $\alpha$ is the first ordinal such that
$I_\alpha \not\leq I'$,
it follows that $J_\alpha = \sum_{\beta < \alpha} I_\beta \leq I'$
and, thus,
$I'$ satisfies all rules in $J_\alpha$.
Furthermore, since $P_\alpha \subseteq P$ and $I'$ is model of $P$,
it follows that
$I'$ also satisfies all rules in $P_\alpha$.
That is, $I'$ is a model of $(J_\alpha \cup P_\alpha)$ and $I_\alpha \not\leq I$ which is a contradiction with the fact that $I_\alpha$ is the least model of $(J_\alpha \cup P_\alpha)$.
Consequently, $I$ is the least model of $P$.

Suppose now that $\restr{I_\alpha\,}{S_\alpha} \neq \restr{I}{S_\alpha}$ for some $\alpha < \mu$ and let $\alpha$ be the first such ordinal.
Then, there is some first ordinal $\alpha'$ and atom $\rH \in S_\alpha$ such that
$I_{\alpha}(\rH) \not\leq I_{\alpha'}(\rH)$.
Note that $\alpha' \leq \alpha$ implies that $I_{\alpha'} \leq I_{\alpha}$ and, thus,
it must be that $\alpha < \alpha'$.
Since $\alpha'$ first ordinal that satisfies $I_{\alpha}(\rH) \not\leq I_{\alpha'}(\rH)$
it follows that
$I_\beta(\rH) \leq I_\alpha(\rH)$ for every $\beta < \alpha'$
and, thus,
$J_{\alpha'}(\rH) \leq I_\alpha(\rH)$.
Since $J_{\alpha'}(\rH)  \leq I_{\alpha}(\rH) \not\leq I_{\alpha'}(\rH)$
and $I_{\alpha'}$ is the least model of $(J_{\alpha'} \cup P_{\alpha'})$,
there must be some rule $\R = (r_i : \rH \lparrow \fF) \in P_{\alpha'}$
which is a contradiction with the fact that $\rH \in S_{\alpha}$ and $\alpha < \alpha'$.
Consequently,
$\restr{I_\alpha\,}{S_\alpha} = \restr{I}{S_\alpha}$ for all $\alpha < \mu$.
\end{proof}

\begin{proofof}{thm:infinity.splitting.nested}
For the only if direction.
Assume that $I$ is a causal stable model of $P$.
Then, $I$ is the least model of the monotonic program $P^I$
and, from Lemma~\ref{lem:infinity.splitting.aux}
there is a unique solution $(I_\alpha)_{\alpha < \mu}$ of program $P^I$
and it satisfies \ (i) \ $I = \sum_{\alpha < \mu} I_\alpha$ and \ (ii) \
$\restr{I_\alpha\,}{S_\alpha} = \restr{I}{S_\alpha}$.
Furthermore, by definition,
\begin{enumerate_thm}
\item $I_0$ is the least model of $P_0^{I}$,
\item $I_{\alpha}$ is a stable model of $(J_{\alpha} \cup P_{\alpha}^{I})$ for any ordinal $0 < \alpha < \mu$ where $J_\alpha = \sum_{\beta < \alpha} I_\beta$.
\end{enumerate_thm}
Since $\restr{I_\alpha\,}{S_\alpha} = \restr{I}{S_\alpha}$ and all atoms occurring in the body of any rule in $P_\alpha$ belong to $S_\alpha$, it follows that $P_\alpha^I = P_\alpha^{I_{\alpha}}$ and, thus,
\begin{enumerate_thm}
\item $I_0$ is the least model of $P_0^{I_\alpha}$,
\item $I_{\alpha}$ is a stable model of $(J_{\alpha} \cup P_{\alpha})^{I_\alpha}=(J_{\alpha} \cup P_{\alpha}^{I_\alpha})$ for any ordinal $0 < \alpha < \mu$ where $J_\alpha = \sum_{\beta < \alpha} I_\beta$.
\end{enumerate_thm}
Consequently, $(I_\alpha)_{\alpha < \mu}$ is a solution of $(P_\alpha)_{\alpha < \mu}$
and it satisfies $I = \sum_{\alpha < \mu} I_\alpha$
and
$\restr{I_\alpha\,}{S_\alpha} = \restr{I}{S_\alpha}$.
\\[-5pt]

\noindent
The other way around.
Assume there is some solution $(I_\alpha)_{\alpha < \mu}$ of $(P_\alpha)_{\alpha < \mu}$
and let $I = \sum_{\alpha < \mu} I_\alpha$.
By definition,
\begin{enumerate_thm}
\item $I_0$ is the least model of $P_0^{I_0}$,
\item $I_{\alpha}$ is the least model of $(J_{\alpha} \cup P_{\alpha}^{I_\alpha})$ for any ordinal $0 < \alpha < \mu$ where $J_\alpha = \sum_{\beta < \alpha} I_\beta$.
\end{enumerate_thm}
Since $S_{\alpha}$ contains all atoms not in the head of any rule in
$\bigcup_{\alpha < \beta < \mu} P_\beta$,
it follows that
\begin{gather*}
\sum_{\beta < \alpha} \restr{I_\beta\,}{S_\alpha}
  \ = \ \restr{J_\alpha\,}{S_\alpha} 
  \ \ \leq \ \ \restr{I_\alpha\,}{S_\alpha}
  \ = \ \restr{J_{\alpha+1}\,}{S_\alpha}
  \ = \ \restr{I_{\alpha+1}\,}{S_\alpha}
  \ = \ \dotsc
  \ = \ \sum_{\beta < \mu} \restr{I_\beta\,}{S_\alpha}
  \ = \ \restr{I}{S_\alpha}
\end{gather*}
and, since $S_\alpha$ contains all atoms occurring in the body of all rules in $P_\alpha$,
it follows that
$P_\alpha^I = P_\alpha^{I_{\alpha}}$ and, thus,
\begin{enumerate_thm}
\item $I_0$ is the least model of $P_0^{I}$,
\item $I_{\alpha}$ is the least model of $(J_{\alpha} \cup P_{\alpha}^{I})$ for any ordinal $0 < \alpha < \mu$ where $J_\alpha = \sum_{\beta < \alpha} I_\beta$.
\end{enumerate_thm}
Hence, $(I_\alpha)_{\alpha < \mu}$ of $(P_\alpha^I)_{\alpha < \mu}$
and, from Lemma~\ref{lem:infinity.splitting.aux}, it follows that $I$ is the least model of $P^I$ and a causal stable model of $P$.
\\[-5pt]

\noindent
Furthermore, if $(I_\alpha)_{\alpha < \mu}$ is a strict solution in $\alpha$, then no atom occurring in $P_\alpha$ occurs in the head of a rule in any $P_\beta$ with $\alpha < \beta < \mu$, and, thus, every atom occurring in $(J_\alpha \cup P_\alpha)$ belongs to~$S_\alpha$.
Consequently, $I_\alpha = \restr{I_\alpha\,}{S_\alpha} = \restr{I}{S_\alpha}$.
\end{proofof}


\subsection{Proof of Proposition~\ref{prop:stratified} and~\ref{prop:stratified.nested}}

\begin{proofof}{prop:stratified.nested}
Let $P_{\alpha+1}$ be the set of rules of the form of~\eqref{eq:rule.nested} such that $\lambda(\rH)=\alpha$ and $P_\alpha=\emptyset$ if $\alpha$ is a limit ordinal.
Then, $(P_\alpha)_{\alpha < \mu}$ is a strict splitting sequence of $P$ and, from Theorem~\ref{thm:infinity.splitting.nested}, an interpretation $I$ is a causal stable model of~$P$ iff there is some solution $( \restr{I}{S_\alpha})_{\alpha < \mu}$ of $(P_\alpha)_{\alpha < \mu}$
such that $I = \sum_{\alpha < \mu}  \restr{I}{S_\alpha}$.
where $S_\alpha$ is the set of all atoms not occurring in the head of any rule in~$\bigcup_{\alpha < \beta < \mu } P_\beta$.
Hence, it is enough to show that every $P_\alpha$ has a unique causal stable model.

By definition, it is clear that $P_\alpha$ has the $\botI$ interpretation as its unique causal stable model when $\alpha$ is a limit ordinal.
In case that $\alpha$ is a successor ordinal, suppose that there are two different causal stable models $J$ and $J'$ of $P_\alpha$.
Since $P$ is stratified, there is no rule in~$\bigcup_{\alpha - 1 < \beta < \mu } P_\beta$ with an atom occurring in $P_{\alpha}$ under the scope of negation or a non-monotonic causal literal in $P_{\alpha}$.
Hence, $J(\rB)=J'(\rB)=\restr{I}{S_{\alpha-1}}(\rB)$ for every atom $\rB$ occurring under the scope of negation or a non-monotonic causal literal and, thus, $P^J=P^{J'}$ and $J$ and $J'$ must be equal which is a contradiction with the assumption.
\end{proofof}

\begin{proofof}{prop:stratified}
Note that, from Proposition~\ref{prop:nested.correspondence}, the causal stable models of programs w.r.t. Definition~\ref{def:causal.model} and~\ref{def:causal.model.nested} agree and, therefore, the statement directly follows from Proposition~\ref{prop:stratified.nested}.
Just note that, according to Definition~\ref{def:causal.P}, $\bot$ is not allowed in the head of the rules.
\end{proofof}

\subsection{Proof of Proposition~\ref{prop:strongly.equivalent}}

\begin{proofof}{prop:strongly.equivalent}
Let $R$ be any causal program over the signature~$\sigma$ of $P$ and $Q$.
Let $\mathcal{I}$, $\mathcal{J}$ respectively be the sets of causal stable models of program $P \cup R$ and $Q \cup R$.
Any causal stable model $I \in \mathcal{I}$ is  the least model of the positive program $(P\cup R)^I = P^I \cup R^I$.
That is, $I$ satisfies all rules in both $P^I$ and $R^I$ and,
since $P \Leftrightarrow Q$, $I$ satisfies all rules in $Q^I$.
Suppose there exists $I'$ which satisfies all rules in $(Q\cup R)^I$ and $I' < I$.
By the same reasoning $I'$ satisfies all rules in $P^I$ (an also in $R^I$) contradicting the assumption that $I$ is the lest model of $(P \cup R)^I$.
Hence, $I$ is the least model of $(P \cup R)^I$, and so, an stable model of $(P \cup R)$.
That is, $I \in \mathcal{J}$.
The other way around is analogous.
\end{proofof}

\subsection{Proof of Proposition~\ref{prop:program.replacmentet}}

\begin{lemma}\label{lem:formula.replacmentet}
Let $F$, $G$ and $H$ be formulas such that $F \Leftrightarrow G$. If a formula $H'$ is obtained from $H$ by replacing some regular occurrences of $F$ by $G$, then $H \Leftrightarrow H'$.\qed
\end{lemma}
\begin{proof}
By structural induction like Lemma~4 in~\cite{lifschitzTT99nested}.
If $H$ is elementary. Then either $H = F$ and $H' = G$ or $H = H'$.
In both cases 
$H \Leftrightarrow H'$.
Otherwise, if $H = F$ and $H' = G$, then also $H \Leftrightarrow H'$.
Hence, in the following we assume that $H \neq F$.
\begin{enumerate}
\item In case $H= H_1, H_2$, then $H' = H_1', H_2'$ and, by induction hypothesis,
$H_i \Leftrightarrow H_i'$ with $i\in\set{1,2}$.
Then
\begin{align*}
I(H^J) &\ = \ I((H_1,H_2)^J)\\
       &\ = \ I(H_1^J,H_2^J)\\
       &\ = \ I(H_1^J) * I(H_2^J)\\
       &\ = \ I((H_1')^J) * I((H_2')^J)\\
       &\ = \ I((H_1',H_2')^J)\\
       &\ = \ I((H')^J)
\end{align*}

\item The case $H= H_1; H_2$ is similar to the previous one.

\item In case $H =  \Not H_1$, then $H' = \Not H_1'$ and, by induction hypothesis,
\mbox{$H_1 \Leftrightarrow H_1'$}.
\begin{align*}
I(H^J) = 1
       &\text{ \ iff \ } I((\Not H_1)^J) = 1\\
       &\text{ \ iff \ } J(H_1^J) = 0\\
       &\text{ \ iff \ } J((H_1')^J) = 0\\
       &\text{ \ iff \ } I((\Not H_1')^J) = 1\\
       &\text{ \ iff \ } I((H')^J) = 1
\end{align*}
and $I(H^J) = 0$ otherwise, that is, iff $I((H')^J) = 0$
\qed
\end{enumerate}
\end{proof}

\begin{proofof}{prop:program.replacmentet}
Similar to the proof of Proposition~3 in~\cite{lifschitzTT99nested}.
Let $Q$ be the program obtained by replacing some occurrences of $F$ by $G$ in $P$.
Assume that $I$ is satisfies all rules in~$Q^J$.
Take any rule \mbox{$(r_i : \ \rH \leftarrow \fF)$} in $P$.
Its corresponding rule \mbox{$(r_i : \ \rH \leftarrow \fG)$} in~$Q$ must satisfy
\begin{gather*}
I(\fG^J) \cdot r_i \leq I(\rA)
\end{gather*}
and, by Lemma~\ref{lem:formula.replacmentet}, it follows that $I(\fF^J) = I(\fG^J)$.
Consequently,
\begin{gather*}
I(\fF^J) \cdot r_i \leq I(\rH)
\end{gather*}
Hence, $I$ satisfies all rules in $P$.
The other way around is similar.
Hence, $I$ satisfies all rules in~$P^J$ iff $I$ satisfies all rules in $Q^J$.
That is $P \Leftrightarrow Q$ and, by Proposition~\ref{prop:strongly.equivalent}, $P$ and $Q$ are strongly equivalent.
\end{proofof}

\subsection{Proof of Proposition~\ref{prop:transformations}}

\begin{proofof}{prop:transformations}
For $(i)$ note that
\begin{align*}
I((\fF,\fG)^J) &\ = \ I(\fF^J,\fG^J)\\
           &\ = \ I(\fF^J) * I(\fG^J)\\
           &\ = \ I(\fG^J) * I(\fF^J)\\
           &\ = \ I((\fG,\fF)^J)
\end{align*}
Similarly, $I((\fF;\fG)^J) \ = \ I((\fG;\fF)^J)$. Note that product and \review{R2.4}{addition} are both commutative.
The same reasoning applies for $(ii)$ and $(iii)$ by noting that product and \review{R2.4}{addition} are also associative and distributes over one over the other.

For $(iv)$,
\begin{align*}
I((\Not\Not\Not \fF)^J) = 1
  &\text{ iff } J((\Not\!\Not \fF)^J) = 0\\
  &\text{ iff } J((\Not \fF)^J) = 1\\
  &\text{ iff } J((\Not \fF)^J) = 1\\
  &\text{ iff } J(F^J) = 0\\
  &\text{ iff } I((\Not \fF)^J) = 1
\end{align*}
and $I((\Not\Not\Not \fF)^J) = 0$ otherwise, that is $I((\Not \fF)^J) = 0$.

Similarly, for $(v)$,
\begin{align*}
I((\Not (\fF;\fG))^J) = 1
  &\text{ iff } J(\fF^J;\fG^J) = 0\\
  &\text{ iff } J(\fF^J) + J(\fG^J) = 0\\
  &\text{ iff } J(\fF^J) = 0 \text{ and } J(\fG^J) = 0\\
  &\text{ iff } I((\Not \fF)^J) = 1 \text{ and } I((\Not \fG)^J) = 1\\
  &\text{ iff } I((\Not \fF)^J) * I((\Not \fG)^J) = 1\\
  &\text{ iff } I((\Not \fF,\Not \fG)^J) = 1
\end{align*}
and $I(\Not (\fF;\fG)^J) = 0$ otherwise, that is $I((\Not \fF,\Not \fG)^J) = 0$.
Furthermore
\begin{align*}
I((\Not (\fF,\fG))^J) = 1
  &\text{ iff } J((\fF,\fG)^J) = 0\\
  &\text{ iff } J(\fF^J) * J(\fG^J) = 0\\
  &\text{ iff } J(\fF^J) = 0 \text{ or } J(\fG^J) = 0\\
  &\text{ iff } I((\Not \fF)^J) = 1 \text{ or } I((\Not \fG)^J) = 1\\
  &\text{ iff } I((\Not \fF)^J) + I((\Not \fG)^J) = 1\\
  &\text{ iff } I((\Not \fF;\Not \fG)^J) = 1
\end{align*}
and $I((\Not (\fF,\fG))^J) = 0$ otherwise, that is $I((\Not \fF;\Not \fG)^J) = 0$.
$(vi)$ and $(vii)$ directly follows from Proposition~\ref{prop:simplification}.
Finally, for $(viii)$, $I(\Not \top) = 0 = \bot$ and $I(\Not \bot) = 1 = \top$. 
\end{proofof}

\subsection{Proof of Proposition~\ref{prop:formula.normal.form}}

\begin{proofof}{prop:formula.normal.form}
The proof follows by structural induction using Proposition~\ref{prop:transformations} and~Lemma~\ref{lem:formula.replacmentet} exactly as in~\cite{lifschitzTT99nested}.
Note that we do not consider strong negation, so all formulas are regular.
\end{proofof}

\subsection{Proof of Proposition~\ref{prop:body.disjuntion.elimination}}

\begin{proofof}{prop:body.disjuntion.elimination}
Note that $I \models (r_i : \ \rA \leftarrow \fF; \fG)^J$ iff
$I \models (r_i : \ \rA \leftarrow \fF^J; \fG^J)$
\begin{gather*}
\big( I(\fF^J) + I(\fG^J) \big) \cdot r_i \leq I(\rA)
\end{gather*}
which, by application distributivity over \review{R2.4}{addition}, is equivalent to
\begin{gather*}
I(F^J) \cdotl r_i \ + \ I(G^J) \cdot r_i \leq I(\rA)
\end{gather*}
which in turn holds iff $I \models ( r_i : \ \rA \ \leftarrow \fF)^J$
and
$I \models ( r_i : \ \rA \ \leftarrow \fG)^J$.
\end{proofof}

\subsection{Proof of Proposition~\ref{prop:program.normal.form.strong}}

\begin{proofof}{prop:program.normal.form.strong}
Propositions~\ref{prop:strongly.equivalent},~\ref{prop:program.replacmentet} and~\ref{prop:formula.normal.form} show that any program is strongly equivalent to a set of rules of the form
\begin{gather}
r_i : \ \rA \ \leftarrow \
    \fF_1;\dotsc;\fF_\npbody,
\end{gather}
where each $F_i$ is a simple conjunction.
Similarly, Propositions~\ref{prop:strongly.equivalent},~\ref{prop:program.replacmentet} and~\ref{prop:body.disjuntion.elimination} show that such set of rules is strongly equivalent to a set of rules of the form
\begin{gather}
r_i : \ \rA \ \leftarrow \ \fF
\end{gather}
where each $F$ is a simply conjunction.
That is, a set of rule of the form \eqref{eq:rule} in which the head can be $\bot$.
\end{proofof}

\subsection{Proof of Proposition~\ref{prop:program.normal.form}}

\begin{proofof}{prop:program.normal.form}
From Proposition~\ref{prop:program.normal.form}, every program can be writing as an equivalent program where all rules $\R$ are of the form
\begin{gather}
r_i : \ \rH \ \leftarrow \ \rB_1, \dotsc, \rB_\npbody
\end{gather}
where $\rH$ is an atom or $\bot$. If $\rH$ is an atom, then $\R$ is already of the form of~\eqref{eq:rule}.
Otherwise, replace rule~$\R$ by a rule $\R'$ of the form of
\begin{gather}
r_i : \ aux_\R \ \leftarrow \ \rB_1, \dotsc, \rB_\npbody, \Not aux_\R
\end{gather}
where $aux_\R$ is a new auxiliary predicate.
Let $Q$ be the result of replacing $\R$ by $\R'$ in $P$.
If $I$ is a causal stable model of $P$, then $I \not\models \rB_j$ for some $1 \leq j \leq \npbody$ and, thus, it is a causal stable model of $Q$.
The other way around, if $I$ is a causal stable model of $Q$, either $I \models aux_\R$ or $I \not\models \rB_j$ for some $1 \leq j \leq \npbody$.
If the former, rule $\R'$ does not belong to $Q^I$ and, thus, there is no rule which $aux_\R$ which contradicts the fact that $I$ must be the least model of $Q^J$.
Hence, $I\not\models aux_\R$ and $I \not\models \rB_j$ for some $1 \leq j \leq \npbody$ and, therefore, $I$ is a causal stable model of $P$.
\end{proofof}

\section{Complexity assessment}
\label{sec:complexity}

First, it has been showed in~\cite{CabalarFF14Jelia} that there may an exponential number of causes for some atom with respect to a casual stable model.
For instance, consider the positive program~$\newprogram\label{prog:exp}$ consisting of following the rules:
\begin{gather*}
\begin{IEEEeqnarraybox}{lCl}
a		&:& p_{1}
\\
c		&:& q_{1}
\end{IEEEeqnarraybox}
\hspace{2cm}
\begin{IEEEeqnarraybox}{lCl}
b		&:& p_{1}
\\
d		&:& q_{1}
\end{IEEEeqnarraybox}
\hspace{2cm}
\begin{IEEEeqnarraybox}{lClCl"l}
m_{i} &:& p_{i}		& \leftarrow &	p_{i-1},\ q_{i-1}
						&	\text{for }  i\in\set{2,\dotsc,n}
\\
n_{i} &:& q_{i}		& \leftarrow &	p_{i-1},\ q_{i-1}
						&	\text{for }  i\in\set{2,\dotsc,n}
\end{IEEEeqnarraybox}
\end{gather*}
Since program~$\programref{prog:exp}$ is positive it has unique causal stable model $I_{\ref{prog:exp}}$.
Furthermore, it is easy to see that the interpretation of atoms $p_1$ and $q_1$ with respect to interpretation $I_{\ref{prog:exp}}$ are $a+b$ and $c+d$, respectively. The interpretation for $p_2$ corresponds to:
\vspace{-3pt}
\begin{IEEEeqnarray*}{l C l l}
I_{\ref{prog:exp}}(p_2) & = && (I(p_1)*I(q_1))\cdot m_2 = ((a+b)*(c+d))\cdot m_2\\
	& = &&	( a*c ) \cdot m_2 \ + \
			( a*d ) \cdot m_2 \ + \
			( b*c ) \cdot m_2 \ + \
			( b*d ) \cdot m_2
\end{IEEEeqnarray*}
This addition cannot be further simplified.
Analogously, $I_{\ref{prog:exp}}(q_2)$ can also be expressed as a sum of four sufficient causes -- we just replace $m_2$ by $n_2$ in $I(p_2)$. But then, $I_{\ref{prog:exp}}(p_3)$ corresponds to $(I_{\ref{prog:exp}}(p_2) * I_{\ref{prog:exp}}(q_2)) \cdot m_3$ and, applying distributivity, this yields a sum of $4 \times 4$ sufficient causes. In the general case, each atom $p_n$ or $q_n$ has $2^{2^{n-1}}$ sufficient causes so that expanding the complete causal value into this additive normal form becomes intractable.
Furthermore, it also has been in~\cite{CabalarFF14Jelia} that deciding whether a term without \review{R2.4}{addition}~$G$ is a brave necessary cause with respect to some regular program $P$ is $\SigmaP{2}$-complete and, thus, deciding the existence of causal stable model is $\SigmaP{2}$-hard even for the class of programs that only contain a unique necessary causal literal.

\begin{proposition}[From~\protect\citeNP{CabalarFF14Jelia}]\label{prop:complx.cause.leq.term}
Given a causal term without \review{R2.4}{addition}~$G \in \causes$ and an causal term $t \in \values$ in which the right-hand operand of every application~``$\cdot$'' is a label, deciding whether $G \leq t$ is feasible in polynomial time.\qed
\end{proposition}

\begin{proposition}\label{prop:complx.term.leq.term}
Let $\set{t,u} \subseteq \values$ be two causal term in which the right-hand operand of every application~``$\cdot$'' is a label.
Then deciding whether $t \leq u$ is in $\coNP$.\qed
\end{proposition}

\begin{proof}
Note hat $t \leq u$ iff every~$G\in\causes$ such that $G \leq t$ also satisfy $G \leq u$ which are decidable in polynomial time (Proposition~\ref{prop:complx.cause.leq.term}).
Consequently, deciding whether $t \leq u$ is $\coNP$.
\end{proof}

\begin{definition}[Causal graph]\label{def:causal.graph}
Given a set of labels $\lb$, a \emph{causal graph (c-graph)} $G \subseteq \lb \times \lb$ is a set of edges transitively and reflexively closed.
By $\graphs$ we denote the set of all c-graphs that can be formed with labels from $\lb$.\QED
\end{definition}

\begin{theorem}[From~\protect\cite{fandinno2015aspocp}]\label{thm:tp.properties.finite}
For any finite and definite program $P$ with $n$ rules, $\lfp(T_P)\!=\!\tpr{n}$ is its least model.\QED
\end{theorem}



\begin{definition}
Let $P$ be a program and $I$ be an interpretation.
By $\mathtt{simply-nec}(P^I)$ we denote the program obtained from $P^I$ by replacing every  causal literal of the form~$(\cliteral{\cquerynec_\ag}{\rA})^{I(\rA)}$
by $\rA$ if $I(\rA) \leq \sum\ag$; and by $0$ otherwise.\qed
\end{definition}


\begin{lemma}\label{lem:aux1:prop:complx.membership.necessary.cliterals}
Let $P$ be a program and $I$ be an interpretation.
If $\tprP{P^I}{\alpha} \leq \tprP{Q}{\alpha} \leq I$,
then
$\tprP{P^I}{\alpha+1} \leq \tprP{Q}{\alpha+1} \leq I$
where $Q=\mathtt{simply-nec}(P^I)$.\qed
\end{lemma}

\begin{proof}
Suppose first that $\tprP{P^I}{\alpha+1}(\rA) \not\leq \tprP{Q}{\alpha+1}(\rA)$ for some atom~$\rA$.
Then, since application and \review{R2.4}{addition} are $\leq$-monotonic, there must be some rule of the form of~\eqref{eq:rule.nested} such that
$\tprP{P^I}{\alpha}(\fF) \not\leq \tprP{Q}{\alpha}(\fF')$
where $\fF'$ is just the result of replacing each causal literal of the form~$(\cliteral{\cquerynec_\ag}{\rA})^{I(\rA)}$
by $\rA$ if $I(\rA) \leq \sum\ag$; and by $0$ otherwise.
Since products and \review{R2.4}{addition} are monotonic, it is enough to show that
\begin{itemize}
\item $\tprP{P^I}{\alpha}(\cliteral{\cquery'}{\rA}) \ \ \leq \ \ \tprP{Q}{\alpha}(\rA)$ if $I(\rA) \leq \sum\ag$, and
\item $\tprP{P^I}{\alpha}(\cliteral{\cquery'}{\rA}) \ \ = \ \ 0$ otherwise.
\end{itemize}
where
\begin{gather*}
\cquery'(G,I(\rA)) \ \ \eqdef \ \
  \begin{cases}
  1 &\text{iff exists some }  G' \leq G \text{ s.t. } G' \leqmax I(\rA)
      \text{ and }
      I(\rA) \leq \sum\ag
  \\
  0 &\text{otherwise}
  \end{cases}
\end{gather*}
By definition, 
\begin{align*}
\tprP{P^I}{\alpha}(\cliteral{\cquerynec_\ag}{\rA})^{I(\rA)} \ \ &\eqdef \ \
  \sum\setbm{G \leqmax \tprP{P^I}{\alpha}(\rA)}{ \cquery'(G,\, I(\rA)\,) = 1  }
\end{align*}
One the one hand,
$\tprP{P^I}{\alpha}(\cliteral{\cquery'}{\rA}) \leq \tprP{P^I}{\alpha}(\rA)$ holds for every causal literal~$(\cliteral{\cquery'}{\rA})$ and, by hypothesis,
it holds that
$\tprP{P^I}{\alpha}(\rA) \leq \tprP{Q}{\alpha}(\rA)$
and, therefore, 
$\tprP{P^I}{\alpha}(\cliteral{\cquery'}{\rA}) \leq \tprP{Q}{\alpha}(\rA)$ also holds.
On the other hand,
$I(\rA) \not\leq \sum\ag$
implies that
$\cquery'(G,\, I(\rA)\,) = 0$ for every $G \in \causes$
and, thus,
$\tprP{P^I}{\alpha+1}(\cliteral{\cquery'}{\rA}) = 0 \leq \tprP{Q}{\alpha+1}(\rA)$.

Similarly, to show that $\tprP{Q}{\alpha} \leq I$ is enough to show
$\tprP{P^I}{\alpha}(\rA) \leq I(\cliteral{\cquery'}{\rA})$ when $I(\rA) \leq \sum\ag$.
Note that, in case that $I(\rA) \not\leq \sum\ag$, the causal literal~$(\cliteral{\cquery'}{\rA})$ has been replaced by~$0$.
Then,
for every $G \leq \tprP{Q}{\alpha}(\cliteral{\cquery'}{\rA})$
there is some $G' \leqmax I(\rA)$ and $\cquery'(G',I(\rA)) = 1$ and, consequently,
it follows that
$G \leq G' \leq I(\cliteral{\cquery'}{\rA})$.

Furthermore, it is easy to see that
$\tprP{Q}{\alpha} \leq I$
implies
$\tprP{Q}{\alpha}(\rA) \leq  I(\cliteral{\cquery'}{\rA})$
for every causal literal $(\cliteral{\cquery'}{\rA})$.
Just note that
if $G\leq\tprP{Q}{\alpha}(\cliteral{\cquery'}{\rA})$,
then $G\leq\tprP{Q}{\alpha}(\rA) \leq I(\rA)$
and there is some $G' \leq G$ such that $G' \leqmax I(\rA)$ and $I(\rA) \leq \sum\ag$.
Notice that facts $G \leq I(\rA)$, $G' \leqmax I(\rA)$ and $G' \leq G$ implies that
$G = G'$ and, thus, $G \leqmax I(\rA)$.
Therefore, $G \leq I(\cliteral{\cquery'}{\rA})$ and, thus,
\begin{gather*}
\tprP{Q}{\alpha}(\cliteral{\cquery'}{\rA}) 
	\ \ \leq \ \ I(\cliteral{\cquery'}{\rA})
\end{gather*}

Note now that the evaluation of conjunctions and disjunctions is $\leq$-monotonic and, thus, it can be probed by induction that
\begin{gather*}
\tprP{P^I}{\alpha}(\fF)
	\ \ \leq \ \ \tprP{Q}{\alpha}(\fF)
	\ \ \leq \ \ I(\fF)
\end{gather*}
for every formula~$\fF$.
Finally, since \review{R2.4}{addition} and application are also $\leq$-monotonic, it can be shown by induction that
\begin{gather*}
\tprP{P^I}{\alpha}(\fF) \cdotl r_i  
	\ \ \leq \ \ \tprP{Q}{\alpha}(\fF) \cdotl r_i
	\ \ \leq \ \ I(\fF) \cdotl r_i
\end{gather*}
and, thus,
\begin{gather*}
\tprP{P^I}{\alpha+1}(\rA)
	\ \ \leq \ \ \tprP{Q}{\alpha+1}(\rA)
	\ \ \leq \ \ I(\rA)
\end{gather*}
for every label $r_i \in \lb$ and atom $\rA \in \at$.
\end{proof}

\begin{lemma}\label{lem:aux1b:prop:complx.membership.necessary.cliterals}
Let $P$ be a program and $I$ be an interpretation.
Then
$\tprP{P^I}{\omega} \leq \tprP{Q}{\omega} \leq I$
where $Q=\mathtt{simply-nec}(P^I)$.\qed
\end{lemma}

\begin{proof}
By definition,
$\tprP{P^I}{0} = \tprP{Q}{0} = \botI \leq I$
and, thus, by induction using Lemma~\ref{lem:aux1:prop:complx.membership.necessary.cliterals},
it follows that
$\tprP{P^I}{\alpha} \leq \tprP{Q}{\alpha} \leq I$
for every successor ordinal~$\alpha$.
For a limit ordinal $\alpha$,
$G \leq \tprP{P^I}{\alpha}$ iff there is some $\beta < \alpha$ s.t.
$G \leq \tprP{P^I}{\beta} \leq \tprP{Q}{\beta} \leq \tprP{Q}{\alpha}$
and, thus,
$\tprP{P^I}{\alpha} \leq \tprP{Q}{\alpha}$.
The proof of
$\tprP{Q}{\alpha} \leq I$ is analogous.
Hence, 
$\tprP{P^I}{\omega} \leq \tprP{Q}{\omega} \leq I$.
\end{proof}

\begin{lemma}\label{lem:aux2:prop:complx.membership.necessary.cliterals}
Let $P$ be a program and $I$ be an interpretation
and $Q=\mathtt{simply-nec}(P^I)$.
If for every atom~$\rA$ and causal term without \review{R2.4}{addition}
$G \leq \tprP{Q}{\alpha}(\rA)$ such that $G \leqmax I(\rA)$,
it holds that
$G \leq \tprP{P^I}{\alpha}(\rA)$,
then for every atom~$\rA$ and causal term without \review{R2.4}{addition}
$G \leq \tprP{Q}{\alpha+1}(\rA)$ such that $G \leqmax I(\rA)$,
it holds that
$G \leq \tprP{P^I}{\alpha+1}(\rA)$\qed
\end{lemma}

\begin{proof}
Suppose there is some atom~$\rA$ and causal term without \review{R2.4}{addition} $G \in \causes$
such that $G \leq \tprP{Q}{\alpha+1}(\rA)$ and $G \leqmax I(\rA)$,
but $G\not\leq \tprP{Q}{\alpha+1}(\rA)$.
Then, since application and \review{R2.4}{addition} are $\leq$-monotonic, there must be some causal term without \review{R2.4}{addition} $G' \in \causes$ and rule of the form of~\eqref{eq:rule.nested} such that
$G \leq G' \cdotl r_i$
and
$G' \leq \tprP{Q}{\alpha}(\fF')$, but
$G' \not\leq \tprP{P^I}{\alpha}(\fF)$
where $\fF'$ is just the result of replacing each causal literal of the form~$(\cliteral{\cquerynec_\ag}{\rA})^{I(\rA)}$
by $\rA$ if $I(\rA) \leq \sum\ag$; and by $0$ otherwise.
Since products and \review{R2.4}{addition} are monotonic and every causal literal in $\fF$ of the form of~$(\cliteral{\cquery'}{\rA})$ is replaced by $0$ in $\fF'$ when $I(\rA) \not\leq \sum\ag$, it is enough to show that
\begin{itemize}
\item $G \leq \tprP{Q}{\alpha}(\rA)$ and $G\in\leqmax I(\rA)$ implies $G \leq \tprP{P^I}{\alpha}(\cliteral{\cquery'}{\rA})$ when $I(\rA) \leq \sum\ag$
\end{itemize}
where $\cquery' \eqdef (\cquerynec_\ag)^{I(\rA)}$.
Indeed, by hypothesis, from $G \leq \tprP{Q}{\alpha}(\rA)$ and $G\leqmax I(\rA)$ it follows that
$G \leq \tprP{P^I}{\alpha}(\rA)$.
Furthermore, $G \leqmax I(\rA)$ and $I(\rA) \leq \sum\ag$ also imply that $(\cquery'(G,I(\rA))=1$ holds and, consequently, it follows that
$G \leq \tprP{P^I}{\alpha}(\cliteral{\cquery'}{\rA})$.
\end{proof}

\begin{lemma}\label{lem:aux3:prop:complx.membership.necessary.cliterals}
Let $P$ be a program and $I$ be an interpretation and $Q=\mathtt{simply-nec}(P^I)$.
Then, $I=\tprP{Q}{\omega}$ iff $I=\tprP{P^I}{\omega}$\qed
\end{lemma}

\begin{proof}
First, assume that
$I=\tprP{P^I}{\omega}$.
From Lemma~\ref{lem:aux1b:prop:complx.membership.necessary.cliterals},
it follows that $\tprP{P^I}{\omega} \leq \tprP{Q}{\omega} \leq I$ and, thus,
$I=\tprP{P^I}{\omega}$ implies $I=\tprP{Q}{\omega}$.

The other way around.
Assume $I=\tprP{Q}{\omega}$ and assume as induction hypothesis that
$G \leq \tprP{Q}{\beta}(\rA)$ and $G \leqmax I(\rA)$,
imply
$G \leq \tprP{P^I}{\beta}(\rA)$
for every ordinal $\beta < \alpha$.
By definition,
$\tprP{Q}{0} = \botI$ and the hypothesis holds vacuous.
Furthermore, using Lemma~\ref{lem:aux2:prop:complx.membership.necessary.cliterals},
the hypothesis holds for every successor ordinal $\alpha$.
For a limit ordinal $\alpha$,
$G \leq \tprP{Q}{\alpha}$ iff there is some $\beta < \alpha$ s.t.
$G \leq \tprP{Q}{\beta}$
and, thus,
$G \leq \tprP{P^I}{\beta} \leq \tprP{P^I}{\alpha}$.
Then, for every $G \leq I(\rA) = \tprP{Q}{\omega}$
there is some $G' \in \causes$ such that
$G' \leqmax I(\rA) = \tprP{Q}{\omega}$
and, thus,
$G \leq G' \leq \tprP{P^I}{\omega}(\rA)$.
That is, $\tprP{Q}{\omega} \leq \tprP{P^I}{\omega}$.
Finally, from Lemma~\ref{lem:aux1b:prop:complx.membership.necessary.cliterals},
it follows that $\tprP{P^I}{\omega} \leq \tprP{Q}{\omega}$ and, thus,
it also holds $I=\tprP{P^I}{\omega}$.
\end{proof}

\begin{lemma}\label{lem:aux4:prop:complx.membership.necessary.cliterals}
Let $P$ be a program and $I$ be an interpretation and $Q=\mathtt{simply-nec}(P^I)$.
Then, $I$ is a causal stable model of $P$ iff
$\tprP{Q}{\omega} = I$.\qed
\end{lemma}

\begin{proof}
By definition, $I$ is a causal stable model of $P$ iff $I$ is the least model of $P^I$ iff $\tprP{P^I}{\omega} = I$ (Theorem~\ref{thm:tp.properties.nested})
iff $\tprP{P^I}{\omega} = I$ (Lemma~\ref{lem:aux3:prop:complx.membership.necessary.cliterals}).
\end{proof}

\begin{proposition}\label{prop:complx.term.leq.sum.of.labels}
Let $t$ be a term and $\ag$ be a set of labels.
Then, $t \leq \sum\ag$ is decidable in polynomial time.\qed
\end{proposition}

\begin{proof}
If $t \in \lb$ is a label, then $t \leq \sum\ag$ iff $t \in \ag$ which is clearly decidable in polynomial time.
Otherwise, we assume as induction hypothesis that $u  \leq \sum\ag$ and $w \leq \sum\ag$ are decidable in polynomial time for every subterms $u$ and $w$ of $t$.
In case that $t = u + w$, then $t \leq \sum\ag$ iff $u \leq \sum\ag$ and $w \leq \sum\ag$ which are both decidable in polynomial time.
Similarly, in case that $t = u * w$ or $t =u \cdot w$, then $t \leq \sum\ag$ iff $u \leq \sum\ag$ or $w \leq \sum\ag$ which are both decidable in polynomial time.
\end{proof}

\begin{proposition}\label{prop:complx.membership.necessary.cliterals}
Let $P$ be a causal program containing only necessary causal literals.
Then, deciding whether there exists a causal stable model of $P$ or not is in $\NP$.\qed
\end{proposition}

\begin{proof}
First, note that there exists some causal stable model $I$ of $P$ iff there must exists some program $Q$ and casual stable model $I$ such that $Q$ is the result of replacing every maximal subformula in $P$ of the form $\Not \fE$ by $1$ if $I\models\Not\rE$ and by $0$ otherwise.
Just note that $Q^I = P^I$ for every interpretation $I$.
Then, from Lemma~\ref{lem:aux4:prop:complx.membership.necessary.cliterals},
$I$ is a causal stable model of $P$
iff $I$ is a causal stable model of $Q$
iff $\tprP{Q'}{\omega}=I$ where $Q'= \mathtt{simply-nec}(P^I)$.

Hence, instead of guessing an interpretation $I$ we will guess a program $Q'$.
Let $Q$ for every maximal subformula  be the result of replacing every maximal subformula in $P$ of the form $\Not \fE$ by a guessed $0$ or $1$ 
and let $Q'$  be he result of replacing every necessary causal literal in $Q$ of the form of~$(\cliteral{\cquerynec_\ag}{\rA})$ by a guessed $0$ or $\rA$.
Note that, since $P$ only contains necessary causal literals, $Q'$ is a positive regular (hence monotonic) program.
From Theorem~\ref{thm:tp.properties.finite}, it follows that $\tprP{Q'}{n}=\tprP{Q'}{\omega}$ is the least fixpoint of $\tpP{Q'}$ and the least model of $Q'$ where $n$ is the number of rules in $Q$, which is the same as the number of rules in $P$.
Let us define $I = \tprP{Q'}{n}$.
Since $Q'$ is a regular program each step of $\tpP{Q'}$ only involves the creation of a term from its subterms, which is feasible in polynomial time and, thus, $I$ can be computable in polynomial time.

Let us now check whether $Q' = \mathtt{simply-nec}(P^I)$.
Then, $\mathtt{fail}$ if $I=\tprP{Q'}{n}$ do not satisfy one of the following conditions
\begin{itemize}
\item $I \models \Not \rE$ for some maximal subformula  whose guessed value was $0$
\item $I \not\models \Not \rE$ for some maximal subformula whose guessed value was $1$
\item $I(\rA) \leq \sum\ag$ for some necessary causal literal~$(\cliteral{\cquerynec_\ag}{\rA})$ whose guessed value was $0$
\item $I(\rA) \not\leq \sum\ag$ for some necessary causal literal~$(\cliteral{\cquerynec_\ag}{\rA})$ whose guessed value was $\rA$
\end{itemize}
If reached this point, then $Q' = \mathtt{simply-nec}(P^I)$ and, hence, we the procedure $\mathtt{succeed}$.
It just remain to show that these four conditions can be checked in polynomial time.
The two first only involve checking whether $I(\fE)=0$ which is feasible simply simplifying the obtained causal term and looking whether it is $0$ or not.
Finally, since $\ag\subseteq\lb$ is a set of labels, from Proposition~\ref{prop:complx.term.leq.sum.of.labels}, it follows that
$I(\rA) \leq \sum\ag$ can be checked in polynomial time.
\end{proof}

\begin{proposition}\label{prop:complx.complete.necessary.cliterals}
Let $P$ be a causal program containing only necessary causal literals.
Then, deciding whether there exists a causal stable model of $P$ or not is in $\NP$-complete (it is $\NP$-hard even in $P$ only contains a single negated regular literal or $P$ is positive but contains a single constraint).\qed
\end{proposition}

\begin{proof}
$\NP$ membership follows directly from Proposition~\ref{prop:complx.membership.necessary.cliterals} while $\NP$-hard follows from the fact that every regular program in also a causal program and deciding the existence of stable model for standard programs in $\NP$-complete.
To show that it is $\NP$-hard even when $P$ only contains a negated regular literal, we reduce the existence of stable model for standard program to the satisfiability of a CNF Boolean formula $\varphi$ to the existence of causal stable model of a program $P$.
We assume without of generality that no clause in $\varphi$ has complementary variables.
For every variable $x_k$ occurring in $\varphi$, let $P_k$ be a program containing rules of the form
\begin{IEEEeqnarray*}{l C ? l ; C ; l}
x_k &:& x_k &\lparrow& 
\\
t_{x_k} &:& x_k &\lparrow& \ag_{tk} \necessary x_k
\\
f_{x_k} &:& x_k &\lparrow& \ag_{fk} \necessary x_k
\end{IEEEeqnarray*}
 where $\ag_{tk}=\set{t_{x_k}, x_k}$ 
and
$\ag_{fk}=\set{f_{x_k}, x_k}$.
For each clause $c_j$ in $\varphi$,
let $P_{j}'$ be a program containing a rule of the form of
\begin{gather}
c_j \lparrow \ag_{jk} \necessary x_k
\end{gather}
for each variable $x_k$ in $c_j$,
where $\ag_{jk}=\set{t_{x_k}, x_k}$ if $x_k$ occurs positively in the clause $c_j$
and
$\ag_{jk}=\set{f_{x_k}, x_k}$ if $x_k$ occurs negatively in $c_j$
and let $P''$ be a program containing the following rule
\begin{IEEEeqnarray*}{l C ? l ; C ; l}
&& p &\lparrow& c_1, \dotsc, c_\npbody
\end{IEEEeqnarray*}
where $c_1,\dotsc,c_\npbody$ are all the clauses in $\varphi$.
Note that no atom occurring in the body of a program $P_k$ occurs in the head of a program $P'_j$ nor $P''$ and no atom occurring in the body of a program $P_j'$ occurs in the head of a program $P''$.
Hence, we can use the Splitting Theorem (Theorem~\ref{thm:splitting.nested}).

Each program $P_k$ has two causal stable models $I_k$ and $J_k$ that satisfy $I_k(x_k) = x_k + t_{x_k}$ and $J_k(x_k) = x_k + f_{x_k}$ and, thus, $P=\bigcup_{k \leq n} P_n$ has $2^n$ causal stable models here $n$ is the number of variables in $\varphi$: each causal stable model $I$
 satisfying $I(x_k) = x_k + t_{x_k}$ or $I(x_k) = x_k + f_{x_k}$ for each variable $x_k$.
We say that an variable $x_k$ is true in an interpretation $I$ if $I(x_k) = x_k + t_{x_k}$ and that is false if $I(x_k) = x_k + f_{x_k}$.
Then, $(P \cup P'_j)$ also has $2^n$ causal models models where
each causal stable model $I$ satisfy
\begin{align*}
I(c_j)
	&\ \ = \ \
	\sum\setm{ x_k + t_{x_k} }{ I(x_k) = x_k + t_{x_k} \text{ and } x_k \text{ occurs positively in } c_j }
\\
	&\ \ + \ \
	\sum\setm{ x_k + f_{x_k} }{ I(x_k) = x_k + t_{x_k} \text{ and } x_k \text{ occurs negatively in } c_j }
\end{align*}
That is, $I(c_j) \neq 0$ iff there is some variable $x_k$ such that $x_k$ is true in $I$ and occurs positively in $c_j$ or there is some variable $x_k$ such that $x_k$ is false in $I$ and occurs negatively in $c_j$
iff $I$ represents an assignment that satisfies the clause $c_j$.
Let $P' = \bigcup_{j \leq \npbody} P'_j$.
Then, $P'$ also has $2^n$ causal models models:
each causal stable model $I$ satisfy for each clause $c_j$ that
$I(c_j) \neq 0$ iff represents an assignment that satisfies the clause $c_j$.
It is easy to see now that $(P' \cup P'')$ has $2^n$ causal stable models where
$I(p) \neq 0$ iff every $c_j$ satisfy $I(c_j) \neq 0$ iff $I$ represents an assignment that satisfy all clauses $c_j$ in $\varphi$ iff $I$ represents an assignment that satisfy $\varphi$.
Finally, let $P$ be the result of adding, to the program~$(P' \cup P'')$, the following rule
\begin{gather*}
 p \lparrow \Not p
\end{gather*}
Then, $P$ has a causal stable model iff there is a causal stable model $I$ of $(P' \cup P'')$ such that $I(p) \neq 0$ iff $I$ represents an assignment that satisfy $\varphi$.
Alternatively, let $Q'$ be the result of adding, to the program~$(P' \cup P'')$, the following rules
\begin{IEEEeqnarray*}{l C ? l ; C ; l}
q &:& q &\lparrow&
\\
	&& q &\lparrow& p
\end{IEEEeqnarray*}
and $Q$ be the result of adding to $Q'$ the constraint
\begin{IEEEeqnarray*}{l C ? l ; C ; l}
	&& \bot &\lparrow& \ag_p \necessary q
\end{IEEEeqnarray*}
where $\ag_p = \set{ q }$.
Then, $Q'$ has $2^n$ causal stable models: each causal stable model $I$ satisfying
$I(q) = q + I(p)$ and $Q$ has a causal stable model $I$
iff there is a causal stable model $I$ of $Q'$ such that $I(q) > q$
iff there is a causal stable model $I$ of $Q'$ such that $I(p) \neq 0$
iff $I$ represents an assignment that satisfy $\varphi$.
\end{proof}








\newpage

\end{document}